%% file: GL_Parallel_WL.tex
\documentclass{article}
\usepackage[left=1in, right=1in, top=1in, bottom=1in]{geometry}
\usepackage{enumitem}
\usepackage{authblk}
\usepackage{graphicx}	
\usepackage{amsmath, amsthm, amssymb}
\usepackage{xspace}
\usepackage{comment}
\usepackage{hyperref}
\usepackage{xcolor}
\usepackage{algorithm}
\usepackage{algpseudocode}
\usepackage{lineno}

\newcommand{\cc}[1]{\ensuremath{\mathsf{#1}}}
\newcommand{\algprobm}[1]{\textsc{#1}\xspace}

\newcommand{\Soc}{\text{Soc}}
\newcommand{\Fac}{\text{Fac}}
\newcommand{\F}{\mathbb{F}}
\newcommand{\Z}{\mathbb{Z}}

\theoremstyle{plain}
\newtheorem{theorem}{Theorem}[section]
\newtheorem{proposition}[theorem]{Proposition}
\newtheorem{corollary}[theorem]{Corollary}
\newtheorem{lemma}[theorem]{Lemma}
\newtheorem{fact}[theorem]{Fact}

\theoremstyle{definition}
\newtheorem{definition}[theorem]{Definition}
\newtheorem{remark}[theorem]{Remark}
\newtheorem{question}[theorem]{Question}

\newcommand{\Lem}[1]{Lem.~\ref{#1}\xspace}
\newcommand{\Cor}[1]{Cor.~\ref{#1}\xspace}
\newcommand{\Prop}[1]{Prop.~\ref{#1}\xspace}
\newcommand{\Thm}[1]{Thm.~\ref{#1}\xspace}

\DeclareMathOperator{\Inn}{Inn}
\DeclareMathOperator{\Aut}{Aut}
\DeclareMathOperator{\ncl}{ncl}

\DeclareMathOperator{\poly}{poly}

\DeclareMathOperator{\cw}{cw}
\DeclareMathOperator{\rk}{rk}

\title{On the Parallel Complexity of Group Isomorphism via Weisfeiler--Leman\footnote{A preliminary version of this work appeared in the proceedings of FCT 2023 \cite{GrochowLevetWL}. ML thanks Keith Kearnes for helpful discussions, which led to a better understanding of the Hella-style pebble game. ML also wishes to thank Richard Lipton for helpful discussions regarding previous results. We wish to thank J. Brachter, P. Schweitzer, and the anonymous referees for helpful feedback. JAG was partially supported by NSF award DMS-1750319 and NSF CAREER award CCF-2047756 and during this work. ML was partially supported by J. Grochow startup funds.}}

\author[1,2]{Joshua A. Grochow}
\author[3]{Michael Levet}
\affil[1]{Department of Computer Science, University of Colorado Boulder}
\affil[2]{Department of Mathematics, University of Colorado Boulder}
\affil[3]{Department of Computer Science, College of Charleston}

\begin{document}
\maketitle

\begin{abstract}
In this paper, we show that the constant-dimensional Weisfeiler--Leman algorithm for groups (Brachter \& Schweitzer, LICS 2020) can be fruitfully used to improve parallel complexity upper bounds on isomorphism testing for several families of groups. In particular, we show:
\begin{itemize}
\item Groups with an Abelian normal Hall subgroup whose complement is $O(1)$-generated are identified by constant-dimensional Weisfeiler--Leman using only a constant number of rounds. This places isomorphism testing for this family of groups into $\textsf{L}$; the previous upper bound for isomorphism testing was $\cc{P}$ (Qiao, Sarma, \& Tang, STACS 2011).

\item We use the individualize-and-refine paradigm to obtain an isomorphism test for groups without Abelian normal subgroups by $\cc{SAC}$ circuits of depth $O(\log n)$ and size $n^{O(\log \log n)}$, 
previously only known to be in $\cc{P}$ (Babai, Codenotti, \& Qiao, ICALP 2012) and $\mathsf{quasiSAC}^1$ (Chattopadhyay, Tor\'an, \& Wagner, \textit{ACM Trans. Comput. Theory}, 2013).

\item We extend a result of Brachter \& Schweitzer (ESA, 2022) on direct products of groups to the parallel setting. Namely, we also show that Weisfeiler--Leman can identify direct products in parallel, provided it can identify each of the indecomposable direct factors in parallel. They previously showed the analogous result for $\cc{P}$.
\end{itemize}

We finally consider the count-free Weisfeiler--Leman algorithm, where we show that count-free WL is unable to even distinguish Abelian groups in polynomial-time. Nonetheless, we use count-free WL in tandem with bounded non-determinism and limited counting to obtain a new upper bound of $\beta_{1}\textsf{MAC}^{0}(\textsf{FOLL})$ for isomorphism testing of Abelian groups. This improves upon the previous $\textsf{TC}^{0}(\textsf{FOLL})$ upper bound due to Chattopadhyay, Tor\'an, \& Wagner (\textit{ibid.}). 
\end{abstract}

\thispagestyle{empty}

\newpage

\setcounter{page}{1}

\section{Introduction}
\label{sec:introduction}
The \algprobm{Group Isomorphism} problem (\algprobm{GpI}) takes as input two finite groups $G$ and $H$, and asks if there exists an isomorphism $\varphi : G \to H$. When the groups are given by their multiplication (a.k.a. Cayley) tables, it is known that $\algprobm{GpI}$ belongs to $\textsf{NP} \cap \textsf{coAM}$. The generator-enumerator algorithm, attributed to Tarjan in 1978 \cite{MillerTarjan}, has time complexity $n^{\log_{p}(n) + O(1)}$, where $n$ is the order of the group and $p$ is the smallest prime dividing $n$. In more than 40 years, this bound has escaped largely unscathed: Rosenbaum \cite{Rosenbaum2013BidirectionalCD} (see \cite[Sec. 2.2]{GR16}) improved this to $n^{(1/4)\log_p(n) + O(1)}$. And even the impressive body of work on practical algorithms for this problem, led by Eick, Holt, Leedham-Green and O'Brien (e.\,g., \cite{BEO02, ELGO02, BE99, CH03}) still results in an $n^{\Theta(\log n)}$-time algorithm in the general case (see \cite[Page 2]{WilsonSubgroupProfiles}). In the past several years, there have been significant advances on algorithms with worst-case guarantees on the serial runtime for special cases of this problem including Abelian groups \cite{Kavitha, Vikas, Savage}, direct product decompositions \cite{WilsonDirectProductsArxiv, KayalNezhmetdinov}, groups with no Abelian normal subgroups \cite{BCGQ, BCQ}, coprime and tame group extensions \cite{Gal09,QST11,BQ, GQ15}, low-genus $p$-groups and their quotients \cite{LW12,BMWGenus2}, Hamiltonian groups \cite{DasSharma}, and groups of almost all orders \cite{DietrichWilson}.

In addition to the intrinsic interest of this natural problem, a key motivation for the $\algprobm{Group Isomorphism}$ problem is its close relation to the \algprobm{Graph Isomorphism} problem ($\algprobm{GI}$). In the Cayley (verbose) model, $\algprobm{GpI}$ reduces to $\algprobm{GI}$ \cite{ZKT}, 
while $\algprobm{GI}$ reduces to the succinct $\algprobm{GpI}$ problem \cite{Heineken1974TheOO, Mekler} (recently simplified \cite{HeQiao}). In light of Babai's breakthrough result that $\algprobm{GI}$ is quasipolynomial-time solvable \cite{BabaiGraphIso}, $\algprobm{GpI}$ in the Cayley model is a key barrier to improving the complexity of $\algprobm{GI}$. Both verbose $\algprobm{GpI}$ and $\algprobm{GI}$ are considered to be candidate $\textsf{NP}$-intermediate problems, that is, problems that belong to $\textsf{NP}$, but are neither in $\textsf{P}$ nor $\textsf{NP}$-complete \cite{Ladner}. There is considerable evidence suggesting that $\algprobm{GI}$ is not $\textsf{NP}$-complete \cite{Schoning, BuhrmanHomer, ETH, BabaiGraphIso, GILowPP, ArvindKurur}. As verbose $\algprobm{GpI}$ reduces to $\algprobm{GI}$, this evidence also suggests that $\algprobm{GpI}$ is not $\textsf{NP}$-complete. It is also known that $\algprobm{GI}$ is strictly harder than $\algprobm{GpI}$ under $\textsf{AC}^{0}$ reductions \cite{ChattopadhyayToranWagner}. Tor\'an showed that $\algprobm{GI}$ is $\textsf{DET}$-hard \cite{Toran}, which provides that $\algprobm{Parity}$ is $\textsf{AC}^{0}$-reducible to $\algprobm{GI}$. On the other hand, Chattopadhyay, Tor\'an, and Wagner showed that $\algprobm{Parity}$ is not $\textsf{AC}^{0}$-reducible to $\algprobm{GpI}$ \cite{ChattopadhyayToranWagner}. To the best of our knowledge, there is no literature on lower bounds for $\algprobm{GpI}$ in the Cayley table model. The absence of such lower bounds begs the question of how much existing polynomial-time isomorphism tests can be parallelized, even for special cases for \algprobm{GpI}.

Despite $\algprobm{GpI}$ in the Cayley table model being strictly easier than $\algprobm{GI}$ under $\textsf{AC}^{0}$-reductions, there are several key approaches in the $\algprobm{GI}$ literature such as parallelization and individualization that have received comparatively little attention in the setting of \algprobm{GpI}---see the discussion of Related Work on Page 5. In this paper, using Weisfeiler--Leman for groups \cite{WLGroups} as our main tool, we begin to bring both of these techniques to bear on $\algprobm{GpI}$. As a consequence, we also make advances in the descriptive complexity theory of finite groups.

\noindent \\ \textbf{Main Results.} In this paper, we show that Weisfeiler--Leman serves as a key subroutine in developing efficient parallel isomorphism tests. 

Brachter \& Schweitzer \cite{WLGroups} actually introduced three different versions of WL for groups. While they are equivalent in terms of pebble complexity up to constant factors, their round complexities---and hence, parallel complexities---may differ. Because of these differences we are careful to specify which version of WL for groups each result uses.

As we are interested in both the Weisfeiler--Leman dimension and the number of rounds, we introduce the following notation.

\begin{definition}
Let $k \geq 2$ and $r \geq 0$ be integers, and let $J \in \{ I, II \}$. The $(k, r)$-WL Version $J$ algorithm for groups is obtained by running $k$-WL Version $J$ for $r$ rounds. Here, the initial coloring counts as the first round ($r = 0$). By conventional definition, each of the WL versions distinguish two group at $r = 0$ if and only if the groups do not have the same order.
\end{definition}

We first examine coprime extensions of the form $H \ltimes N$ where $N$ is Abelian. When either $H$ is elementary Abelian or $H$ is $O(1)$-generated, Qiao, Sarma, \& Tang \cite{QST11} gave a polynomial-time isomorphism test for these families of groups, using some nontrivial representation theory. Here, as a proof of concept that WL can successfully use and parallelize some representation theory (which was not yet considered in \cite{WLGroups, BrachterSchweitzerWLLibrary}), we use WL to improve their result's parallel complexity in the case that $H$ is $O(1)$-generated. We remark below about the difficulties in extending WL to handle the case that $H$ is Abelian (without restricting the number of generators).

\begin{theorem} \label{ThmMainQST}
Groups of the form $H \ltimes N$, where $N$ is Abelian, $H$ is $O(1)$-generated, and $|H|$ and $|N|$ are coprime are identified by $(O(1), O(1))$-WL Version II. Consequently, isomorphism between a group of the above form and arbitrary groups can be decided in $\textsf{L}$.
\end{theorem}

\begin{remark}
Despite Qiao, Sarma, and Tang giving a polynomial-time algorithm for the case where $H$ and $N$ are coprime, $N$ is arbitrary Abelian, and $H$ is elementary Abelian (no restriction on number of generators for $H$ or $N$), we remark here on some of the difficulties we encountered in getting WL to extend beyond the case of $H$ being $O(1)$-generated. When $H$ is $O(1)$-generated, we may start by pebbling the generators of $H$. After this, by Taunt's Lemma (reproduced as Lemma~\ref{LemTaunt} below), all that is left is to identify the multiset of $H$-modules appearing in $N$. In contrast, when $H$ is not $O(1)$-generated, this strategy fails quite badly: if only a small subset of $H$'s generators are pebbled, then it leaves open automorphisms of $H$ that could translate one $H$-module structure to another. But the latter translation-under-automorphism problem is equivalent to the \emph{entire} problem in this family of groups (see, e.g., \cite[Theorem~1.2]{QST11}).

This same difficulty is encountered even when using the more powerful \emph{second} Ehrenfeucht--Fra\"iss\'e pebble game in Hella's \cite{Hella1989, Hella1993} hierarchy, in which Spoiler may pebble two elements per turn instead of just one. This second game in Hella's hierarchy is already quite powerful: it identifies semisimple groups using only $O(1)$ pebbles and $O(1)$ rounds \cite{GLDescriptiveComplexity}. It seems plausible to us that with only $O(1)$ pebbles, neither ordinary WL nor this second game in Hella's hierarchy identifies coprime extensions where both $H,N$ are Abelian with no restriction on the number of generators. 
\end{remark}

We next parallelize a result of Brachter \& Schweitzer \cite{BrachterSchweitzerWLLibrary}, who showed that Weisfeiler--Leman can identify direct products in polynomial-time provided it can also identify the indecomposable direct factors in polynomial-time. Specifically, we show:

\begin{theorem} \label{ThmProductCanonicalForm}
Let $G = G_{1} \times \cdots \times G_{d}$ be a decomposition into indecomposable direct factors, let $k \geq 5$, and let $r := r(n)$. If $G$ and $H$ are not distinguished by $(k, r+O(\log n))$-WL Version II, then there exist direct factors $H_{i} \leq H$ such that $H = H_{1} \times \cdots \times H_{d}$ such that for all $i \in [d]$, $G_{i}$ and $H_{i}$ are not distinguished by $(k-1, r)$-WL Version II.
\end{theorem}

Prior to \Thm{ThmProductCanonicalForm}, the best-known upper bound on computing direct product decompositions was $\textsf{P}$ \cite{WilsonDirectProductsArxiv, KayalNezhmetdinov}. In group isomorphism testing, the main benefit of these results is that, for the purposes of getting polynomial-time isomorphism tests, they reduce the problem to the case of directly indecomposable groups. Although we do not know how to show that Weisfeiler--Leman can be used to compute the direct product decomposition itself, \Thm{ThmProductCanonicalForm} shows that for the purposes of using WL to distinguish groups, WL achieves this same main benefit---reducing to the directly indecomposable case---in $\textsf{TC}^1$ (at the cost of increasing the dimension by 1 and adding $O(\log n)$ rounds).

We next consider groups without Abelian normal subgroups, one of the few other classes of groups of structural importance for which a highly non-trivial polynomial-time algorithm is known. Using the individualize-and-refine paradigm, we obtain a new upper bound for not only deciding isomorphisms, but also listing isomorphisms. To state this result, recall that an $\textsf{SAC}$ circuit is a Boolean circuit using $\textsf{AND}, \textsf{OR}$, and $\textsf{NOT}$ gates, in which the $\textsf{AND}$ gates have fan-in $2$, and the $\textsf{OR}$ gates have unbounded fan-in (see Section~\ref{sec:Complexity} for more details). We show that isomorphisms of groups without Abelian normal subgroups can be listed by an $\mathsf{SAC}$ circuit of $O(\log n)$ depth and $n^{O(\log \log n)}$ size. Although this does not improve upon the upper bound of $\textsf{P}$ for isomorphism testing in this class of groups \cite{BCQ}, this does parallelize the previous bound of $n^{\Theta(\log \log n)}$ runtime for listing isomorphisms \cite{BCGQ}. We note that for listing isomorphisms, our size bound is essentially optimal, as Babai \emph{et al.} (\emph{ibid.}) observed that such groups can have $n^{\Theta(\log \log n)}$ many isomorphisms. Compared to what is known for general groups, general group isomorphism can be solved by $O(\log n)$-depth $n^{O(\log n)}$-size $\mathsf{SAC}$ circuits \cite{ChattopadhyayToranWagner}, and for groups without Abelian normal subgroups our result improves the size much closer to polynomial while maintaining the depth.

\begin{theorem} \label{ThmSemisimpleIntro}
Let $G$ be a group of order $n$ without Abelian normal subgroups, and let $H$ be arbitrary group of order $n$. We can test isomorphism between $G$ and $H$ using an $\textsf{SAC}$ circuit of depth $O(\log n)$ and size $n^{\Theta(\log \log n)}$. Furthermore, all such isomorphisms can be listed in this bound.
\end{theorem}

\begin{remark}
The key idea in proving \Thm{ThmSemisimpleIntro} is to prescribe an isomorphism between $\Soc(G)$ and $\Soc(H)$ (as in \cite{BCGQ}), and then use Weisfeiler--Leman to test in $\textsf{L}$ whether the given isomorphism of $\Soc(G) \cong \Soc(H)$ extends to an isomorphism of $G \cong H$. The procedure from \cite{BCGQ} for choosing all possible isomorphisms between socles is easily seen to parallelize; our key improvement is in the parallel complexity of testing whether such an isomorphism of socles extends to the whole groups.

Previously, this latter step was shown to be polynomial-time computable \cite[Proposition~3.1]{BCGQ} via membership checking in the setting of permutation groups. Now, although membership checking in permutation groups is in $\cc{NC}$ \cite{BabaiLuksSeress}, the proof there uses several different group-theoretic techniques, and relies on the Classification of Finite Simple Groups (see the end of the introduction of \cite{BabaiLuksSeress} for a discussion). Furthermore, there is no explicit upper bound on which level of the $\cc{NC}$ hierarchy these problems are in, just that it is $O(1)$. Thus, it does not appear that membership testing in the setting of permutation groups is known to be even $\textsf{AC}^{1}$-computable. So already, our use of $\mathsf{SAC}$ circuits
 is new (the $n^{\Theta(\log \log n)}$ size comes only from parallelizing the first step). Furthermore, Weisfeiler--Leman provides a much simpler algorithm; indeed, although we also rely on the fact that all finite simple groups are 2-generated (a result only known via CFSG), this is the only consequence of CFSG that we use. We note, however, that although WL improves the parallel complexity of these particular instances of membership testing, it requires access to the multiplication table for the underlying group, so this technique cannot be leveraged for more general membership testing in permutation groups.
\end{remark}

In the case of serial complexity, if the number of simple direct factors of Soc(G) is just slightly less than maximal, even listing isomorphism can be done in $\cc{FP}$ \cite{BCGQ}. Under the same restriction, we get an improvement in the parallel complexity to $\mathsf{FL}$:

\begin{corollary}[Cf. {\cite[Corollary~4.4]{BCGQ}}]
Let $G$ be a group without Abelian normal subgroups, and let $H$ be arbitrary. Suppose that the number of non-Abelian simple direct factors of $\Soc(G)$ is $O(\log n/\log \log n)$. Then we can decide isomorphism between $G$ and $H$ in $\textsf{L}$, as well as list all such isomorphisms, in $\textsf{FL}$.
\end{corollary}

It remains open as to whether isomorphism testing of groups without Abelian normal subgroups is even in $\textsf{NC}$. This would follow if such groups were identified by WL with $O(1)$ pebbles in $(\log n)^{O(1)}$ rounds; while we do not yet know whether this is the case, there is a higher-arity version of WL which identifies such groups with $O(1)$ pebbles in $O(1)$ rounds \cite{GLDescriptiveComplexity}, but that higher-arity version is not known to be implementable by efficient algorithms.

Given the lack of lower bounds on \algprobm{GpI}, and Grohe \& Verbitsky's parallel WL algorithm, it is natural to wonder whether our parallel bounds could be improved. One natural approach to this is via the \emph{count-free} WL algorithm, which compares the set rather than the multiset of colors at each iteration. We show unconditionally that this algorithm fails to serve as a polynomial-time isomorphism test even for Abelian groups. 

\begin{theorem} \label{thm:MainCountFree}
There exists an infinite family $(G_{n}, H_{n})_{n \geq 1}$ where $G_{n} \not \cong H_{n}$ are Abelian groups of the same order and count-free WL requires dimension at least $(1/3)\log_2 |G_n|$ to distinguish $G_n$ from $H_n$.
\end{theorem}

\begin{remark}
Even prior to \cite{CFI}, it was well-known that the count-free variant of Weisfeiler--Leman failed to place \algprobm{GI} into $\textsf{P}$ \cite{ImmermanLander1990}. In fact, count-free WL fails to distinguish almost all graphs \cite{FaginCountFree, IMMERMAN198276}, while two iterations of the standard counting $1$-WL almost surely assign a unique label to each vertex \cite{BabaiKucera, BabaiErdosSelkow}. In light of the equivalence between count-free WL and the logic $\textsf{FO}$ (first-order logic \textit{without} counting quantifiers), this rules out $\textsf{FO}$ as a viable logic to capture $\textsf{P}$ on unordered graphs. Finding such a logic is a central open problem in Descriptive Complexity Theory. On ordered structures such a logic was given by Immerman \cite{ImmermanPTime} and Vardi \cite{VardiPTime}.

\Thm{thm:MainCountFree} establishes the analogous result, ruling out $\textsf{FO}$ as a candidate logic to capture $\textsf{P}$ on unordered groups. This suggests that some counting may indeed be necessary to place $\algprobm{GpI}$ into $\textsf{P}$. As $\textsf{DET}$ is the best known lower bound for $\algprobm{GI}$ \cite{Toran}, counting is indeed necessary to place $\algprobm{GI}$ into $\textsf{P}$. There are no such lower bound known for \algprobm{GpI}. Furthermore, the work of \cite{ChattopadhyayToranWagner} shows that \algprobm{GpI} is not hard (under $\textsf{AC}^{0}$-reductions) for any complexity class that can compute \algprobm{Parity}, such as \textsf{DET}. Determining which families of groups can(not) be identified by count-free WL remains an intriguing open question.
\end{remark}

While count-free WL is not sufficiently powerful to compare the multiset of colors, it turns out that $O(\log \log n)$-rounds of count-free $O(1)$-WL Version I will distinguish two elements of different orders. Thus, the multiset of colors computed by the count-free $(O(1), O(\log \log n))$-WL Version I for non-isomorphic Abelian groups $G$ and $H$ will be different. We may use $O(\log n)$ non-deterministic bits to guess the color class where $G$ and $H$ have different multiplicities, and then an $\textsf{MAC}^{0}$ circuit to compare said color class. This yields the following.

\begin{theorem} \label{thm:AbelianGpI}
Abelian Group Non-Isomorphism is in $\beta_{1}\textsf{MAC}^{0}(\textsf{FOLL})$.
\end{theorem}

Here, $\beta_{1}\textsf{MAC}^{0}(\textsf{FOLL})$ denotes the composition of an $\textsf{FOLL}$ circuit with a $\beta_{1}\textsf{MAC}^{0}$ circuit. We refer to Section~\ref{sec:Complexity} for formal details of these complexity classes. Here, we will briefly sketch the algorithm in hopes of providing some intuition regarding the complexity bounds. We may easily check in $\textsf{AC}^{0}$ whether a group is Abelian. So suppose that $G, H$ are Abelian groups. We first compute the order of each element, for both groups by an $\mathsf{AC}$ circuit of depth $O(\log \log n)$, that is, in the class $\mathsf{FOLL}$ \cite{BKLM}. If $G \not \cong H$, then there exists $d \in \mathbb{N}$ where, without loss of generality, $G$ has more elements of order $d$ than $H$. We may use $O(\log n)$ non-deterministic bits to guess such a $d$ (this is the $\beta_{1}$), and a single $\textsf{Majority}$ gate to compare the number of elements of order $d$ in $G$ vs. $H$ (this is the $\textsf{MAC}^{0}$ circuit). We refer to the proof (see Theorem~\ref{ThmAbelian}) for full details.

\begin{remark}
We note that this and \Thm{ThmSemisimpleIntro} illustrate uses of WL for groups as a \emph{subroutine} in isomorphism testing, which is how it is so frequently used in the case of graphs. To the best of our knowledge, the only previous uses of WL as a subroutine for \algprobm{GpI} were in \cite{QiaoLiWL, BGLQW}. In particular, \Thm{thm:AbelianGpI} motivated follow-up work by Collins \& Levet \cite{CollinsLevetWL, CollinsUndergradThesis}, who leveraged count-free WL Version I in a similar manner to obtain novel parallel complexity bounds for isomorphism testing of several families of groups. Most notably, they improved the complexity of isomorphism testing for the \textit{CFI groups} from $\cc{TC}^1$ \cite{WLGroups} to $\beta_{1}\textsf{MAC}^{0}(\textsf{FOLL})$. The CFI groups are highly non-trivial, arising via Mekler's construction \cite{Mekler, HeQiao} from the CFI graphs \cite{CFI}.
\end{remark}

\begin{remark}
The previous best upper bounds for isomorphism testing of Abelian groups are linear time \cite{Kavitha, Vikas, Savage} and $\textsf{L} \cap \textsf{TC}^{0}(\textsf{FOLL})$ \cite{ChattopadhyayToranWagner}. As $\beta_{1}\textsf{MAC}^{0}(\textsf{FOLL}) \subseteq \textsf{TC}^{0}(\textsf{FOLL})$, \Thm{thm:AbelianGpI} improves the upper bound for isomorphism testing of Abelian groups.
\end{remark}

\paragraph{Methods.}
We find the comparison of methods at least as interesting as the comparison of complexity.  Here we discuss at a high level the methods we use for our main theorems above, and compare them to the methods of their predecessor results.

For \Thm{ThmMainQST}, its predecessor in Qiao--Sarma--Tang \cite{QST11} leveraged a result of Le Gall \cite{Gal09} on testing conjugacy of elements in the automorphism group of an Abelian group. (By further delving into the representation theory of Abelian groups, they were also able to solve the case where $H$ and $N$ are coprime and both are Abelian without any restriction on number of generators; we leave that as an open question in the setting of WL.) Here, we use the standard pebbling game equivalent to WL (see Section~\ref{sec:pebble}). Our approach is to first pebble generators for the complement $H$, which fixes an isomorphism between $H$ and its image. For groups that decompose as a coprime extension of $H$ and $N$, the isomorphism type is completely determined by the multiplicities of the indecomposable $H$-module direct summands (\Lem{LemTaunt}). So far, this is the same group-theoretic structure leveraged by Qiao, Sarma, and Tang \cite{QST11}. However, we then use the representation-theoretic fact that, since $|N|$ and $|H|$ are coprime, each indecomposable $H$-module is generated by a single element (\Lem{LemThevenaz}); this is crucial in our setting, as it allows Spoiler to pebble that one element in the WL pebbling game. Then, as the isomorphism of $H$ is fixed, we show that any subsequent bijection that Duplicator selects must restrict to $H$-module isomorphisms on each indecomposable $H$-submodule of $N$ that is a direct summand. 

For \Thm{ThmSemisimpleIntro}, solving isomorphism of semisimple groups took a series of two papers \cite{BCGQ, BCQ}. Our result is really only a parallel improvement on the first of these (we leave the second as an open question). In Babai \emph{et al.} \cite{BCGQ}, they used \algprobm{Code Equivalence} techniques to identify semisimple groups where the minimal normal subgroups have a bounded number of non-Abelian simple direct factors, and to identify general semisimple groups in time $n^{O(\log \log n)}$. In contrast, WL---along with individualize-and-refine in the second case---provides a single, combinatorial algorithm that is able to detect the same group-theoretic structures leveraged in previous works to solve isomorphism in these families.

In parallelizing Brachter \& Schweitzer's direct product result in \Thm{ThmProductCanonicalForm}, we use two techniques. The first is simply carefully analyzing the number of rounds used in many of the proofs. In several cases, a careful analysis of the rounds used was not sufficient to get a strong parallel result. In those cases, we use the notion of \emph{rank}, which may be of independent interest and have further uses. 

Given a subset $C$ of group elements, the $C$-rank of $g \in G$ is the minimal word-length over $C$ required to generate $g$. Our next result shows that if $C$ is easily identified by Weisfeiler--Leman, then WL can identify $\langle C \rangle$ in $O(\log n)$ rounds. More precisely, we say that $C \subseteq G$ is \emph{distinguished by the version $J$ pebble game with $k$ pebbles and $r$ rounds} if 
for all $g \in C$ and $g' \notin C$, Spoiler can win this pebble game, played on $G$ and $G$, starting from the configuration $g \mapsto g'$. We then have:


\begin{lemma}[Rank lemma; cf. Lemma~\ref{LemmaRank}] \label{LemmaRankIntro}
Let $k \geq 3$, and $J \in \{I,II\}$. Suppose $C \subseteq G$ is distinguished (as defined above) by the Version $J$ pebble game with $k$ pebbles and $r$ rounds. In the pebble game played on $G$ and another group $H$, suppose that Duplicator chooses a bijection $f$ such that for some $g \in G$, $\rk_{C}(g) \neq \rk_{f(C)}(f(g))$. Then Spoiler can win with $k+1$ pebbles and $r + O(\log d)$ rounds, where $d = \text{diam}(\text{Cay}(\langle C \rangle, C)) \leq |\langle C \rangle| \leq |G|$.
%
\end{lemma}

One application of our Rank Lemma is that WL identifies verbal subgroups where the words are easily identified. Given a set of words $w_1(x_1, \dotsc, x_n), \dotsc, w_m(x_1, \ldots, x_n)$, the corresponding \emph{verbal subgroup} is the subgroup generated by $\{w_i(g_1, \dotsc, g_n) : i=1,\dotsc, m, g_j \in G\}$. One example that we use in our results is the commutator subgroup. If Duplicator chooses a bijection $f : G \to H$ such that $f([x,y])$ is not a commutator in $H$, then Spoiler pebbles $[x,y] \mapsto f([x,y])$ and wins in two additional rounds. Thus, by our Rank Lemma, if Duplicator does not map the commutator subgroup $[G, G]$ to the commutator subgroup $[H,H]$, then Duplicator wins with $1$ additional pebble and $O(\log n)$ additional rounds.

Brachter \& Schweitzer \cite{BrachterSchweitzerWLLibrary} obtained a similar result about verbal subgroups using different techniques. Namely, they showed that if WL assigns a distinct coloring to certain subsets $S_{1}, \ldots, S_{t}$, then WL assigns a unique coloring to the set of group elements satisfying systems of equations over $S_{1}, \ldots, S_{t}$. They analyzed the WL colorings directly. As a result, it is not clear how to compose their result with the pebble game in a manner that also allows us to control rounds. For instance, while their result implies that if Duplicator does not map $f([G,G]) = [H,H]$ then Spoiler wins, it is not clear how Spoiler wins nor how quickly Spoiler can win. Our result addresses these latter two points more directly. Recall that the number of rounds is the crucial parameter affecting both the parallel complexity and quantifier depth.

\paragraph{Related Work.} There has been considerable work on efficient parallel ($\textsf{NC}$) isomorphism tests for graphs \cite{LindellTreeCanonization, BirgitKoblerMcKenzieToran, KoblerVerbitsky,WagnerBoundedTreewidth, ElberfeldSchweitzer, GroheVerbitsky, GroheKieferPlanar, DattaLimayeNimbhorkarPrajaktaThieraufWagner, datta_et_al:LIPIcs:2009:2314, ARVIND20121}. In contrast with the work on serial runtime complexity, the literature on the space and parallel complexity for $\algprobm{GpI}$ is quite minimal. Around the same time as Tarjan's $n^{\log_{p}(n) + O(1)}$-time algorithm for $\algprobm{GpI}$ \cite{MillerTarjan}, Lipton, Snyder, and Zalcstein showed that $\algprobm{GpI} \in \textsf{DSPACE}(\log^{2}(n))$ \cite{LiptonSnyderZalcstein}. This bound has been improved to $\beta_{2}\textsf{NC}^{2}$ ($\textsf{NC}^{2}$ circuits that receive $O(\log^{2}(n))$ non-deterministic bits as input) \cite{Wolf}, and subsequently to $\beta_{2}\textsf{L} \cap \beta_{2}\textsf{FOLL} \cap \beta_{2}\textsf{SC}^{2}$ \cite{ChattopadhyayToranWagner, TangThesis}. In the case of Abelian groups, Chattopadhyay, Tor\'an, and Wagner showed that $\algprobm{GpI} \in \textsf{L} \cap \textsf{TC}^{0}(\textsf{FOLL})$ \cite{ChattopadhyayToranWagner}. Tang showed that isomorphism testing for groups with a bounded number of generators can also be done in $\textsf{L}$ \cite{TangThesis}. Since composition factors of permutation groups can be identified in $\textsf{NC}$ \cite{BabaiLuksSeress} (see also \cite{BealsCompositionFactors} for a CFSG-free proof), isomorphism testing \emph{between} two permutation groups that are both direct products of simple groups (Abelian or non-Abelian) can be done in $\textsf{NC}$, using the regular representation, though this does not allow one to test isomorphism of such a group against an arbitrary permutation group. To the best of our knowledge, no other specific family of groups is known to admit an $\textsf{NC}$-computable isomorphism test prior to our paper.

Combinatorial techniques, such as individualization with Weisfeiler--Leman refinement, have also been incredibly successful in $\algprobm{GI}$, yielding efficient isomorphism tests for several families \cite{GroheVerbitsky, KieferPonomarenkoSchweitzer, GroheKieferPlanar, grohe_et_al:LIPIcs:2019:10693, grohe2019canonisation, BabaiWilmes, ChenSunTeng}. Weisfeiler--Leman is also a key subroutine in Babai's quasipolynomial-time isomorphism test \cite{BabaiGraphIso}. Despite the successes of such combinatorial techniques, they are known to be insufficient to place $\algprobm{GI}$ into $\textsf{P}$ \cite{CFI, NeuenSchweitzerIR}. In contrast, the use of combinatorial techniques for $\algprobm{GpI}$ is relatively new \cite{QiaoLiWL, BGLQW, WLGroups, BrachterSchweitzerWLLibrary}, and it is a central open problem as to whether such techniques are sufficient to improve even the long-standing upper-bound of $n^{\Theta(\log n)}$ runtime.

Examining the distinguishing power of the counting logic $\mathcal{C}_{k}$ serves as a measure of descriptive  complexity for groups. In the setting of graphs, the descriptive complexity has been extensively studied, with \cite{GroheBook} serving as a key reference in this area. There has been recent work relating first order logics and groups \cite{FiniteGroupsFOL}, as well as work examining the descriptive complexity of finite abelian groups \cite{DescriptiveComplexityAbelianGroups}. However, the work on the descriptive complexity of groups is scant compared to the algorithmic literature on \algprobm{GpI}. 

Ehrenfeucht--Fra\"iss\'e games \cite{Ehrenfeucht, Fraisse}, also known as pebbling games, serve as another tool in proving the inexpressibility of certain properties in first-order logics. 
Pebbling games have served as an important tool in analyzing graph properties like reachability \cite{ajtai_fagin_1990, ARORA199797}, designing parallel algorithms for graph isomorphism \cite{GroheVerbitsky}, and isomorphism testing of random graphs \cite{Rossman2009EhrenfeuchtFrassGO}.

\section{Preliminaries}

\subsection{Groups}

Unless stated otherwise, all groups are assumed to be finite and represented by their Cayley tables. 
For a group of order $n$, the Cayley table has $n^{2}$ entries, each represented by a binary string of size $\lceil \log_{2}(n) \rceil$. For an element $g$ in the group $G$, we denote the \textit{order} of $g$ by $|g|$. We use $d(G)$ to denote the minimum size of a generating set for the group $G$. 

The \textit{socle} of a group $G$, denoted $\text{Soc}(G)$, is the subgroup generated by the minimal normal subgroups of $G$. If $G$ has no Abelian normal subgroups, then $\text{Soc}(G)$ decomposes as the direct product of non-Abelian simple factors. The \textit{normal closure} of a subset $S \subseteq G$, which we denote $\text{ncl}(S)$, is the smallest normal subgroup of $G$ that contains $S$.

We say that a normal subgroup $N \trianglelefteq G$ \textit{splits} in $G$ if there exists a subgroup $H \leq G$ such that $H \cap N = \{1\}$ and $G = HN$. The conjugation action of $H$ on $N$ allows us to express multiplication of $G$ in terms of pairs $(h, n) \in H \times N$. We note that the conjugation action of $H$ on $N$ induces a group homomorphism $\theta : H \to \Aut(N)$ mapping $h \mapsto \theta_{h}$, where $\theta_{h} : N \to N$ sends $\theta_{h}(n) = hnh^{-1}$. So given $(H, N, \theta)$, we may define the group $H \ltimes_{\theta} N$ on the set $\{ (h, n) : h \in H, n \in N \}$ with the product $(h_{1}, n_{1})(h_{2}, n_{2}) = (h_{1}h_{2}, \theta_{h_{2}^{-1}}(n_{1})n_{2})$. We refer to the decomposition $G = H \ltimes_{\theta} N$ as a \textit{semidirect product} decomposition. When the action $\theta$ is understood, we simply write $G = H \ltimes N$. 

A \emph{Hall} subgroup of a group $G$ is a subgroup $N$ such that $|N|$ is coprime to $|G/N|$. We are particularly interested in semidirect products when $N$ is a normal Hall subgroup. To this end, we recall the Schur--Zassenhaus Theorem \cite[(9.1.2)]{Robinson1982}.

\begin{theorem}[Schur--Zassenhaus]
Let $G$ be a finite group of order $n$, and let $N$ be a normal Hall subgroup. Then there exists a complement $H \leq G$, such that $\text{gcd}(|H|, |N|) = 1$ and $G = H \ltimes N$. Furthermore, if $H$ and $K$ are complements of $N$, then $H$ and $K$ are conjugate.
\end{theorem}

We will use the following standard observation a few times:

\begin{fact} \label{FactWordLength}
Let $G = \langle g_1, \dotsc, g_d \rangle$. Then every element of $G$ can be written as a word in the $g_i$ of length at most $|G|$.
\end{fact}

\begin{proof}
Consider the Cayley graph of $G$ with generating set $g_1, \dotsc, g_d$. Words correspond to walks in this graph. We need only consider simple walks---those which never visit any vertex more than once---since if a walk visits a group element $g$ more than once, then the part of that walk starting and ending at $g$ is a word that equals the identity element, so it can be omitted. But the longest simple walk is at most the number of vertices, which is $|G|$.
\end{proof}

\subsection{Weisfeiler--Leman}

We begin by recalling the Weisfeiler--Leman algorithm for graphs, which computes an isomorphism-invariant coloring. Let $\Gamma$ be a graph, and let $k \geq 2$ be an integer. The $k$-dimensional Weisfeiler--Leman, or $k$-WL, algorithm begins by constructing an initial coloring $\chi_{k,0} : V(\Gamma)^{k} \to \mathcal{K}$, where $\mathcal{K}$ is our set of colors, by assigning each $k$-tuple a color based on its isomorphism type. That is, two $k$-tuples $(v_{1}, \ldots, v_{k})$ and $(u_{1}, \ldots, u_{k})$ receive the same color under $\chi_{k,0}$ iff the map $v_i \mapsto u_i$ (for all $i \in [k]$) is an isomorphism of the induced subgraphs $\Gamma[\{ v_{1}, \ldots, v_{k}\}]$ and $\Gamma[\{u_{1}, \ldots, u_{k}\}]$ and for all $i, j$, $v_i = v_j \Leftrightarrow u_i = u_j$. 

For $r \geq 1$, the coloring computed at the $r$th iteration of  Weisfeiler--Leman is refined as follows. For a $k$-tuple $\overline{v} = (v_{1}, \ldots, v_{k})$ and a vertex $x \in V(\Gamma)$, define
\[
\overline{v}(v_{i}/x) = (v_{1}, \ldots, v_{i-1}, x, v_{i+1}, \ldots, v_{k}).
\]

The coloring computed at the $(r+1)$st iteration, denoted $\chi_{k, r+1}$, stores the color of the given $k$-tuple $\overline{v}$ at the $r$th iteration, as well as the colors under $\chi_{k,r}$ of the $k$-tuples obtained by substituting a single vertex in $\overline{v}$ for another vertex $x$. We examine this multiset of colors over all such vertices $x$. This is formalized as follows:
\begin{align*}
\chi_{k,r+1}(\overline{v}) = &( \chi_{r}(\overline{v}), \{\!\!\{ ( \chi_{r}(\overline{v}(v_{1}/x)), \ldots, \chi_{r}(\overline{v}(v_{k}/x) ) \bigr| x \in V(\Gamma) \}\!\!\} ),
\end{align*}
where $\{\!\!\{ \cdot \}\!\!\}$ denotes a multiset.

Note that the coloring $\chi_{k,r}$ computed at iteration $r$ induces a partition of $V(\Gamma)^{k}$ into color classes. The Weisfeiler--Leman algorithm terminates when this partition is not refined, that is, when the partition induced by $\chi_{k,r+1}$ is identical to that induced by $\chi_{k,r}$. The final coloring is referred to as the \textit{stable coloring}, which we denote $\chi_{k,\infty} := \chi_{k,r}$. 

The \textit{count-free} variant of $k$-WL works identically as the classical variant, except at the refinement step, we consider the set of colors rather than the full multi-set. We re-use the notation $\chi_{k,r}$ to denote the coloring computed by count-free $(k,r)$-WL; context should make it clear whether $\chi_{k,r}$ refers to count-free or counting WL (we never use $\chi_{k,r}$ to denote the count-free coloring when discussing counting WL, nor vice versa).
Precisely:
\begin{align*}
\chi_{k,r+1}(\overline{v}) = &( \chi_{r}(\overline{v}), \{ ( \chi_{r}(\overline{v}(v_{1}/x)), \ldots, \chi_{r}(\overline{v}(v_{k}/x) ) \bigr| x \in V(\Gamma) \} ).
\end{align*}

Let $k \geq 2, r \geq 0$, and let $G$ be a graph. We say that the classical counting variant of $(k,r)$-WL \textit{distinguishes} $G$ from the graph $H$ if there exists a color class $C$ such that:
\[
|\{ \bar{x} \in V(G)^k : \chi_{k,r}(\bar{x}) = C \}| \neq |\{ \bar{x} \in V(H)^{k} : \chi_{k,r}(\bar{x}) = C\}|.
\]
Similarly, the count-free variant of $(k,r)$-WL \textit{distinguishes} $G$ from the graph $H$ if $|V(G)| \neq |V(H)|$ or there exists a color class $C$ and some $\bar{x} \in V(G)^k$ where $\chi_{k,r}(\bar{x}) = C$, but for all $\bar{y} \in V(H)^k$, $\chi_{k,r}(\bar{y}) \neq C$. We say that the classical counting (resp., count-free) $(k,r)$-WL \textit{identifies} $G$ if for all $H \not \cong G$, $(k,r)$-WL (resp., count-free $(k,r)$-WL) distinguishes $G$ from $H$. The terms \textit{distinguish} and \textit{identify} also extend in the natural way to WL on groups.

Brachter \& Schweitzer introduced three variants of WL for groups. We will restrict attention to the first two variants. WL Versions I and II are both executed directly on the Cayley tables, where $k$-tuples of group elements are initially colored. For WL Version I, two $k$-tuples $(g_{1}, \ldots, g_{k})$ and $(h_{1}, \ldots, h_{k})$ receive the same initial color iff (a) for all $i, j, \ell \in [k]$, $g_{i}g_{j} = g_{\ell} \iff h_{i}h_{j} = h_{\ell}$, and (b) for all $i, j \in [k]$, $g_{i} = g_{j} \iff h_{i} = h_{j}$. For WL Version II, $(g_{1}, \ldots, g_{k})$ and $(h_{1}, \ldots, h_{k})$ receive the same initial color iff the map $g_{i} \mapsto h_{i}$ for all $i \in [k]$ extends to an isomorphism of the generated subgroups $\langle g_{1}, \ldots, g_{k} \rangle$ and $\langle h_{1}, \ldots, h_{k} \rangle$. For both WL Versions I and II, refinement is performed in the classical manner as for graphs. Namely, for a given $k$-tuple $\overline{g}$ of group elements,
\begin{align*}
\chi_{k,r+1}(\overline{g}) = &( \chi_{r}(\overline{g}), \{\!\!\{ ( \chi_{r}(\overline{g}(g_{1}/x)), \ldots, \chi_{r}(\overline{g}(g_{k}/x) ) \bigr| x \in G \}\!\!\} ).
\end{align*}

The count-free variants of WL Versions I and II are defined in the identical manner as for graphs.

\subsection{Pebbling Game} \label{sec:pebble}

We recall the bijective pebble game introduced by Hella \cite{Hella1989, Hella1993} for WL on graphs. This game is often used to show that two graphs $X$ and $Y$ cannot be distinguished by $k$-WL. The game is an Ehrenfeucht--Fra\"iss\'e game (cf., \cite{Ebbinghaus:1994, Libkin}), with two players: Spoiler and Duplicator. We begin with $k+1$ pairs of pebbles. Prior to the start of the game, each pebble pair $(p_{i}, p_{i}')$ is initially placed either beside the graphs or on a given pair of vertices $v_{i} \mapsto v_{i}'$ (where $v_{i} \in V(X), v_{i}' \in V(Y)$). Each round $r$ proceeds as follows.
\begin{enumerate}
\item Spoiler picks up a pair of pebbles $(p_{i}, p_{i}^{\prime})$. 
\item We check the winning condition, which will be formalized later.
\item Duplicator chooses a bijection $f_{r} : V(X) \to V(Y)$ (where here, we emphasize that the bijection chosen depends on the round- and, implicitly, the pebbling configuration at the start of said round).
\item Spoiler places $p_{i}$ on some vertex $v \in V(X)$. Then $p_{i}^{\prime}$ is placed on $f(v)$. 
\end{enumerate} 

Let $v_{1}, \ldots, v_{m}$ be the vertices of $X$ pebbled at the end of step 1 at round $r$ of the game, and let $v_{1}^{\prime}, \ldots, v_{m}^{\prime}$ be the corresponding pebbled vertices of $Y$. Spoiler wins precisely if the map $v_{\ell} \mapsto v_{\ell}^{\prime}$ does not extend to an isomorphism of the induced subgraphs $X[\{v_{1}, \ldots, v_{m}\}]$ and $Y[\{v_{1}^{\prime}, \ldots, v_{m}^{\prime}\}]$. Duplicator wins otherwise. Spoiler wins, by definition, at round $0$ if $X$ and $Y$ do not have the same number of vertices. We note that $\overline{v} \in X^{k}$ and $\overline{v'} \in Y^{k}$ are not distinguished by the first $r$ rounds of $k$-WL if and only if Duplicator wins the first $r$ rounds of the $(k+1)$-pebble game starting from the configuration $\overline{v} \mapsto \overline{v'}$ \cite{Hella1989, Hella1993, CFI}. 

For groups instead of graphs, Versions I and II of the pebble game are defined analogously, where Spoiler pebbles group elements on the Cayley tables. Precisely, for groups $G$ and $H$, each round proceeds as follows.
\begin{enumerate}
\item Spoiler picks up a pair of pebbles $(p_{i}, p_{i}^{\prime})$. 
\item We check the winning condition, which will be formalized later.
\item Duplicator chooses a bijection $f_{r} : G \to H$.
\item Spoiler places $p_{i}$ on some vertex $g \in G$. Then $p_{i}^{\prime}$ is placed on $f(g)$. 
\end{enumerate} 

Suppose that $(g_{1}, \ldots, g_{\ell}) \mapsto (h_{1}, \ldots, h_{\ell})$ have been pebbled. In Version I, Duplicator wins at the given round if this map satisfies the initial coloring condition of WL Version I: (a) for all $i, j, m \in [\ell]$, $g_{i}g_{j} = g_{m} \iff h_{i}h_{j} = h_{m}$, and (b) for all $i, j \in [\ell]$, $g_{i} = g_{j} \iff h_{i} = h_{j}$. In Version II, Duplicator wins at the given round if the map $(g_{1}, \ldots, g_{\ell}) \mapsto (h_{1}, \ldots, h_{\ell})$ extends to an isomorphism of the generated subgroups $\langle g_{1}, \ldots, g_{\ell} \rangle$ and $\langle h_{1}, \ldots, h_{\ell} \rangle.$ Brachter \& Schweitzer established that for $J \in \{I, II\}$, $(k,r)$-WL Version J is equivalent to version J of the $(k+1)$-pebble, $r$-round pebble game \cite{WLGroups}. 

\subsection{Weisfeiler--Leman as a Parallel Algorithm} \label{sec:WLparallel}

Grohe \& Verbitsky \cite{GroheVerbitsky} previously showed that for fixed $k$, the classical $k$-dimensional Weisfeiler--Leman algorithm for graphs can be effectively parallelized. More precisely, each iteration (including the initial coloring) can be implemented using a logspace uniform $\textsf{TC}^{0}$ circuit. As they mention \cite[Remark~3.4]{GroheVerbitsky}, their implementation works for any first-order structure, including groups. However, because for groups we have different versions of WL, we explicitly list out the resulting parallel complexities, which differ slightly between the versions.

\begin{itemize}
\item \textbf{WL Version I:} Let $(g_{1}, \ldots, g_{k})$ and $(h_{1}, \ldots, h_{k})$ be two $k$-tuples of group elements. We may test in $\textsf{AC}^{0}$ whether (a) for all $i, j, m \in [k]$, $g_{i}g_{j} = g_{m} \iff h_{i}h_{j} = h_{m}$, and (b) $g_{i} = g_{j} \iff h_{i} = h_{j}$. So we may decide if two $k$-tuples receive the same initial color in $\textsf{AC}^{0}$. Comparing the multiset of colors at the end of each iteration (including after the initial coloring), as well as the refinement steps, proceed identically as in \cite{GroheVerbitsky}. Thus, for fixed $k$, each iteration of $k$-WL Version I can be implemented using a logspace uniform $\textsf{TC}^{0}$. 

\item \textbf{WL Version II:} Let $(g_{1}, \ldots, g_{k})$ and $(h_{1}, \ldots, h_{k})$ be two $k$-tuples of group elements. We may use the marked isomorphism test of Tang \cite{TangThesis} to test in $\textsf{L}$ whether the map sending $g_{i} \mapsto h_{i}$ for all $i \in [k]$ extends to an isomorphism of the generated subgroups $\langle g_{1}, \ldots, g_{k} \rangle$ and $\langle h_{1}, \ldots, h_{k} \rangle$. So we may decide whether two $k$-tuples receive the same initial color in $\textsf{L}$. Comparing the multiset of colors at the end of each iteration (including after the initial coloring), as well as the refinement steps, proceed identically as in \cite{GroheVerbitsky}. Thus, for fixed $k$, the initial coloring of $k$-WL Version II is $\textsf{L}$-computable, and each refinement step is $\textsf{TC}^{0}$-computable.
\end{itemize}

\subsection{Complexity Classes} \label{sec:Complexity}
We assume familiarity with the complexity classes $\textsf{P}, \textsf{NP}, \textsf{L}$, $\textsf{NL}$, $\textsf{NC}^{k}, \textsf{AC}^{k}$, and $\textsf{TC}^{k}$- we defer the reader to standard references \cite{ComplexityZoo, AroraBarak}. The complexity class $\textsf{SAC}^{k}$ is defined analogously to $\textsf{AC}^{k}$, except that the $\textsf{AND}$ gates have bounded fan-in (while the OR gates may still have unbounded fan-in). The complexity class $\textsf{FOLL}$ is the set of languages decidable by uniform circuit families with \textsf{AND}, \textsf{OR}, and \textsf{NOT} gates of depth $O(\log \log n)$, polynomial size, and unbounded fan-in. It is known that $\textsf{AC}^{0} \subsetneq \textsf{FOLL} \subsetneq \textsf{AC}^{1}$, and it is open as to whether $\textsf{FOLL}$ is contained in $\textsf{NL}$ \cite{BKLM}. 

The complexity class $\textsf{MAC}^{0}$ is the set of languages decidable by constant-depth uniform circuit families with a polynomial number of \textsf{AND}, \textsf{OR}, and \textsf{NOT} gates, and at most one $\textsf{Majority}$ gate. The class $\textsf{MAC}^{0}$ was introduced (but not so named) in \cite{VotingPolynomials}, where it was shown that $\textsf{MAC}^{0} \subsetneq \textsf{TC}^{0}$. This class was subsequently given the name $\textsf{MAC}^{0}$ in \cite{LearnabilityAC0}.

For a complexity class $\mathcal{C}$, we define $\beta_{i}\mathcal{C}$ to be the set of languages $L$ such that there exists an $L' \in \mathcal{C}$ such that $x \in L$ if and only if there exists $y$ of length at most $O(\log^{i} |x|)$ such that $(x, y) \in L'$. For any $i, c \geq 0$, $\beta_{i}\textsf{FO}((\log \log n)^{c})$ cannot compute \algprobm{Parity} \cite{ChattopadhyayToranWagner}.

We will also allow circuits to compute functions by using multiple output gates. For function complexity classes $\mathcal{C}_{1}, \mathcal{C}_{2}$, the complexity class $\mathcal{C}_{1}(\mathcal{C}_{2})$ is the class of $h = g \circ f$, where $g$ 
is $\mathcal{C}_{1}$-computable and 
$f$ 
is $\mathcal{C}_{2}$-computable . For instance, $\beta_{1}\textsf{MAC}^{0}(\textsf{FOLL})$ is the set of functions $h = g \circ f$, where $f$ is $\textsf{FOLL}$-computable and $g$ is $\beta_{1}\textsf{MAC}^{0}$-computable.

The function class $\textsf{FP}$ is the class of polynomial-time computable functions and $\textsf{FL}$ is the class of logspace-computable functions.

\section{Weisfeiler--Leman for coprime extensions}

In this section, we consider groups that admit a Schur--Zassenhaus decomposition of the form $G = H \ltimes N$, where $N$ is Abelian, and $H$ is $O(1)$-generated and $|H|$ and $|N|$ are coprime. Qiao, Sarma, and Tang \cite{QST11} previously exhibited a polynomial-time isomorphism test for this family of groups, as well as the family where $H$ and $N$ are arbitrary Abelian groups of coprime order. This was extended by Babai \& Qiao \cite{BQ} to groups where $|H|$ and $|N|$ are coprime, $N$ is Abelian, and $H$ is an arbitrary group given by generators in any class of groups for which isomorphism can be solved efficiently. Among the class of such coprime extensions where $H$ is $O(1)$-generated and $N$ is Abelian, we are able to improve the parallel complexity to $\cc{L}$ via WL Version II. 

\subsection{Additional preliminaries for groups with Abelian normal Hall subgroup} \label{sec:PrelimsQST}
Here we recall additional preliminaries needed for our algorithm in the next section. None of the results in this section are new, though in some cases we have rephrased the known results in a form more useful for our analysis.

Recall that a \emph{Hall} subgroup of a group $G$ is a subgroup $N$ such that $|N|$ is coprime to $|G/N|$. When a Hall subgroup is normal, we refer to the group as a coprime extension. Coprime extensions are determined entirely by the isomorphism types of $N$, $H$ and their actions: 

\begin{lemma}[{Taunt \cite{Taunt1955}}]  \label{LemmaSemidirect}
Let $G = H \ltimes_{\theta} N$ and $\hat{G} = \hat{H} \ltimes_{\hat{\theta}} \hat{N}$. If $\alpha : H \to \hat{H}$ and $\beta : N \to \hat{N}$ are isomorphisms such that for all $h \in H$ and all $n \in N$,
\[
\hat{\theta}_{\alpha(h)}(n) = (\beta \circ \theta_{h} \circ \beta^{-1})(n),
\]
then the map $(h, n) \mapsto (\alpha(h), \beta(n))$ is an isomorphism of $G \cong \hat{G}$. Conversely, if $G$ and $\hat{G}$ are isomorphic and $|H|$ and $|N|$ are coprime, then there exists an isomorphism of this form.
\end{lemma}

\begin{remark}
\Lem{LemmaSemidirect} can be significantly generalized to arbitrary extensions where the subgroup is characteristic. When the characteristic subgroup is Abelian, this is standard in group theory, and has been useful in practical isomorphism testing (see, e.g., \cite{Holt2005HandbookOC}). In general, the equivalence of group extensions deals with both \algprobm{Action Compatibility} and \algprobm{Cohomology Class Isomorphism}. Generalizations of cohomology to non-Abelian coefficient groups was done by Dedecker in the 1960s (e.g. \cite{dedecker}) and Inassaridze at the turn of the 21st century \cite{inassaridze}. Unaware of this prior work on non-Abelian cohomology at the time, Grochow \& Qiao re-derived some of it in the special case of $H^2$---the cohomology most immediately relevant to group extensions and the isomorphism problem---and showed how it could be applied to isomorphism testing \cite[Lemma~2.3]{GQcoho}, generalizing Taunt's Lemma. In the setting of coprime extensions, the Schur--Zassenhaus Theorem provides that the first and second cohomology is trivial. Thus, in our setting we need only consider \algprobm{Action Compatibility}.
\end{remark}

A $\Z H$-module is an abelian group $N$ together with an action of $H$ on $N$, given by a group homomorphism $\theta\colon H \to \Aut(N)$. We refer to these as ``$H$-modules.'' A \emph{submodule} of an $H$-module $N$ is a subgroup $N' \leq N$ such that the action of $H$ on $N'$ sends $N'$ into itself, and thus the restriction of the action of $H$ to $N'$ gives $N'$ the structure of an $H$-module compatible with that on $N$. Given a subset $S \subseteq N$, the smallest $H$-submodule containing $S$ is denoted $\langle S \rangle_H$, and is referred to as the $H$-submodule \emph{generated by} $S$. An $H$-module generated by a single element is called \emph{cyclic}. Note that a cyclic $H$-module $N$ need not be a cyclic Abelian group.

Two $H$-modules $N,N'$ are isomorphic (as $H$-modules), denoted $N \cong_H N'$, if there is a group isomorphism $\varphi\colon N \to N'$ that is $H$-equivariant, in the sense that $\varphi(\theta(h)(n)) = \theta'(h)(\varphi(n))$ for all $h \in H, n \in N$. An $H$-module $N$ is \emph{decomposable} if $N \cong_H N_1 \oplus N_2$ where $N_1, N_2$ are nonzero $H$-modules (and the direct sum can be thought of as a direct sum of Abelian groups); otherwise $N$ is \emph{indecomposable}. An equivalent characterization of $N$ being decomposable is that there are nonzero $H$-submodules $N_1, N_2$ such that $N = N_1 \oplus N_2$ as Abelian groups (that is, $N$ is generated as a group by $N_1$ and $N_2$, and $N_1 \cap N_2 = 0$). The Remak--Krull--Schmidt Theorem says that every $H$-module decomposes as a direct sum of indecomposable modules, and that the multiset of $H$-module isomorphism types of the indecomposable modules appearing is independent of the choice of decomposition, that is, it depends only on the $H$-module isomorphism type of $N$. We may thus write 
\[
N \cong_H N_1^{\oplus m_1} \oplus N_2^{\oplus m_2} \oplus \dotsb \oplus N_k^{\oplus m_k}
\]
unambiguously, where the $N_i$ are pairwise non-isomorphic indecomposable $H$-modules. When we refer to the multiplicity of an indecomposable $H$-module as a direct summand in $N$, we mean the corresponding $m_i$.\footnote{For readers familiar with (semisimple) representations over fields, we note that the multiplicity is often equivalently defined as $\dim_\F \text{Hom}_{\F H}(N_i, N)$. However, when we allow $N$ to be an Abelian group that is not elementary Abelian, we are working with $(\Z/p^k \Z)[H]$-modules, and the characterization in terms of hom sets is more complicated, because one indecomposable module can be a submodule of another, which does not happen with semisimple representations. }

The version of Taunt's Lemma that will be most directly useful for us is:
\begin{lemma} \label{LemTaunt}
Suppose that $G_i = H \ltimes_{\theta_i} N$ for $i=1,2$ are two semi-direct products with $|H|$ coprime to $|N|$. Then $G_1 \cong G_2$ if and only if there is an automorphism $\alpha \in \Aut(H)$ such that each indecomposable $\Z H$-module appears as a direct summand in $(N, \theta_1)$ and in $(N, \theta_2 \circ \alpha)$ with the same multiplicity.
\end{lemma}

The lemma and its proof are standard but we include it for completeness.

\begin{proof}
If there is an automorphism $\alpha \in \Aut(H)$ such that the multiplicity of each indecomposable $\Z H$-module as a direct summand of $(N, \theta_1)$ and $(N, \theta_2 \circ \alpha)$ are the same, then there is a $\Z H$-module isomorphism $\beta \colon (N, \theta_1) \to (N, \theta_2 \circ \alpha)$ (in particular, $\beta$ is an automorphism of $N$ as a group). Then it is readily verified that the map $(h,n) \mapsto (\alpha(h), \beta(n))$ is an isomorphism of the two groups.

Conversely, suppose that $\varphi \colon G_1 \to G_2$ is an isomorphism. Since $|H|$ and $|N|$ are coprime, $N$ is characteristic in $G_i$, so we have $\varphi(N) = N$. And by order considerations $\varphi(H)$ is a complement to $N$ in $G_2$. We have $\theta_1(h)(n) = hnh^{-1}$. Since $\varphi$ is an isomorphism, we have $\varphi(\theta_1(h)(n)) = \varphi(hnh^{-1}) = \varphi(h) \varphi(n) \varphi(h)^{-1} = \theta_2(\varphi(h))(\varphi(n))$. Thus $\theta_1(h)(n) = \varphi^{-1}(\theta_2(\varphi(h))(\varphi(n)))$. So we may let $\alpha = \varphi|_H$, and then we have that $(N,\theta_1)$ is isomorphic to $(N, \theta_2 \circ \varphi|_H)$, where the isomorphism of $H$-modules is given by $\varphi|_N$. The Remak--Krull--Schmidt Theorem then gives the desired equality of multiplicities.
\end{proof}

The following lemma is needed for the case when $N$ is Abelian, but not elementary Abelian. A $(\Z/p^k \Z)[H]$-module is a $\Z H$-module $N$ where the exponent of $N$ (the LCM of the orders of the elements of $N$) divides $p^k$. 

\begin{lemma}[{see, e.\,g., Thev\'{e}naz \cite{Thevenaz}}] \label{LemThevenaz}
Let $H$ be a finite group. If $p$ is coprime to $|H|$, then any indecomposable $(\Z/p^k \Z)[H]$-module is generated (as an $H$-module) by a single element.
\end{lemma}

\begin{proof}
Thev\'{e}naz \cite[Cor.~1.2]{Thevenaz} shows that there are cyclic $(\Z/p^k \Z)[H]$-modules $M_1, \dotsc, M_n$, each with underlying group of the form $(\Z/p^k \Z)^{d_i}$ for some $d_i$,  such that each indecomposable $(\Z/p^k \Z)[H]$-module is of the form $M_i / p^j M_i$ for some $i,j$, and for distinct pairs $(i,j)$ we get non-isomorphic modules. 
\end{proof}

\subsection{Coprime extensions with an $O(1)$-generated complement}
Our approach is to first pebble generators for the complement $H$, which fixes an isomorphism of $H$. As the isomorphism of $H$ is then fixed, we show that any subsequent bijection that Duplicator selects must restrict to $H$-module isomorphisms on each indecomposable $H$-submodule of $N$ that is a direct summand. For groups that decompose as a coprime extension of $H$ and $N$, the isomorphism type is completely determined by the multiplicities of the indecomposable $H$-module direct summands (\Lem{LemTaunt}). We then leverage the fact that, in the coprime case, indecomposable $H$-modules are generated by single elements (\Lem{LemThevenaz}), making it easy for Spoiler to pebble.

\begin{theorem}
Let $G = H \ltimes N$, where $N$ is Abelian, $H$ is $O(1)$-generated, and $\text{gcd}(|H|, |N|) = 1$. We have that $(O(1), O(1))$-WL Version II identifies $G$.
\end{theorem}

\begin{proof}
Let $K$ be a group of order $n$ such that $K \not \cong G$. We will first show that if there is no action such that $K$ decomposes as $H \ltimes N$, then Spoiler can win with $O(1)$ pebbles and $O(1)$ rounds. Let $f : G \to K$ be the bijection that Duplicator selects. As $N \leq G$, as a subset, is uniquely determined by its orders---it is precisely the set of all elements in $G$ whose orders divide $|N|$---we may assume that $K$ has a normal Hall subgroup of size $|N|$. For first, if for some $n \in N$, $|n| \neq |f(n)|$, Spoiler can pebble $n \mapsto f(n)$ and win immediately. By reversing the roles of $K$ and $G$, it follows that $K$ must have precisely $|N|$ elements whose orders divide $|N|$. Second, suppose that those $|N|$ elements do not form a subgroup. Then there are two elements $x,y \in f(N)$ such that $xy \notin f(N)$. At the first round, Spoiler pebbles $a := f^{-1}(x) \mapsto x$. Let $f' : G \to K$ be the bijection Duplicator selects at the next round. As $K$ has precisely $|N|$ elements of order dividing $|N|$, we may assume that $f'(N) = f(N)$ (setwise). Let $b \in N$ s.t. $f'(b) = y$. Spoiler pebbles $b \mapsto y$. Now as $N$ is a group, $ab \in N$. However, as $f(a)f'(b) \not \in f(N)$, $|ab| \neq |f(a)f'(b)|$. So the map $(a, b) \mapsto (x, y)$ does not extend to an isomorphism. Spoiler now wins.

Now we have that $f(N)$ is a subgroup of $K$, and because it is the set of all elements of these orders, it is characteristic and thus normal. Suppose that $f(N) \not \cong N$. We have two cases: either $f(N)$ is not Abelian, or $f(N)$ is Abelian but $N \not \cong f(N)$. Suppose first that $f(N)$ is not Abelian. Let $x \in f(N)$ such that $x \not \in Z(f(N))$, and let $g := f^{-1}(x) \in N$. Spoiler pebbles $g \mapsto x$. Let $f' : G \to K$ be the bijection that Duplicator selects at the next round. We may again assume that $f'(N) = f(N)$ (setwise), or Spoiler wins with two additional pebbles and two additional rounds. Now let $y \in f(N)$ such that $[x,y] \neq 1$. Let $h \in G$ such that $f'(h) = y$. Spoiler pebbles $h \mapsto y$. Now the map $(g, h) \mapsto (x, y)$ does not extend to an isomorphism, so Spoiler wins. Suppose instead that $f(N)$ is Abelian. As Abelian groups are determined by their orders, we have by the discussion in the first paragraph that Spoiler wins with $2$ pebbles and $2$ rounds.

So now suppose that $N \cong f(N) \leq K$ is a normal Abelian Hall subgroup, but that $f(N)$ does not have a complement isomorphic to $H$. We note that if $K$ contains a subgroup $H'$ that is isomorphic to $H$, then by order considerations, $H'$ and $f(N)$ would intersect trivially in $K$ and we would have that $K = H' \cdot f(N)$. That is, $K$ would decompose as $K = H \ltimes f(N)$. So as $f(N)$ does not have a complement in $K$ that is isomorphic to $H$, by assumption we have that $K$ does not contain any subgroup isomorphic to $H$. 
In this case, Spoiler pebbles $k := d(H)$ generators of $H$ in $G$. As $K$ has no subgroup isomorphic to $H$, Spoiler immediately wins after the generators for $H \leq G$ have been pebbled. 

Finally, suppose that $K = H \ltimes N$. Spoiler uses the first $k$ rounds to pebble generators $(g_{1}, \ldots, g_{k}) \mapsto (h_{1}, \ldots, h_{k})$ for $H$. As $K = H \ltimes N$, we may assume that the map $(g_{1}, \ldots, g_{k}) \mapsto (h_{1}, \ldots, h_{k})$ induces an isomorphism with a copy of $H \leq K$; otherwise, Spoiler immediately wins. Let $f : G \to K$ be the bijection that Duplicator selects. As $G, K$ are non-isomorphic groups of the form $H \ltimes N$, they differ only in their actions. Now the actions are determined by the multiset of indecomposable $H$-modules in $N$. As $|H|, |N|$ are coprime, we have by \Lem{LemThevenaz} that the indecomposable $H$-modules are cyclic. As $G \not \cong K$, we have by \Lem{LemmaSemidirect} that there exists $n \in N \leq G$ such that $\langle n \rangle_{H}$ is indecomposable, and $\langle n \rangle_{H}$ and $\langle f(n) \rangle_{f(H)}$ are inequivalent $H$-modules. Spoiler now pebbles $n \mapsto f(n)$. Thus, the following map
\[
(g_{1}, \ldots, g_{k}, n) \mapsto (h_{1}, \ldots, h_{k}, f(n))
\]

\noindent does not extend to an isomorphism. So Spoiler wins. 
\end{proof}

\begin{remark}
We see that the main places we used coprimality were: (1) that $N$ was characteristic, and (2) that all indecomposable $H$-modules (in particular, those appearing in $N$) were cyclic. 
\end{remark}

\section{A ``rank'' lemma}
\begin{definition} \label{DefRank}
Let $C \subseteq G$ be a subset of a group $G$ that is closed under taking inverses. We define the \emph{C-rank} of $g \in G$, denoted $\rk_C(g)$, as the minimum $m$ such that $g$ can be written as a word of length $m$ in the elements of $C$. If $g$ cannot be so written, we define $\rk_C(g) = \infty$.
\end{definition}

Our definition and results actually extend to subsets that are not closed under taking inverses, but we won't have any need for that case, and it would only serve to make the wording less clear.

\begin{remark}
Our terminology is closely related to the usage of ``$X$-rank'' in algebra and geometry, which generalizes the notions of matrix rank and tensor rank: if $X \subseteq V$ is a subset of an $\mathbb{F}$-vector space, then the $X$-rank of a point $v \in V$ is the smallest number of elements of $x \in X$ such that $v$ lies in their linear span. If we replace $X$ by the union $\mathbb{F}^* X$ of its nonzero scaled versions (which is unnecessary in the most common case, in which $X$ is the cone over a projective variety), then the $X$-rank in the sense of algebraic geometry would be the $\mathbb{F}X$-rank in our terminology above. For example, matrix rank is $X$-rank inside the space of $n \times m$ matrices under addition, where $X$ is the set of rank-1 matrices (which is already closed under nonzero scaling).
\end{remark}

We say that $C \subseteq G$ is distinguished by the version $J$ pebble game with $k$ pebbles and $r$ rounds if for all $g \in C$ and $g' \notin C$, Spoiler can win this pebble game, played on $G$ and $G$, starting from the configuration $g \mapsto g'$. In terms of WL colorings, this is equivalent to: after $r$ rounds of $(k+1)$-WL Version $J$ coloring, the colors of all tuples containing elements of $C$ are distinct from the colors of all tuples not containing elements of $C$. In particular, if $C \subseteq G$ is distinguished by $(k+1,r)$-WL, then for any group $H$ of the same order as $G$, we define $C_H$ to be the set of those elements in $H$ whose $(k+1,r)$-WL colors are among the $(k+1,r)$-WL colors of $C$. It follows that in the Version $J$, $k$-pebble game on $G$ and $H$, Duplicator must choose bijections $f\colon G \to H$ that map $C$ to $C_H$, or lose within an additional $r$ rounds.

\begin{lemma}[Rank lemma] \label{LemmaRank}
Let $k \geq 3$, and $J \in \{I,II\}$. Suppose $C \subseteq G$ is distinguished (as defined above) by the Version $J$ pebble game with $k$ pebbles and $r$ rounds; let $C_H$ be the corresponding subset of $H$ as above. In the pebble game played on $G$ and another group $H$, suppose that at the end of a round there is a pebble on $g \mapsto h$ such that $\rk_C(g) \neq \rk_{C_H}(h)$. Then Spoiler can win with $k$ additional pebbles (beyond the one on $g \mapsto h$) and $r + \log d + O(1)$ additional rounds, where $d =\rk_C(g) \leq \text{diam}(\text{Cay}(\langle C \rangle, C)) \leq |\langle C \rangle| \leq |G|$.
\end{lemma}

Our primary uses of this lemma in this paper are to show that if $C$ is distinguished by $(k,r)$-WL, then $\langle C \rangle$ is distinguished by $(k, r + \log n + O(1))$-WL. However, the preservation of $C$-rank itself, rather than merely the subgroup generated by $C$, seems potentially useful for future applications. In particular, \Lem{LemmaRank} shows that WL can identify verbal subgroups in $O(\log n)$ rounds, provided WL can readily identify each word.

One sees that the version of the Rank Lemma stated in the Introduction (Lemma~\ref{LemmaRankIntro}) follows immediately from the above: if Duplicator chooses a bijection $f$ that does not preserve the rank of $g$, Spoiler's first move is to pebble $g \mapsto f(g)$, and then Lemma~\ref{LemmaRank} applies (note the $k+1$ in the statement of Lemma~\ref{LemmaRankIntro}, compared to $k$ in Lemma~\ref{LemmaRank}, is to account for this first pebble).


\begin{proof}

Throughout, we will use $\rk(x)$ to denote $\rk_C(x)$ if $x \in G$, and $\rk_{C_H}(x)$ if $x \in H$.

Without loss of generality (by swapping the roles of $G$ and $H$ if needed), we may assume that $\rk(g) < \rk(h)$. If $\rk(g)=1$, then by assumption, since $C = \{g \in G: \rk(g)=1\}$, Spoiler wins with $k$ pebbles in $r$ rounds.

On the other hand, if $\rk(g) > 1$, let $f : G \to H$ be the bijection that Duplicator selects at the next round. For brevity let $q = \rk(g)$. Note that this case works even if $\rk(h)=\infty$, for then by our assumption that $\rk(g) < \rk(h)$, we still have $q = \rk(g)$ is finite. Now write $g = g_{1} \cdots g_{q}$, where for each $i$, $g_i \in C$. For $1 \leq i \leq j \leq q$, write $g[i, \ldots, j] := g_{i} \cdots g_{j}$. We consider the following cases.
\begin{itemize}
\item \textbf{Case 1:} Suppose first that for all $x \in G$ with $\rk(x) < q$, we have $\rk(x) = \rk(f(x))$. In this case, Spoiler pebbles $g[2, \ldots, q] \mapsto f(g[2, \ldots, q])$. Let $f' : G \to H$ be the bijection that Duplicator selects at the next round; Spoiler pebbles $g_1 \mapsto f'(g_1)$.

If $\rk(g_{1}) = \rk(f'(g_{1})) = 1$, then $f'(g_{1}) \cdot f(g[2, \ldots, r+1]) \neq h$, since $\rk(f'(g_1)) = 1$ and $\rk(f(g[2, \ldots, q])) = q-1$, so their product has rank at most $q=\rk(g) < \rk(h)$. In this case, Spoiler 
wins immediately since the map on the pebbled elements $(g, g[2,\ldots,q], g_1) \mapsto (h, f(g[2, \ldots, q]), f'(g_1))$ does not satisfy the winning condition for Spoiler in the Version I (and hence, the Version II) game. 

If instead, $1=\rk(g_{1}) < \rk(f'(g_{1}))$, then Spoiler 
wins with $k-1$ additional pebbles and $r$ additional rounds by assumption. Note that once $g_{1} \mapsto f'(g_{1})$ has been pebbled, Spoiler can reuse the pebble on $g$; this is why we only need $k-1$ additional pebbles in this part rather than $k$ additional pebbles.

\item \textbf{Case 2:} Suppose instead that the hypothesis of Case 1 is not satisfied. Then there exists some $x \in \langle C \rangle$ with $\rk(x) \neq \rk(f(x))$ and $\rk(x) < q = \rk(g)$. We will now do a binary-search-like procedure to eventually reach an element of $C$ that is mapped outside of $C$. To make this clearer, we introduce two additional parameters as in binary search. Let $lo = 1$ and $hi=q$. In the next two rounds, Spoiler pebbles $x := g[lo, \ldots, lo + \lceil (hi-lo)/2 \rceil-1]$ and $y := g[lo + \lceil (hi-lo)/2 \rceil, \ldots, hi]$. Note that $\rk(x) = \lceil (hi-lo)/2 \rceil$ and $\rk(y) = (hi - lo) - \lceil(hi-lo)/2 \rceil + 1$, with $\rk(g) = \rk(x) + \rk(y)$, for if $x$ or $y$ could be written as a shorter word in $C$, then so could $g$. 

Let $f' : G \to H$ be the next bijection that Duplicator selects. If $h \neq f'(x)f'(y)$ then Spoiler immediately wins. On the other hand, if $h = f'(x) f'(y)$, 
then either
\[
\rk(x) < \rk(f'(x)) \qquad \text{ or } \qquad \rk(y) < \rk(f'(y)),
\]
since $\rk(h) > \rk(g) = \rk(x) + \rk(y)$. 

Without loss of generality, suppose that $\rk(x) < \rk(f'(x))$. Spoiler then updates $hi$ to $lo + \lceil (hi-lo)/2 \rceil-1$, and as $x \mapsto f'(x)$ is already pebbled, Spoiler iterates on this strategy as in binary search, reusing the pebbles on $g$ and $y$ in the next iteration, with $x$ now playing the role that $g$ played in the first iteration. Thus, iterating this procedure uses at most $3$ pebbles (and we have $3 \leq k$ by assumption). After $\log q + O(1)$ iterations, we reach the case where $hi=lo+1$, and we have pebbled an element $z \mapsto f''(z)$ such that $1 = \rk(z) < \rk(f''(z))$, and there are two additional pebbles placed that can be reused. 

By assumption, every element of $\langle C \rangle$ can be written as a word of length at most $d$ in the elements of $C$, so we have $q \leq d$, and thus only $\log d + O(1)$ rounds of the preceding iteration are required. (That $d \leq |\langle C \rangle|$ follows from Fact~\ref{FactWordLength}.)

Finally, since Spoiler has pebbled $z \in C$ mapping to an element outside of $C_H$, Spoiler implements the assumed strategy that identifies $C$ with $k$ pebbles and $r$ rounds. However, at the time Spoiler begins this strategy, two of the placed pebbles can be re-used, so the strategy only requires an additional $k-2$ pebbles, for a total of $k$ additional pebbles beyond the one that we started with on $g \mapsto h$. (We do not need to count $+1$ for the pebble on $z$, since the definition of ``distinguished'' includes ``starting from the position [$z \mapsto f''(z)$].'') \qedhere 
\end{itemize}

\end{proof}

\section{Direct products}

Brachter \& Schweitzer previously showed that Weisfeiler--Leman Version II can detect direct product decompositions in polynomial-time. Precisely, they showed that the WL dimension of a group $G$ is at most one more than the WL dimensions of the direct factors of $G$. We strengthen the result to control for rounds, showing that only additional $O(\log n)$ rounds are needed:

\begin{theorem} \label{ThmProduct}
Let $G = G_{1} \times \cdots \times G_{d}$ be a decomposition into indecomposable direct factors, let $k \geq 5$, and let $r := r(n)$. If $G$ and $H$ are not distinguished by $(k, r+O(\log n))$-WL Version II, then there exist direct factors $H_{i} \leq H$ such that $H = H_{1} \times \cdots \times H_{d}$ such that for all $i \in [d]$, $G_{i}$ and $H_{i}$ are not distinguished by $(k-1, r)$-WL Version II.
\end{theorem}

The structure and definitions in this section closely follow those of \cite[Sec.~6]{BrachterSchweitzerWLLibrary} for ease of comparison.

\subsection{Abelian and Semi-Abelian Case}

\begin{definition}[{\cite[Def.~6.5]{BrachterSchweitzerWLLibrary}}]
Let $G$ be a group. We say that $x \in G$ \textit{splits} from $G$ if there exists a complement $H \leq G$ such that $G = \langle x \rangle \times H$.
\end{definition}

We recall the following technical lemma \cite[Lemma~6.6]{BrachterSchweitzerWLLibrary} that characterizes the elements that split from an Abelian $p$-group.

\begin{lemma}[{\cite[Lemma~6.6]{BrachterSchweitzerWLLibrary}}]
Let $A$ be a finite Abelian $p$-group, and let $A = A_{1} \times \cdots \times A_{m}$ be an arbitrary cyclic decomposition. Then $a = (a_{1}, \ldots, a_{m}) \in A$ splits from $A$ if and only if there exists some $i \in [m]$ such that $|a| = |a_{i}|$ and $a_{i} \in A_{i} \setminus (A_{i})^{p}$. In particular, $x$ splits from $A$ if and only if there does not exist a $y \in A$ such that $|xy^{p}| < |x|$.
\end{lemma}

We utilize this lemma to show that WL Version II can detect elements that split from $A$.

\begin{lemma} \label{LemmaSylowSplit}
Let $A, B$ be Abelian groups, $p$ a prime, and $A_p \leq A$ and $B_p \leq B$ their respective Sylow $p$-subgroups. Let $f : A \to B$ be the bijection Duplicator selects. If $f(A_p) \neq B_p$, or if $x \in A_p$ splits from $A_p$, but $f(x) \in B_p$ does not split from $B_p$, then Spoiler can win the WL Version II game with $2$ pebbles and $2$ rounds.
\end{lemma}

\begin{proof}
If $f(A_p) \neq B_p$, then there is some $x \in A_p$ such that $f(x) \notin B_p$. But then $|x|$ is a power of $p$, while $|f(x)|$ is not, so Spoiler pebbles $x \mapsto f(x)$ and immediately wins. Thus we may assume $f(A_p) = B_p$.

Next, suppose $x \in A_p$ splits from $A_p$ but $f(x)$ does not split from $B_p$. Spoiler begins by pebbling $x \mapsto f(x)$. If $|x| \neq |f(x)|$ then Spoiler already wins, so we may now assume that $|x|=|f(x)|$. Let $f' : A \to B$ be the bijection that Duplicator selects at the next round. As $f(x)$ does not split from $B$, there exists $z \in B$ such that $|f(x) \cdot z^{p}| < |f(x)|$. Let $y = (f')^{-1}(z) \in A$. Spoiler pebbles $y \mapsto f'(y) = z$. Now $|xy^p|=|x| = |f(x)| \neq |f(x) \cdot z^p|$. So Spoiler immediately wins.
\end{proof}

To characterize when an element splits in a general Abelian group $A$, we begin by considering the decomposition of $A$ into its Sylow subgroups: $A = P_{1} \times \cdots \times P_{m}$. Now $x = (x_{1}, \ldots, x_{m}) \in A$ splits from $A$ if and only if for each $i \in [m]$, $x_{i}$ is either trivial or splits from $P_{i}$. See, e.g., \cite[Lemma~6.8]{BrachterSchweitzerWLLibrary}. The next lemma shows that this group-theoretic structure is useful in the pebble game.

\begin{lemma} \label{LemmaSplitAbelian}
Let $A, B$ be Abelian groups. Let $A = P_{1} \times \cdots \times P_{m}$ and $B = Q_{1} \times \cdots \times Q_{m}$, where the $P_{i}$ are the Sylow subgroups of $A$ and the $Q_{i}$ are the Sylow subgroups of $B$, where for each $i$, $P_i$ and $Q_i$ are $p_i$-subgroups for the same prime $p_i$. Let $f : A \to B$ be the bijection that Duplicator selects. Let $x = (x_{1}, \ldots, x_{m})$ be the decomposition of $x$, where $x_{i} \in P_{i}$, and let $f(x) = (y_{1}, \ldots, y_{m})$, where $y_{i} \in Q_{i}$. Suppose that Spoiler pebbles $x \mapsto f(x)$. Let $f' : A \to B$ be the bijection that Duplicator selects at the next round.
\begin{enumerate}[label=(\alph*)]
\item If $f'(x_{i}) \neq y_{i}$, then Spoiler can win with $1$ additional pebble and $1$ additional round.

\item If $x \in A$ splits from $A$, but $f(x)$ does not split from $B$, then Spoiler can win with $3$ additional pebbles and $3$ additional rounds (including the round at which the bijection $f'$ was selected).
\end{enumerate}
\end{lemma}

\begin{proof}
\begin{enumerate}[label=(\alph*)]
\item Suppose there exists an $i \in [m]$ such that $f'(x_{i}) \neq y_{i}$. Spoiler pebbles $x_{i} \mapsto f'(x_{i})$. Suppose that $P_{i}, Q_{i}$ are Sylow $p$-subgroups of $A, B$ respectively. As $x_{i} \in P_{i}$, we have that $\langle x \cdot x_{i}^{-1} \rangle$ has order coprime to $p$. However, as $f(x_{i}) \neq y_{i}$, $\langle f(x) \cdot f(x_{i})^{-1} \rangle$ has order divisible by $p$. So $|x \cdot x_{i}^{-1}| \neq |f(x) \cdot f(x_{i})^{-1}|$. Thus, Spoiler wins at the end of this round.

\item We recall that nilpotent groups are direct products of their Sylow subgroups. Furthermore, for a given prime divisor $p$, the Sylow $p$-subgroup of a nilpotent group is unique and contains all the elements whose order is a power of $p$. Thus, each Sylow subgroup of a nilpotent group is characteristic as a set. So now by (a), we may assume that $f'(x_{i}) = y_{i}$ for all $i$, or Spoiler can win with 1 pebble and 1 additional round. From what we have noted before the lemma, $y$ does not split from $Q$ iff there is an $i$ such that $y_i$ does not split from $Q_i$. Let $i \in [m]$ such that $x_{i}$ splits from $P_{i}$, but $f'(x_{i}) = y_{i}$ does not split from $Q_{i}$. Spoiler pebbles $x_{i} \mapsto f'(x_{i}) = y_{i}$. Now by \Lem{LemmaSylowSplit}, applied to $x_{i} \mapsto f'(x_{i})$, Spoiler wins with $2$ more pebbles (for a total of $3$ additional pebbles beyond the one on $x$) and $2$ more rounds (for a total of $3$ additional rounds starting from and including the round at which $f'$ was selected). \qedhere
\end{enumerate}
\end{proof}

We now show that Duplicator must select bijections that preserve both the center and commutator subgroups setwise. Here is our first application of the Rank Lemma~\ref{LemmaRank}, which was not present in \cite{BrachterSchweitzerWLLibrary}. We begin with the following standard definition.

\begin{definition}
For a group $G$ and $g \in G$, the \textit{commutator width} of $g$, denoted $\cw(g)$, is its $\{[x,y] : x,y \in G\}$-rank (see Definition~\ref{DefRank}). The commutator width of $G$, denoted $\cw(G)$, is the maximum commutator width of any element of $[G,G]$.
\end{definition}

\begin{lemma} \label{IdentifyCenterCommutator}
Let $G, H$ be finite groups of order $n$. Let $f : G \to H$ be the bijection that Duplicator selects in the Version II pebble game. 
\begin{enumerate}[label=(\alph*)]
\item If $f(Z(G)) \neq Z(H)$, then Spoiler can win with $2$ pebbles and $2$ rounds.
\item If there exist $x, y \in G$ such that $f([x,y])$ is not a commutator $[h,h']$ for any $h,h' \in H$ (that is, $\cw(f([x,y])) > 1$), then Spoiler can win with $3$ pebbles and $3$ rounds.
\item If there exists $g \in G$ such that $\cw(g) \neq \cw(f(g))$, then Spoiler can win with $4$ pebbles and $O(\log \cw(G)) \leq O(\log n)$ rounds.
\end{enumerate}
\end{lemma}

Brachter \& Schweitzer previously showed that $2$-WL Version II identifies $Z(G)$, and $3$-WL Version II identifies the commutator $[G,G]$ \cite{BrachterSchweitzerWLLibrary}. Here, using our Rank Lemma~\ref{LemmaRank} for commutator width, we obtain that $3$-WL Version II identifies the commutator in $O(\log n)$ rounds.

\begin{proof}[Proof of \Lem{IdentifyCenterCommutator}]
\noindent 
\begin{enumerate}[label=(\alph*)]
\item Let $x \in Z(G)$ such that $f(x) \not \in Z(H)$. Spoiler begins by pebbling $x \mapsto f(x)$. Let $f' : G \to H$ be the bijection that Duplicator selects at the next round. Let $y \in H$ such that $f'(x)$ and $y$ do not commute. Let $a := (f')^{-1}(y) \in G$. Spoiler pebbles $a \mapsto f'(a) = y$ and wins.

\item Spoiler pebbles $[x, y] \mapsto f([x,y])$. At the next two rounds, Spoiler pebbles $x, y$. Regardless of Duplicator's choices, Spoiler wins.

\item Apply the Rank Lemma~\ref{LemmaRank} to the set of commutators. Suppose that there exist $x, y \in G$ such that $g = [x,y]$, but for all $x', y' \in H$, $f(g) \neq [x', y']$. Then by part (b), Spoiler can win with $3$ pebbles and $3$ rounds. So by the Rank Lemma~\ref{LemmaRankIntro}, if $\cw(g) \neq \cw(f(g))$, then Spoiler can win with $4$ pebbles and $O(\log n)$ rounds.
\end{enumerate}
\end{proof}

By \Lem{IdentifyCenterCommutator}, Duplicator must select bijections that preserve the center and commutator subgroups setwise (or Spoiler can win). A priori, these bijections need not restrict to isomorphisms on the center or commutator. We note, however, that we may easily decide whether two groups have isomorphic centers, as the center is Abelian. Precisely, by \cite[Corollary~5.3]{BrachterSchweitzerWLLibrary}, $(2,1)$-WL Version II identifies Abelian groups. Note that we need an extra round to handle the case in which Duplicator maps an element of $Z(G)$ to some element not in $Z(H)$. So $(2,2)$-WL Version II both distinguishes $Z(G)$ (as a subset of $G$) and identifies its isomorphism type.

We now turn to detecting elements that split from arbitrary groups. To this end, we recall the following lemma from \cite{BrachterSchweitzerWLLibrary}.

\begin{lemma}[{\cite[Lemma~6.9]{BrachterSchweitzerWLLibrary}}]
Let $G$ be a finite group and $z \in Z(G)$. Then $z$ splits from $G$ if and only if $z[G,G]$ splits from $G/[G,G]$ and $\langle z \rangle \cap [G,G] = 1$. 
\end{lemma}

%
%
As they did \cite[Corollary~6.10]{BrachterSchweitzerWLLibrary}, we apply this lemma to show that WL can detect the set of central elements that split from an arbitrary finite group, but now we also control the number of rounds:

\begin{lemma}[{Compare rounds cf. \cite[Corollary~6.10]{BrachterSchweitzerWLLibrary}\footnote{Their statement omits the hypothesis that $x \in Z(G)$; we believe this hypothesis is crucial to the proof, but they---and we---only apply this lemma to cases where $x$ is in the center anyway.}}] \label{LemmaGeneralSplit}
Let $G, H$ be finite groups. Let $f : G \to H$ be the bijection that Duplicator selects. Suppose that $x \in Z(G)$
splits from $G$, but $f(x)$ does not split from $H$. Then Spoiler can win with $4$ pebbles and $O(\log n)$ rounds.
\end{lemma}

\begin{proof}
By \Lem{IdentifyCenterCommutator}, we have that if $x \not \in [G, G]$ but $f(x) \in [H, H]$, then Spoiler can win with $4$ pebbles and $O(\log n)$ rounds. (We note that the round count in the preceding step is really the main innovation here, as it relies on our Rank Lemma. The rest is essentially as in \cite{BrachterSchweitzerWLLibrary}.) So suppose that $x \not \in [G, G]$ and $f(x) \not \in [H, H]$. It suffices to check whether $x[G,G]$ splits from $G/[G,G]$, but $f(x)[H, H]$ does not split from $H/[H,H]$. 
By \cite[Lemma~4.11]{BrachterSchweitzerWLLibrary}, 
it suffices to consider the pebble game on $(G/[G,G], H/[H,H])$; we note that their Lemma~4.11 is a round-by-round simulation, plus $O(1)$ additional rounds, so their simulation preserves rounds up to an additive constant.
Spoiler begins by pebbling $x \mapsto f(x)$ in the game on $(G, H)$, which by the proof of \cite[Lemma~4.11]{BrachterSchweitzerWLLibrary} corresponds to Spoiler pebbling $x[G,G] \mapsto f(x)[H,H]$ in the game on $(G/[G,G], H/[H,H])$. We now apply \Lem{LemmaSplitAbelian} to $G/[G,G]$ and $H/[H,H]$. \qedhere
\end{proof}

We now consider splitting in two special cases.

\begin{lemma}[{\cite[Lemma~6.11]{BrachterSchweitzerWLLibrary}}]
Let $U \leq G$ and $x \in Z(G) \cap U$. If $x$ splits from $G$, then $x$ splits from $U$.
\end{lemma}

\begin{lemma}[{\cite[Lemma~6.12]{BrachterSchweitzerWLLibrary}}] \label{BSLemma612}
Let $G = G_{1} \times G_{2}$, and let $z := (z_{1}, z_{2}) \in Z(G)$ be an element of order $p^{k}$ for some prime $p$. Then $z$ splits from $G$ if and only if there exists an $i \in \{1,2\}$ such that $z_{i}$ splits from $G_{i}$ and $|z_{i}| = |z|$.
\end{lemma}

We now consider what Brachter \& Schweitzer call the semi-Abelian case: that is, where our groups have the form $H \times A$, where $H$ has no Abelian direct factors and $A$ is Abelian. (Of course every finite group can be written this way, possibly with $H$ or $A$ trivial; the point of the semi-Abelian case is essentially to reduce to the case where there are no (remaining) Abelian direct factors.)

\begin{theorem} [{Compare rounds cf. \cite[Lemma~6.13]{BrachterSchweitzerWLLibrary}}] 
\label{IdentifyAbelianDirectFactor}
Let $G_{1} = H \times A$, with a maximal Abelian direct factor $A$. Then the isomorphism class of $A$ is identified by $(4, O(1))$-WL Version II. That is, if $(4, O(1))$-WL fails to distinguish $G$ and $\widetilde{G}$, then $\widetilde{G}$ has a maximal Abelian direct factor isomorphic to $A$. 
\end{theorem}

\begin{proof}
We adapt the proof of \cite[Lemma~6.13]{BrachterSchweitzerWLLibrary} to control for rounds. Let $\widetilde{G}$ be a group such that $(4, O(1))$-WL Version II fails to distinguish $G$ and $\widetilde{G}$. By \Lem{IdentifyCenterCommutator} and the subsequent discussion, we may assume that $Z(G) \cong Z(\widetilde{G})$ using $(2,2)$-WL Version II. As Abelian groups are direct products of their Sylow subgroups, it follows that $Z(G)$ and $Z(\widetilde{G})$ have isomorphic Sylow subgroups. Write $\widetilde{G} = \widetilde{H} \times \widetilde{A}$, where $\widetilde{A}$ is the maximal Abelian direct factor. As $Z(G) \cong Z(\widetilde{G})$, we write $Z$ for the Sylow $p$-subgroup of $Z(G) \cong Z(\widetilde{G})$. Consider the primary decomposition of $Z$:
\[
Z := Z_{1} \times \ldots \times Z_{m},
\]
where $Z_{i} \cong (\mathbb{Z}/p^{i}\mathbb{Z})^{e_{i}}$, for $e_{i} \geq 1$. For each $i \in [m]$, there exist subgroups $H_{i} \leq Z(H)$ and $A_{i} \leq A$ such that $Z_{i} \cong H_{i} \times A_{i}$. Similarly, there exist $\widetilde{H_{i}} \leq \widetilde{H}$ and $\widetilde{A_{i}} \leq \widetilde{A}$ such that $Z_{i} \cong \widetilde{H_{i}} \times \widetilde{A_{i}}$. As $Z(G) \cong Z(\widetilde{G})$, we have that $H_{i} \times A_{i} \cong \widetilde{H_{i}} \times \widetilde{A_{i}}$. It suffices to show that each $A_{i} \cong \widetilde{A_{i}}$. As $H$ does not have any Abelian direct factors, we have by \cite[Lemma~6.12]{BrachterSchweitzerWLLibrary} (reproduced as \Lem{BSLemma612} above) that a central element $x$ of order $p^{i}$ splits from $G$ if and only if the projection of $x$ onto $A_{i}$, denoted $A_{i}(x)$, has order $p^{i}$. The same holds for $\widetilde{G}$ and the $\widetilde{A_{i}}$ factors. By \Lem{LemmaSplitAbelian}, we may assume that Duplicator selects bijections $f : G \to \widetilde{G}$ such that if $x \in Z(G)$ splits from $Z_{i}$, then $f(x)$ splits from $f(Z_{i})$. Otherwise, Spoiler pebbles $x \mapsto f(x)$, and then following \Lem{LemmaSplitAbelian}(b), wins with $3$ additional pebbles (for a total of $4$ pebbles) and $3$ additional rounds (for a total of $4$ rounds). The result follows.
\end{proof}

We now recall the definition of a component-wise filtration, introduced by Brachter \& Schweitzer \cite{BrachterSchweitzerWLLibrary} to control the non-Abelian part of a direct product.

\begin{definition} [{\cite[Def.~6.14]{BrachterSchweitzerWLLibrary}}]
Let $G = L \times R$. A \textit{component-wise filtration} of $U \leq G$ with respect to $L$ and $R$ is a chain of subgroups $\{1\} = U_{0} \leq \cdots \leq U_{r} = U$, such that for all $i \in [r-1]$, we have that $U_{i+1} \leq U_{i}(L \times \{1\})$ or $U_{i+1} \leq U_{i}(\{1\} \times R)$. 

Now fix $J \in \{I, II\}$. 
We say that $U_{i}$ is $(k,r)$-WL$_{J}$-detectable if, whenever $g \in U_{i}, g' \not \in U_{i}$, $(k,r)$-WL Version $J$ will assign different colors to $(g, \ldots, g)$ vs. $(g', \ldots, g')$. The filtration is $(k,r)$-WL$_{J}$-detectable, provided all subgroups in the chain are $(k,r)$-WL$_{J}$-detectable.
\end{definition}

Brachter \& Schweitzer previously showed \cite[Lemma~6.15]{BrachterSchweitzerWLLibrary} that there exists a component-wise filtration of $Z(G)$ with respect to $H$ and $A$ that is detectable by $4$-WL Version I. We extend this result to control for rounds. The proof that such a filtration exists is identical to that of \cite[Lemma~6.15]{BrachterSchweitzerWLLibrary}; we get a bound on the rounds using our \Lem{LemmaGeneralSplit}, which is a round-controlled version of their Corollary~6.10. For completeness, we indicate the needed changes here.

\begin{lemma}[Compare rounds cf. {\cite[Lemma~6.15]{BrachterSchweitzerWLLibrary}}] \label{LemmaFiltration}
Let $G = H \times A$, with maximal Abelian direct factor $A$. There exists a component-wise filtration of $Z(G)$, with respect to $H$ and $A$, $\{1\} = U_{0} \leq \cdots \leq U_{r} = Z(G)$ that is $(4, O(\log n))$-WL$_{II}$-detectable.
\end{lemma}

\begin{proof}[{Proof outline, highlighting differences from \cite{BrachterSchweitzerWLLibrary}.}]
We follow the strategy of \cite[Lemma~6.15]{BrachterSchweitzerWLLibrary}. First, we highlight our key changes. Their proof has two parts: 
the fact that central $e$-th powers are detectable, and their Corollary~6.10. Using our \Lem{LemmaGeneralSplit} in place of their Corollary~6.10, we get 4 pebbles and $O(\log n)$ rounds, so all that is left to handle is central $e$-th powers. Suppose Duplicator selects a bijection $f : G \to H$ where $g = x^{e}$ for some $x \in Z(G)$ and $f(g)$ is not a central $e$th power. By \Lem{IdentifyCenterCommutator}(a), Duplicator must map the center to the center or Spoiler can win with 2 pebbles in 2 rounds, 
so we need only handle the condition of being an $e$-th power. At the first round, Spoiler pebbles $g \mapsto f(g)$. At the next round, Spoiler pebbles $x$ and wins. Thus Spoiler can win with 2 pebbles in 2 rounds. 
\end{proof}

\begin{proof}
We now turn to the details. Let $p_{1} < \cdots < p_{n}$ be the primes dividing $|G|$, and write $Z_{p_{i}}$ for the Sylow $p_{i}$-subgroup of $Z(G) = Z(H) \times A$. We now turn to recalling the characteristic filtration defined by Brachter and Schweitzer. Suppose we already have a component-wise filtration of
\[
U = Z_{p_{1}} \times \cdots \times Z_{p_{i-1}} \times \{z \in Z_{p_{i}} : |z| < p_{i}^{m}\}
\]
with respect to $H$ and $A$, that is $(4, O(\log n))$-WL$_{II}$-detectable. We now seek to extend such a component-wise filtration to $U\{z \in Z_{p_{i}} : |z| \leq p_{i}^{m} \}$. To simplify notation, let $p := p_{i}$, and let $N$ be maximal such that $p^{N}$ divides $|Z(G)|$. 

Set $V_{0} := \{ z \in Z_{p} : |z| < p_{i}^{m} \}$. For $j \geq 1$, define: 
\[
V_{j} := \langle \{ z^{p^{N-j}} : z \in Z_{p}, |z^{p^{N-j}}| \leq p^m \} \rangle V_{j-1}.
\]

\noindent Now define:
\[
W_{j} := \langle \{ z^{p^{N-j+m}} : z \in Z_{p}, |z| \leq p^{N-j+m}, \text{ and } z \text{ does not split from G} \} \rangle V_{j-1}.
\]

\noindent By construction, we have that:
\[
U = UV_{0} \leq UW_{1} \leq UV_{1} \leq \cdots \leq UW_{n} \leq UV_{n} = U\{z \in Z_{p_{i}} : |z| \leq p_{i}^{m} \}.
\]
Brachter and Schweitzer established in the proof of \cite[Lemma~6.15]{BrachterSchweitzerWLLibrary} that $V_{j-1} \leq W_{j} \leq V_{j}$ for all $j \geq 1$. By the discussion in the outline above, 
we have that each $W_j, V_j$ are $(4, O(\log n))$-WL$_{II}$ detectable in $G$. The result now follows.
\end{proof}

We now show that in the semi-Abelian case $G = H \times A$, with maximal Abelian direct factor $A$, the WL dimension of $G$ depends on the WL dimension of $H$.

\begin{lemma}[{Compare rounds to \cite[Lemma~6.16]{BrachterSchweitzerWLLibrary}}] \label{SemiAbelianDirectProduct}
Let $G = H \times A$ and $\widetilde{G} = \widetilde{H} \times \widetilde{A}$, where $A$ and $\widetilde{A}$ are maximal Abelian direct factors. Let $k \geq 5$ and $r := r(n)$. If $(k, r + O(\log n))$-WL Version II fails to distinguish $G$ and $\widetilde{G}$, then $(k-1, r)$-WL Version II fails to distinguish $H$ and $\widetilde{H}$.
\end{lemma}

\begin{proof}
By \Thm{IdentifyAbelianDirectFactor}, we may assume that $A \cong \widetilde{A}$ (otherwise, $(4, O(1))$-WL Version II distinguishes $G$ from $\widetilde{G}$). Consider the component-wise filtrations from the proof of \cite[Lemma~6.15]{BrachterSchweitzerWLLibrary}, $\{1\} = U_{0} \leq \cdots \leq U_{r} = Z(G)$ with respect to the decomposition $G = H \times A$ and $\{1\} = \widetilde{U_{0}} \leq \cdots \leq \widetilde{U_{r}} = Z(\widetilde{G})$ with respect to the decomposition $\widetilde{G} = \widetilde{H} \times \widetilde{A}$. 

Let $V_{i}, W_{i}, \widetilde{V_{i}}, \widetilde{W_{i}}$ be as defined in the proof of \cite[Lemma~6.15]{BrachterSchweitzerWLLibrary} and recalled in the proof of \Lem{LemmaFiltration}. We showed in the proof of \Lem{LemmaFiltration} that for any bijection $f : G \to \widetilde{G}$ Duplicator selects, $f(V_{i}) = \widetilde{V_{i}}$ and $f(W_{i}) = \widetilde{W_{i}}$, or Spoiler may win with $4$ pebbles and $O(\log n)$ rounds.

In the proof of \cite[Lemma~6.16]{BrachterSchweitzerWLLibrary}, Brachter \& Schweitzer established that for all $1 \neq x \in Z(H) \times \{1\}$ and all $1 \neq  y \in \{1\} \times A$, $\min\{ i : x \in U_{i} \} \neq \min\{ i : y \in U_{i}\}$. 
Furthermore, by \cite[Lemma~4.14]{BrachterSchweitzerWLLibrary}, we may assume that Duplicator selects bijections at each round that respect the subgroup chains and their respective cosets, without altering the number of rounds (their proof is round-by-round). It follows that whenever $g_{1}g_{2}^{-1} \in Z(H) \times \{1\}$, we have that $f(g_{1})f(g_{2})^{-1} \not \in \{1\} \times A$.

Furthermore, Brachter \& Schweitzer also showed in the proof of \cite[Lemma~6.16]{BrachterSchweitzerWLLibrary} that Duplicator must map $H \times \{1\}$ to a system of representatives modulo $\{1\} \times \widetilde{A}$. Thus, Spoiler can restrict the game to $H \times \{1\}$. Now if $(k-1, r)$-WL Version II distinguishes $H$ and $\widetilde{H}$, then Spoiler can ultimately reach a configuration $((h_{1}, 1), \ldots, (h_{k-1}, 1)) \mapsto (x_{1}, a_{1}), \ldots, (x_{k-1}, a_{k-1})$ such that the induced configuration over $(G/(\{1\} \times A), \widetilde{G}/(\{1\} \times \widetilde{A}))$ fulfills the winning condition for Spoiler. That is, considered as elements of $G/(\{1\} \times A)$ (resp., $\widetilde{G}/(\{1\} \times \widetilde{A})$), the map $(h_1, \dotsc, h_{k-1}) \mapsto (x_1, \dotsc, x_{k-1})$ does not extend to an isomorphism. This implies that the pebbled map $((h_{1}, 1), \ldots, (h_{k-1}, 1)) \mapsto ((x_{1}, a_{1}), \ldots, (x_{k-1}, a_{k-1}))$ in the original groups (rather than their quotients) does not extend to an isomorphism. For suppose $f$ is any bijection extending the pebbled map. By the above, without loss of generality, $f$ maps $H \times \{1\}$ to a system of coset representatives of $\{1\} \times \widetilde{A}$, that is, if Duplicator can win, Duplicator can win with such a map. Let $\overline{f}$ be the induced bijection on the quotients $G/(\{1\} \times A) \to \tilde{G} / (\{1\} \times \tilde{A} )$. Since the pebbled map on the quotients does not extend to an isomorphism, there is a word $w$ such that $\overline{f}(w(h_1, \dotsc, h_{k-1})) \neq w(x_1, \dotsc, x_{k-1})$. But then when we consider $f$ restricted to $H \times \{1\}$, we find that $f(w((h_1,1), \dotsc, (h_{k-1},1))) = f( (w(h_1, \dotsc, h_{k-1}), 1) ) \neq (w(x_1, \dotsc, x_{k-1}), w(a_1, \dotsc, a_{k-1}))$, because their $H$ coordinates are different.
\end{proof}

Lem \ref{SemiAbelianDirectProduct} yields the following corollary.

\begin{corollary} \label{cor:DasSharmaGeneralized}
Let $G = H \times A$, where $H$ is identified by $(O(1), O(\log n))$-WL and does not have an Abelian direct factor, and $A$ is Abelian. Then $(O(1), O(\log n))$-WL identifies $G$. In particular, isomorphism testing of $G$ and an arbitrary group $\widetilde{G}$ is in $\textsf{TC}^{1}$.
\end{corollary}

\begin{proof}
By \Lem{SemiAbelianDirectProduct}, as $H$ is identified by $(O(1), O(\log n))$-WL, we have that $G$ is identified by $(O(1), O(\log n))$-WL. As only $O(\log n)$ rounds are required, we apply the parallel WL implementation due to Grohe \& Verbitsky \cite{GroheVerbitsky} to obtain the bound of $\textsf{TC}^{1}$ for isomorphism testing. 
\end{proof}

\begin{remark}
Das \& Sharma \cite{DasSharma} previously exhibited a nearly-linear time algorithm for groups of the form $H \times A$, where $H$ has size $O(1)$ and $A$ is Abelian. \Cor{cor:DasSharmaGeneralized} considers a more general family of groups, including when $H$ is $O(1)$-generated. 
In addition to extending to a larger family of groups, for groups $H \times A$ where $|H|=O(1)$ and $A$ is Abelian, although \Cor{cor:DasSharmaGeneralized} does not improve the runtime relative to Das \& Sharma's result (in fact, the runtime is worse), it does establish a new parallel upper bound for isomorphism testing.
\end{remark}

\subsection{Additional preliminaries}

We now introduce some additional preliminaries.

\begin{definition}
Let $G_{1}, G_{2}$ be groups, and let $Z_{i} \leq Z(G_{i})$ be central subgroups. Given an isomorphism $\varphi : Z_{1} \to Z_{2}$, the \textit{central product} of $G_{1}$ and $G_{2}$ with respect to $\varphi$ is:
\[
G_{1} \times_{\varphi} G_{2} = (G_{1} \times G_{2}) / \{ (g, \varphi(g^{-1})) : g \in Z_{1} \}.
\]

A group $G$ is the \textit{(internal) central product} of subgroups $G_{1}, G_{2} \leq G$, provided that $G = G_{1}G_{2}$ and $[G_{1}, G_{2}] = \{1\}$. 
\end{definition}

In general, a group may have several inherently different central decompositions. On the other hand, indecomposable \emph{direct product} decompositions are unique in the following sense.

\begin{lemma}[{See, e.g., \cite[3.3.8]{Robinson1982}}] \label{LemmaDirectProdDecomp}
Let $G = G_{1} \times \ldots \times G_{m} = H_{1} \times \ldots \times H_{n}$ be two direct decompositions of $G$ into directly indecomposable factors. Then $n = m$, and there exists a permutation $\sigma \in \text{Sym}(n)$ such that for all $i$, $G_{i} \cong H_{\sigma(i)}$ and $G_{i}Z(G) = H_{\sigma(i)}Z(G)$.
\end{lemma}

By the preceding lemma, the multiset of subgroups $\{\!\!\{ G_{i}Z(G) \}\!\!\}$ is invariant under automorphisms of $G$. 

\begin{definition}[{\cite[Def.~6.3]{BrachterSchweitzerWLLibrary}}] \label{def:directlyInduced}
We say that a central decomposition $\{H_1, H_2\}$ of $G = H_{1}H_{2}$ is \textit{directly induced} if there exist subgroups $K_{i} \leq H_{i}$ ($i = 1, 2$) such that $G = K_{1} \times K_{2}$ and $H_{i} = K_{i}Z(G)$.
\end{definition}

\begin{lemma}[{\cite[Lemma~6.4]{BrachterSchweitzerWLLibrary}}]
Let $k \geq 4, r \geq 1$. Let $G_{1}, G_{2}, H_{1}, H_{2}$ be groups such $G_{i}$ and $H_{i}$ are not distinguished by $(k, r)$-WL. Then $G_{1} \times G_{2}$ and $H_{1} \times H_{2}$ are not distinguished by $(k, r)$-WL.
\end{lemma}

\begin{remark}
The statement of \cite[Lemma~6.4]{BrachterSchweitzerWLLibrary} does not mention rounds; however, the proof holds when considering rounds.
\end{remark}

\subsection{General case}

Following the strategy in \cite{BrachterSchweitzerWLLibrary}, we reduce the general case to the semi-Abelian case, in which groups are of the form $H \times A$, where $A$ is Abelian and $H$ does not have any Abelian direct factors. Consider a direct decomposition $G = G_{1} \times \ldots \times G_{d}$, where each $G_{i}$ is directly indecomposable. The multiset of subgroups $\{\!\!\{ G_{i}Z(G) \}\!\!\}$ is independent of the choice of decomposition (\Lem{LemmaDirectProdDecomp}). We first show that Weisfeiler--Leman detects $\bigcup_{i} G_{i}Z(G)$. Next, we utilize the fact that the connected components of the non-commuting graph on $G$, restricted to $\bigcup_{i} G_{i}Z(G)$, correspond to the subgroups $G_{i}Z(G)$. 

\begin{definition}
For a group $G$, the \textit{non-commuting graph} $X_{G}$ has vertex set $G \backslash Z(G)$,\footnote{We have changed our definition of the non-commuting graph relative to that in \cite{BrachterSchweitzerWLLibrary} to only have vertices corresponding to elements outside the center, in agreement with \cite{ABDOLLAHI2006468} and prior literature on the non-commuting graph going back to when Erd\H{o}s first defined it; the rest of the results go through with this modified definition as stated.} and an edge $\{g, h\}$ precisely when $[g, h] \neq 1$.
\end{definition}

\begin{proposition}[{\cite[Proposition~2.1]{ABDOLLAHI2006468}}]
If $G$ is non-Abelian, then $X_G$ is connected.
\end{proposition}

Our goal now is to first construct a canonical central decomposition of $G$ that is detectable by WL. This decomposition will serve to approximate $\bigcup_{i} G_{i}Z(G)$ from below.

\begin{definition}[{\cite[Definition~6.19]{BrachterSchweitzerWLLibrary}}] \label{MaximalCentralizerDef}
Let $G$ be a finite, non-Abelian group. Let $M_{1}$ be the set of non-central elements $g$ whose centralizers $C_{G}(g)$ have maximal order among all non-central elements. For $i \geq 1$, define $M_{i+1}$ to be the union of $M_i$ and the set of non-central\footnote{We believe this to be the correct definition, by which the remaining results go through. See Footnote~\ref{fn:noncentral}.} elements $g \in G \setminus \langle M_{i} \rangle$ that have maximal centralizer order $|C_{G}(g)|$ amongst the non-central elements in $G \setminus \langle M_{i} \rangle$. Let $M := M_{\infty}$ be the stable set resulting from this procedure.

Consider the subgraph $X_{G}[M]$, and let $X_{1}, \ldots, X_{m}$ be the connected components. Set $N_{i} := \langle X_{i} \rangle$. We refer to $N_{1}, \ldots, N_{m}$ as the \textit{non-Abelian components} of $G$. 
\end{definition}

Brachter \& Schweitzer previously established the following \cite{BrachterSchweitzerWLLibrary}.

\begin{lemma}[{\cite[Lemma~6.20]{BrachterSchweitzerWLLibrary}}] \label{LemBSCanonicalDecomposition}
In the notation of Definition \ref{MaximalCentralizerDef}, we have the following.
\begin{enumerate}[label=(\alph*)]
\item $M$ is $3$-WL$_{II}$-detectable.
\item $G = N_{1} \cdots N_{m}$ is a central decomposition of $G$. For all $i$, $Z(G) \leq N_{i}$ and $N_{i}$ is non-Abelian. In particular, $M$ generates $G$.\footnote{\label{fn:noncentral}We clarify a minor confusion in the definitions and proof from \cite{BrachterSchweitzerWLLibrary}. In the proof of their 6.20(2), they write that by definition the process of defining $M$ doesn't stop until $G = \langle M \rangle$, but they also say that by construction, no central elements are added to $M$, despite their definition allowing central elements in the construction of $M_i$ for $i > 1$. We believe only allowing non-central elements in the construction of $M_i$ to be correct and make the rest of the results go through; in particular, we cannot allow central elements in $M_i$ for $i \geq 1$ or some of the ``non-abelian components'' would actually be cyclic subgroups of the center. Nonetheless, although we do not get $G = \langle M \rangle$ by definition/construction anymore, their argument that $M_i = M_iZ(G)$ for all $i$ is still correct, and thus $Z(G) \leq \langle M_1 \rangle \leq \langle M \rangle$, and the process does stop by definition once $G \backslash \langle M_i \rangle \leq Z(G)$. Therefore indeed when the process stops we have $G = \langle M \rangle$. }

\item If $G = G_{1} \times \ldots \times G_{d}$ is an arbitrary direct decomposition, then for each $i \in [m]$, there exists a unique $j \in [d]$ such that $N_{i} \subseteq G_{j}Z(G)$. Collect all such $i$ for one fixed $j$ in an index set $I_{j}$. Then
\[
N_{j_{1}}N_{j_{2}} \cdots N_{j_{\ell}} = G_{j}Z(G),
\]
where $I_{j} = \{ j_{1}, \ldots, j_{\ell} \}$. 
\end{enumerate}
\end{lemma}

We note that \Lem{LemBSCanonicalDecomposition}(b)--(c) are purely group-theoretic statements. For our purposes, it is necessary, however, to adapt \Lem{LemBSCanonicalDecomposition}(a) to control for rounds. This is our second use of our Rank Lemma~\ref{LemmaRank}, this time applied to the set $M_i$ from Definition \ref{MaximalCentralizerDef}. We in fact use the Rank Lemma a few times as stated, and in the course of the proof we prove a ``relative modulo $\langle M_{i-1} \rangle$'' version of the Rank Lemma as well; although such a relative version of the rank lemma modulo a characteristic subgroup would be interesting to state in general, we use several features specific to the case at hand that dissuade us from doing so here.

\begin{lemma} \label{LemBSCanonicalDecompositionA}
Let $G$ and $H$ be finite non-Abelian groups, let $M_{i,G}$ (resp., $M_{i,H}$) denote the sets from Definition \ref{MaximalCentralizerDef} for $G$ (resp., $H$). Let $f : G \to H$ be a bijection that Duplicator selects. If for some $i$, $\rk_{M_{i,G}}(g) \neq \rk_{M_{i,H}}(f(g))$, then Spoiler can win with $3$ pebbles and $O(\log n)$ rounds.
\end{lemma}

\begin{proof}
Let $M_{i, G}, M_{G}$ be the sets in $G$ as in Definition \ref{MaximalCentralizerDef}, and let $M_{i, H}, M_{H}$ be the corresponding sets in $H$. We show that $f(M_{i, G}) = M_{i, H}$ and $f(\langle M_{i,G}\rangle) = \langle M_{i,H} \rangle$ for all $i$. These statements imply that $f(M_{G}) = M_{H}$. The proof proceeds by induction over $i$ to show that if the preceding does not hold, then Spoiler can win with 3 pebbles in $\log |\langle M_{i,G} \rangle|) + O(i)$ rounds. Within the inductive step for each $i$, we prove a ``quotient'' version of the  Rank Lemma (Lemma~\ref{LemmaRank}), applied to $M_i$-rank modulo $\langle M_{i-1} \rangle$. Toward that end, we note that each $M_i$ is indeed closed under taking inverses,  since $C_G(g) = C_G(g^{-1})$.


We first handle the case of $M_1$. If $|C_{G}(g)| \neq |C_{H}(f(g))|$, then Duplicator may win with $2$ pebbles and $2$ rounds. Without loss of generality, suppose that $|C_{G}(g)| > |C_{H}(f(g))|$. Spoiler pebbles $g \mapsto f(g)$. Let $f' : G \to H$ be the bijection that Duplicator selects at the next round. Now there exists $x \in C_{G}(g)$ such that $f'(x) \not \in C_{H}(f(g))$. Spoiler pebbles $x \mapsto f'(x)$ and wins immediately.

Thus $M_{1,G}$ is distinguished by 2 pebbles in 2 rounds. (Since $M_1$ is a characteristic subset of any group, and we have shown that for any two finite non-Abelian groups $G,H$, Duplicator must map $M_{1,G}$ to $M_{1,H}$ or Spoiler wins using 2 pebbles in 2 rounds, we may apply the latter statement to the case $G = H$ in order to satisfy the definition of a subset being ``distinguished.'' We will use this same kind of argument in the future without further comment.) By the Rank Lemma~\ref{LemmaRank} applied to $M_1$-rank, we get that the subgroup $\langle M_{1,G} \rangle$ is distinguished by 2 pebbles in $\log |\langle M_{1,G} \rangle| + O(1)$ rounds (if Duplicator ever selects a bijection $f$ with $f(\langle M_{1,G} \rangle) \neq \langle M_{1,H} \rangle$, Spoiler can win with 3 pebbles in $\log |\langle M_{1,G} \rangle| + O(1)$ rounds).

\newcommand{\relrk}{\text{relrk}}
Now let $i > 1$, and suppose that $M_{i-1,G}$ and $\langle M_{i-1,G} \rangle$ are both distinguished by 2 pebbles in at most $\log |\langle M_{i-1,G} \rangle | + c(i-1)$ rounds for some constant $c$. (Here we give $c$ a name rather than hiding it with big-Oh notation, because we will have various $O(\cdot)$'s floating around and want to keep careful track of which are which.) 

First we claim that $M_{i,G}$ is distinguished by 2 pebbles in $\log |\langle M_{i-1,G} \rangle | + c(i-1)$ (n.b. the subscript $i-1$) rounds. Suppose the pebble game is started from the position $g \mapsto h$ for some $g \in M_{i,G}$ and $h \notin M_{i,H}$. If $h \in \langle M_{i-1,H} \rangle$, then by reversing the roles of $G$ and $H$ and using our inductive assumption, we get that Spoiler can win with 2 pebbles in $\log |\langle M_{i-1,H} \rangle | + c(i-1)$ rounds. If $h \notin \langle M_{i-1,H} \rangle$, then we have $h$ is not in $M_{i,H}$ iff $|C_H(h) | \neq |C_G(g)|$. In the latter case, as above, Spoiler can win with 2 pebbles and 2 rounds. Thus in either case, as long as $c \geq 2$ (which we can indeed choose), Spoiler can win with 2 pebbles in $\log | \langle M_{i-1,G} \rangle | + c(i-1)$ rounds, and therefore $M_{i,G}$ is distinguished using 2 pebbles in $\log | \langle M_{i-1,G} \rangle | + c(i-1)$ rounds. 

In the remainder of the proof we will show that the subgroup $\langle M_{i,G} \rangle$ is distinguished by 2 pebbles in $\log|\langle M_{i,G} \rangle| + ci$ rounds.\footnote{\label{fn:gains} We note that if we applied the Rank Lemma here directly, we would instead get 2 pebbles and $\log|\langle M_{i,G} \rangle| + O(1)$ plus the rounds to distinguish $\langle M_{i-1,G} \rangle$, and when applied inductively this would lead to a bound of $O(\log^2 n)$ instead of $O(\log n)$ rounds.} For this part of the proof, we define the \emph{relative rank} of $g \in G$ as the minimum $r$ such that $g$ can be written as $g = g_1 g_2 \dotsc g_r m$ where $g_j \in M_{i,G}$ (just the set, not $\langle M_{i,G} \rangle$) for all $j=1,\dotsc,r$ and $m \in \langle M_{i-1, G} \rangle$ (the subgroup, not just the set). We define elements of $\langle M_{i-1,G} \rangle$ to have relative rank 0. For $g \notin \langle M_{i,G} \rangle$ we define the relative rank to be $\infty$; it is readily seen that an element has finite relative rank iff it is in $\langle M_{i,G} \rangle$. We use the notation $\relrk(g)$ for the relative rank of $g$. We define relative rank in $H$ analogously, using $M_{i,H}$ and $M_{i-1,H}$ in place of $M_{i,G}$ and $M_{i-1,G}$, respectively; we re-use the same notation $\relrk(h)$ for relative rank in $H$.

We will need the following two facts about relative rank:
\begin{enumerate}

\item[R1] For all $x, y \in G$ (resp. $H$), we have $\relrk(xy) \leq \relrk(x) + \relrk(y)$. If either relative rank is infinite the conclusion is vacuously true. Otherwise, suppose $x = x_1 \dotsb x_r m$ and $y = y_1 \dotsb y_s m'$ with $m,m' \in \langle M_{i-1} \rangle$; $x_j,y_j \in M_{i}$ for all $j$; and $r = \relrk(x), s = \relrk(y)$. Then $xy = x_1 \dotsb x_r y_1 \dotsb y_s (y^{-1}my) m'$, where we have $y^{-1} m y$ is in $\langle M_{i-1} \rangle$ since the latter is a normal subgroup.

\item[R2] For all $x \in \langle M_i \rangle$, we have $\relrk(x) \leq |\langle M_i \rangle| / |\langle M_{i-1} \rangle|$. (This is where the crucial gains come from, as the $M_i$-rank could in principle be as large as $|\langle M_i \rangle|$; see also Footnote~\ref{fn:gains}.) For suppose we have $x = x_1 \dotsb x_r m$ with $m \in \langle M_{i-1} \rangle$ and $r > |\langle M_i \rangle| / |\langle M_{i-1} \rangle|$. Then among the partial products $x_1, x_1 x_2, x_1 x_2 x_3, \dotsc, x_1 \dotsb x_r$, some coset of $\langle M_{i-1} \rangle$ is encountered twice; suppose that $x_1 \dotsb x_j$ and $x_1 \dotsb x_k$ with $j < k$ are in the same coset of $\langle M_{i-1} \rangle$. Then we can rewrite $x_1 \dotsb x_k$ as $x_1 \dotsb x_j m'$ for some $m' \in \langle M_{i-1} \rangle$. Finally, since $\langle M_{i-1} \rangle$ is normal, we then get $m' x_{k+1} \dotsb x_r = x_{k+1} \dotsb x_r m''$ for a unique $m'' \in \langle M_{i-1} \rangle$, and thus 
\[
x = x_1 \dotsb x_k x_{k+1} \dotsb x_r m = (x_1 \dotsb x_j m') x_{k+1} \dotsb x_r m = x_1 \dotsb x_j x_{k+1} \dotsb x_r m'' m,
\]
so $\relrk(x) < r$. 
\end{enumerate}


Returning to our pebble game starting from the position $g \mapsto h$, without loss of generality (by swapping the roles of $G$ and $H$ if needed), we may assume that $\relrk(g) < \relrk(h)$. If $\relrk(g)=0$ (i.\,e., $g \in \langle M_{i-1,G} \rangle$) then by the inductive hypothesis, then Spoiler can win with 2 additional pebbles in $\log |\langle M_{i-1, G} \rangle| + c(i-1)$ rounds.

On the other hand, if $\relrk(g) > 0$, let $f \colon G \to H$ be the bijection that Duplicator selects at the next round. For brevity let $q = \relrk(g)$. Note that this case works even if $\relrk(h)=\infty$, for then by our assumption that $\relrk(g) < \relrk(h)$, we still have $q = \relrk(g)$ is finite. Now write $g = g_{1} \cdots g_{q} m$, where for each $j$, $g_j \in M_{i,G}$ and we have $m \in \langle M_{i-1,G} \rangle$. For $1 \leq j \leq k \leq q$, write $g[j, \ldots, k] := g_{j} \cdots g_{k}$; continuing this notation, write $g[j,\ldots,q+1] = g_j g_{j+1} \dotsc g_q m$ (that is, we essentially consider $m$ to be at index $q+1$). We consider the following cases.
\begin{itemize}
\item \textbf{Case 1:} Suppose first that for all $x \in G$ with $\relrk(x) < q$, we have $\relrk(x) = \relrk(f(x))$. In this case, Spoiler pebbles $g[2, \ldots, q+1] \mapsto f(g[2, \ldots, q+1])$. Let $f' : G \to H$ be the bijection that Duplicator selects at the next round; Spoiler now pebbles $g_1 \mapsto f'(g_1)$. 

If $\relrk(g_{1}) = \relrk(f'(g_{1})) = 1$, then $f'(g_{1}) \cdot f(g[2, \ldots, q+1]) \neq h$, since $\relrk(f'(g_1)) = 1$ and $\relrk(f(g[2, \ldots, q+1])) \leq q-1$, so their product has rank at most $q=\relrk(g) < \relrk(h)$ by fact R1 above. Spoiler now wins since the map on the pebbled elements $(g, g[2,\ldots,q+1], g_1) \mapsto (h, f(g[2, \ldots, q+1]), f'(g_1))$ does not satisfy the winning condition for Spoiler in the Version I (and hence, the Version II) game. 

On the other hand, suppose instead that $1=\relrk(g_{1}) < \relrk(f'(g_{1}))$. Note that $\relrk(f'(g_1)) > 1$ implies that $f'(g_1)$ is not in $M_{i,H}$ (relative rank is $\leq 1$ iff an element is in $M_{i,H} \cdot \langle M_{i-1,H} \rangle$). But since $g_1$ is, by definition, an element of $M_{i,G}$, from the fact that $M_{i,G}$ is distinguished, we have that Spoiler can win with $2$ additional pebbles and $\log|\langle M_{i-1, G}\rangle| + c(i-1)$ additional rounds. Note that only 2 pebbles are needed rather than 3, since the original pebble on $g$ can be reused at this point.

\item \textbf{Case 2:} Suppose instead that the hypothesis of Case 1 is not satisfied. Then there exists some $x \in \langle M_{i,G} \rangle$ with $\relrk(x) \neq \relrk(f(x))$ and $\relrk(x) < q = \relrk(g)$. We will now do a binary-search-like procedure to eventually reach either (a) an element of $M_{i,G}$ that is mapped outside of $M_{i,H}$ or (b) an element of $\langle M_{i-1,G} \rangle$ that is mapped outside of $\langle M_{i-1,H} \rangle$. To make this clearer, we introduce two additional parameters as in binary search. Let $lo = 1$ and $hi=q+1$. In the next two rounds, Spoiler pebbles $x := g[lo, \ldots, lo + \lceil (hi-lo)/2 \rceil-1]$ and $y := g[lo + \lceil (hi-lo)/2 \rceil, \ldots, hi]$. Note that $\relrk(x) = \lceil (hi-lo)/2 \rceil$ and $\relrk(g) = \relrk(x) + \relrk(y)$, for if $x$ or $y$ could be written as a shorter word in $M_{i,G}$, then so could $g$, by fact R1. 

Let $f' : G \to H$ be the next bijection that Duplicator selects. If $h \neq f'(x)f'(y)$ then Spoiler immediately wins. On the other hand, if $h = f'(x) f'(y)$, 
then either
\[
\relrk(x) < \relrk(f'(x)) \qquad \text{ or } \qquad \relrk(y) < \relrk(f'(y)),
\]
since $\relrk(h) > \relrk(g) = \relrk(x) + \relrk(y)$. 

Without loss of generality, suppose that $\relrk(x) < \relrk(f'(x))$. Spoiler then updates $hi$ to $lo + \lceil (hi-lo)/2 \rceil-1$, and as $x \mapsto f'(x)$ is already pebbled, Spoiler iterates on this strategy as in binary search, reusing the pebbles on $g$ and $y$ in the next iteration, with $x$ now playing the role that $g$ played in the first iteration. Thus, iterating this procedure uses at most $3$ pebbles (the one that started on $g$ plus 2 more). After at most $\log q + O(1)$ iterations, we reach the case where $hi=lo+1$, and we have pebbled either (a) an element $z \mapsto f''(z)$ such that $z \in M_{i,G}$ and $\relrk(f''(z)) > 1$ or (b) an element $m \mapsto f''(m)$ with $m \in \langle M_{i-1,G} \rangle$ and $\relrk(f''(m)) > 0$. At this point, there are two pebbles in play that can be reused.

By fact R2, every element of $\langle M_{i,G} \rangle$ has relative rank at most $|\langle M_{i,G} \rangle | / |\langle M_{i-1, G} \rangle|$, and thus only $\log \frac{|\langle M_{i,G} \rangle |}{|\langle M_{i-1, G} \rangle|} + O(1)$ rounds of the preceding iteration are required. Choose $c$ large enough so that it is greater than the preceding $O(1)$. 

Finally, since both $M_{i,G}$ and $\langle M_{i-1,G} \rangle$ are distinguished by 2 pebbles in $\log |\langle M_{i-1,G} \rangle| + c(i-1)$ rounds, and we have either (a) pebbled an element of $M_{i,G}$ mapped outside of $M_{i,H}$ or (b) pebbled an element of $\langle M_{i-1,G} \rangle$ mapped outside of $\langle M_{i-1,H} \rangle$, in either case Spoiler can reuse the aforementioned 2 pebbles to win in $\log |\langle M_{i-1, G} \rangle| + c(i-1)$ additional rounds. 

Thus $\langle M_{i,G} \rangle$ is distinguished by 2 pebbles in at most 
\[
\log \frac{|\langle M_{i,G} \rangle |}{|\langle M_{i-1, G} \rangle|} + c + \log |\langle M_{i-1, G} \rangle| + c(i-1)
\]
total rounds. The preceding expression readily simplifies to $\log |\langle M_{i,G} \rangle| + ci$, as claimed.
\end{itemize}

%
%

Note that, since the $M_i$ are constructed so that $\langle M_i \rangle$ strictly contains $\langle M_{i-1} \rangle$, the maximum value of $i$ is at most $\log|G| = \log n$. By the above induction, we then get that each $M_i$ is distinguished by 2 pebbles $\log |\langle M_i \rangle| + ci$ rounds, where the latter is at most $\log n + c \log n = O(\log n)$. Finally, we apply the Rank Lemma~\ref{LemmaRank} to get that $M_i$-rank must be preserved or Spoiler can win with 3 pebbles and $O(\log n) + \log n + O(1) = O(\log n)$ rounds.
\end{proof}


\begin{definition} \label{def:full}
Let $G = N_{1} \cdots N_{m}$ be the decomposition into non-Abelian components, and let $G = G_{1} \times \cdots \times G_{d}$ be an arbitrary direct product decomposition. We say that $x \in G$ is \textit{full} for $(G_{j_{1}}, \ldots, G_{j_{r}})$, if 
\[
\{ i \in [m] : [x, N_{i}] \neq 1 \} = \bigcup_{\ell=1}^{r} I_{j_{\ell}},
\]

\noindent where the $I_{j_{\ell}}$ are as in \Lem{LemBSCanonicalDecomposition}(c). For all $x \in G$, define $C_{x} := \prod_{[x, N_{i}] = \{1\}} N_{i}$ and $N_{x} = \prod_{[x, N_{i}] \neq \{1\}} N_{i}.$
\end{definition}

We now recall some technical lemmas from \cite{BrachterSchweitzerWLLibrary}.

\begin{remark}[{\cite[Observation~6.22]{BrachterSchweitzerWLLibrary}}]
For an arbitrary collection of indices $J \subseteq [m]$, the group elements $x \in G$ that have $C_{x} = \prod_{i \in J} N_{i}$ are exactly those elements of the form $x = z \prod_{i \in J} n_{i}$ with $z \in Z(G)$ and $n_{i} \in N_{i} \setminus Z(G)$. In particular, full elements exist for every collection of non-Abelian direct factors and any direct decomposition, and they are exactly given by products over non-central elements from the corresponding non-Abelian components.
\end{remark}

\begin{lemma}[{\cite[Lemma~6.23]{BrachterSchweitzerWLLibrary}}] \label{BSLem623}
Let $G$ be non-Abelian, and let $G = G_{1} \times \cdots \times G_{d}$ be an indecomposable direct decomposition. For all $x \in G$, we have a central decomposition $G = C_{x}N_{x}$, with $Z(G) \leq C_{x} \cap N_{x}$. The decomposition is directly induced (Def.~\ref{def:directlyInduced}) if and only if $x$ is full for a collection of direct factors of $G$.
\end{lemma}

\begin{lemma}[{Compare rounds cf. \cite[Lemma~6.24]{BrachterSchweitzerWLLibrary}}] \label{ModifiedLem624}
Let $G = G_{1} \times G_{2}$. For $k \geq 4$, assume that $(k,r)$-WL Version II detects $G_{1}Z(G)$ and $G_{2}Z(G)$. Let $H$ be a group such that $(k,r + O(\log n))$-WL Version II does not distinguish $G$ and $H$. Then for $i \in \{1,2\}$, there exist subgroups $H_{i} \leq H$ such that $H = H_1 \times H_2$ and $(k,r)$-WL does not distinguish $G_{i}Z(G)$ and $H_{i}Z(H)$.
\end{lemma}

\begin{proof}
The proof is largely identical to that of \cite[Lemma~6.24]{BrachterSchweitzerWLLibrary}. We adapt their proof to control for rounds. 

As $(k,r)$-WL detects $G_{1}Z(G)$ and $G_{2}Z(G)$, we have for any two bijections $f, f' : G \to H$ chosen by Duplicator, that $f(G_{i}Z(G)) = f'(G_{i}Z(G))$ for $i \in \{1,2\}$, or Spoiler wins with $k+1$ pebbles and $r$ rounds. It follows that there exist subgroups of $\widetilde{H_{i}} \leq H$ such that $f(G_{i}Z(G)) = \widetilde{H_{i}}$ (or Spoiler wins with $k+1$ pebbles and $r$ rounds). As $Z(G) \leq G_{i}Z(G)$, we have necessarily that $Z(H) \leq \widetilde{H_{i}}$ or Spoiler can win with 2 pebbles and 2 rounds by \Lem{IdentifyCenterCommutator}(a). Consider the decompositions $Z(G) = Z(G_{1}) \times Z(G_{2})$ and $G_{i}Z(G) = G_{i} \times Z(G_{i+1 \mod 2})$. By \Lem{BSLemma612}, we have that if $x$ splits from $Z(G)$, then $x$ also splits from $G_{1}Z(G)$ or $G_{2}Z(G)$. 

Write $\widetilde{H_{i}} = R_{i} \times B_{i}$, where $B_{i}$ is a maximal Abelian direct factor of $\widetilde{H_{i}}$.

\begin{quotation}
\noindent \textbf{Claim 1:} For all choices of $R_{i}, B_{i}$, it holds that $R_{1} \cap R_{2} = \{1\}$. Otherwise, Spoiler can win with $4$ additional pebbles and $O(\log n)$ additional rounds.

\begin{proof}[Proof of Claim 1]
By assumption, $\widetilde{H_{1}} \cap \widetilde{H_{2}} = Z(H)$. So $R_{1} \cap R_{2} \leq Z(H)$. Suppose to the contrary that there exists a nontrivial $z \in R_{1} \cap R_{2}$. Then there exists a $z \in R_1 \cap R_2$ such that $|z| = p$ for some prime $p$. Then there also exists a central $p$-element $w$ that splits from $Z(H)$ and where $z \in \langle w \rangle$. 
Write $w = (r_{i}, b_{i})$ with respect to the chosen direct decomposition for $\widetilde{H_{i}}$.  As $z \in \langle w \rangle$, we have that there is some $m$ so that $w^{m} = z \in R_{1} \cap R_{2}$. So $w^{m} \neq 1$. Furthermore, we may write $w^{m} = (r_{1}^{m}, 1) = (r_{2}^{m}, 1)$. As $w$ has $p$-power order, we have as well that $|b_{i}| < |r_{i}|$ for each $i \in \{1,2\}$ and therefore $|w|=|r_i|$. Now $w$ does not split from $\widetilde{H_{i}}$; otherwise, by \Lem{BSLemma612}, we would have that $r_{i}$ splits from $R_{i}$. However, neither $R_{1}$ nor $R_{2}$ admit Abelian direct factors.

It follows that $w$ has $p$-power order in $\widetilde{H_{1}} \cap \widetilde{H_{2}}$, splits from $Z(H)$, but does not split from $\widetilde{H_{1}}$ or $\widetilde{H_{2}}$. Such elements do not exist in $G$: if $x \in G_1Z(G) \cap G_2Z(G)$ has $p$-power order and splits from $Z(G)$, we claim it must also split from at least one of $\widetilde{G_1}$ or $\widetilde{G_2}$. For $G_1 Z(G) \cap G_2 Z(G) = Z(G) = Z(G_1) \times Z(G_2)$, so by \Lem{BSLemma612}, if $x$ splits from $Z(G)$, then $x = (x_1, x_2)$ with $x_i \in Z(G_i)$ for $i=1,2$, and there exists $i$ such that $|x| = |x_i|$ and $x_i$ splits from $Z(G_i)$. But now consider $x$ as an element of $\widetilde{G}_{3-i} = G_{3-i} \times Z(G_i)$. By \Lem{BSLemma612} again, but now applied to $G_{3-i} \times Z(G_i)$, we find that since $|x|=|x_i|$ and $x_i$ splits from $Z(G_i)$, $x$ splits from $G_{3-i} \times Z(G_i) = \widetilde{G}_{3-i}$.
Thus, in this case, we have by \Lem{LemmaGeneralSplit} that Spoiler can win with $4$ additional pebbles and $O(\log n)$ additional rounds.
\end{proof}
\end{quotation}

\noindent We next consider maximal Abelian direct factors $A \leq G$ and $B \leq H$. Write $H = R \times B$. By \Thm{IdentifyAbelianDirectFactor}, we may assume that $A \cong B$ (or Spoiler can win with $5$ additional pebbles and $O(1)$ additional rounds). We now argue that $R_{1}$ and $R_{2}$ can be chosen such that $R_{1}R_{2} \cap B = \{1\}$. For $i \in \{1,2\}$, we may write:
\[
\widetilde{H_{i}} = \langle (r_{1}, b_{1}), \ldots, (r_{t}, b_{t}) \rangle \leq R \times B.
\] 

As $B \leq Z(H) \leq \widetilde{H_{i}}$, we have that:
\[
\widetilde{H_{i}} = \langle (r_{1}, 1), (1, b_{1}), \ldots, (r_{t}, 1), (b_{t}, 1) \rangle = \langle (r_{1}, 1), \ldots, (r_{t}, 1)\rangle \times B.
\]

It follows that we may choose $R_{1}R_{2} \leq R$. By Claim 1, we have that $R_{1} \cap R_{2} = \{1\}$. So $R_{1}R_{2}B = R_{1} \times R_{2} \times B \leq H$. As $(k, r+O(\log n))$-WL fails to distinguish $G$ and $H$, we have by \Thm{IdentifyAbelianDirectFactor} that $|R_{1}| \cdot |R_{2}| \cdot |B| = |H|$. So in fact, $H = R_{1} \times R_{2} \times B$, which we may write as $(R_{1} \times B_{1}) \times (R_{2} \times B_{2})$, where $B_{i} \leq H_{i}$ are chosen such that $B = B_{1} \times B_{2}$ and $B_{i}$ is isomorphic to a maximal Abelian direct factor of $G_{i}$. Furthermore, we have that $R_{i}Z(H) = \widetilde{H_{i}}$, by construction. The result follows.
\end{proof}

\begin{lemma} \label{LemmaPreservesFullness}
Let $G = N_1 \cdots N_{m}$ and $H = Q_{1} \cdots Q_{m}$ be the decompositions of $G$ and $H$ into non-Abelian components (Def.~\ref{MaximalCentralizerDef}). Let $G = G_{1} \times \ldots \times G_{d}$ be a decomposition into indecomposable direct factors. Let $f : G \to H$ be the bijection that Duplicator selects. If there exists $x \in G$ that is full (Def.~\ref{def:full}) for $(G_{j_{1}}, \ldots, G_{j_{\ell}})$, but for all collections $(H_{j_1}, \ldots, H_{j_m})$ of indecomposable direct factors of $H$, $f(x)$ is not full for that collection of direct factors, 
 then Spoiler may win with $5$ pebbles and $O(\log n)$ rounds.
 
Furthermore, if there exists $x \in G$ is full for exactly one directly indecomposable direct factor of $G$ and $f(x)$ is not full for exactly one indecomposable direct factor of $H$, then Spoiler can win with $5$ pebbles and $O(\log n)$ rounds.
\end{lemma}

\begin{proof}
Spoiler begins by pebbling $x \mapsto f(x)$. Let $f' : G \to H$ be the bijection Duplicator selects at the next round. By \Lem{LemBSCanonicalDecompositionA}, we may assume that $f'(N_{x}) = N_{f(x)}$ and $f'(C_{x}) = C_{f(x)}$, or Spoiler wins with $3$ pebbles and $O(\log n)$ rounds. In particular, there exists a constant $t$ such that $(4, t \log(n))$-WL detects $C_{x}Z(G)$ and $N_{x}Z(G)$. Thus the hypotheses of \Lem{ModifiedLem624} are satisfied by taking $r := t\log(n)$. As the central decomposition $G = C_{x}N_{x}$ is directly induced, we have that by \Lem{ModifiedLem624}, the central decomposition $H = C_{f(x)}N_{f(x)}$ has to be directly induced or Spoiler can win with $5$ pebbles and $O(\log n)$ rounds. So by \Lem{BSLem623}, we have that $f(x)$ is full for some collection of direct factors of $H$.

To see the ``furthermore'', note that $x$ is full for exactly one indecomposable direct factor iff $C_x$ is inclusion-maximal among all $C_{x'}$. If $x$ has this property but $f(x)$ does not, then Spoiler pebbles $x \mapsto f(x)$. In the next round, let $f'$ be the bijection chosen by Duplicator. Since Duplicator must choose bijections that map $C_x$ to $C_{f(x)}$ but $C_{f(x)}$ is not inclusion-maximal, there is some $y' \in H$ such that $y'$ is full for some collection of indecomposable direct factors of $H$, $C_{y'}$ properly contains $C_{f(x)}$, and there is some $z' \in C_{y'}$ that commutes with $y'$ but not with $f(x)$. Spoiler pebbles $y := (f')^{-1}(y') \mapsto y'$. Let $f''$ be the bijection chosen by Duplicator on the next round, then Spoiler pebbles $z := (f'')^{-1}(z') \mapsto z'$. 

Now, by the same argument as above, Duplicator's bijection $f''$ must map $C_y$ to $C_{y'}$ or Spoiler wins with 1 more pebble (Duplicator can reuse the two pebbles on $x$ and $z$ to implement the strategy of \Lem{LemBSCanonicalDecompositionA}) and $O(\log n)$ more rounds. Since $C_{y'} \geq C_{f(x)}$, we must have $C_y \geq C_x$. But then by maximality of $C_x$, we have $C_{y} = C_x$. Therefore $z$ is in $C_x$ while $z'$ does not commute with $f(x)$, so Spoiler wins immediately.
\end{proof}

\begin{corollary} \label{CorApproximateBelow}
Let $G = G_{1} \times \ldots \times G_{d}$ be a decomposition of $G$ into directly indecomposable factors. Let $H$ be arbitrary. Let $\mathcal{F}_{G} \subseteq G$ be the set of elements that are full for exactly one indecomposable direct factor of $G$, 
and define $\mathcal{F}_{H}$ analogously.\footnote{We note that \cite{BrachterSchweitzerWLLibrary} use $\mathcal{F}$ for these, and $\mathcal{F}_G$ for $\bigcup_{g \in \mathcal{F}} N_g = \bigcup_i G_i Z(G)$. We use $\mathcal{F}_G$ as written to indicate which group it is in, and we do not introduce special notation for $\bigcup_{g \in \mathcal{F}} N_g = \bigcup_i G_i Z(G)$.} If Duplicator does not select a bijection $f : G \to H$ satisfying:
\[
f\left( \bigcup_{g \in \mathcal{F}_{G}} N_{g} \right) = \bigcup_{h \in \mathcal{F}_{H}} N_{h},
\]
then Spoiler can win using $5$ pebbles and $O(\log n)$ rounds.
\end{corollary}

\begin{proof}
By \Lem{LemmaPreservesFullness}, we may assume that $f(\mathcal{F}_{G}) = \mathcal{F}_{H}$ (or Spoiler wins with $5$ pebbles and $O(\log n)$ rounds). Now suppose that for some $g \in \mathcal{F}_G$ that there exists an $x \in N_{g}$ such that $f(x) \not \in N_{h}$ for any $h \in \mathcal{F}_{H}$. Spoiler pebbles $x \mapsto f(x)$. Let $f' : G \to H$ be the bijection Duplicator selects at the next round. Again, we may assume that $f(\mathcal{F}_{G}) = \mathcal{F}_{H}$ (or Spoiler wins with 4 additional pebbles and $O(\log n)$ additional rounds; note that although \Lem{LemmaPreservesFullness} uses 5 pebbles, here we only need 4 more as Spoiler can reuse the pebble that was previous on $x \mapsto f(x)$). Spoiler now pebbles $g \mapsto f'(g)$. Now on any subsequent bijection where these two pebbles have not moved, Duplicator cannot map $N_{g} \mapsto N_{f'(g)}$. So by \Lem{LemBSCanonicalDecompositionA}, Spoiler wins with $4$ additional pebbles and $O(\log n)$ rounds.
\end{proof}

\begin{theorem}
Let $k \geq 5$, and let $r := r(n)$. Let $G = G_1 \times \cdots \times G_d$ be a decomposition into indecomposable direct factors. If $G$ and $H$ are not distinguished by $(k, r + O(\log n))$-WL Version II, then there exist indecomposable direct factors $H_{i} \leq H$ such that $H = H_{1} \times \cdots \times H_{d}$; and for all $i \in [d]$, $G_i$ and $H_i$ are not distinguished by $(k-1,r)$-WL Version II. Furthermore, $G$ and $H$ have isomorphic maximal Abelian direct factors, and $(k-1, r)$-WL Version II fails to distinguish $G_{i}Z(G)$ from $H_{i}Z(H)$, for all $i \in [d]$.
\end{theorem}

\begin{proof}
We may assume that $H$ is non-Abelian as well, or Spoiler can win with 2 pebbles in 2 rounds, by pebbling a pair of non-commuting elements of $G$. Let $f : G \to H$ be the bijection that Duplicator selects. By \Cor{CorApproximateBelow}, we may assume that $f(\mathcal{F}_{G}) = \mathcal{F}_{H}$. It follows that $H$ must admit a decomposition $H = H_{1} \times \ldots \times H_{\ell}$, where the $H_{j}$ factors are directly indecomposable and $\bigcup_{h \in \mathcal{F}_H} N_h = \bigcup_{j} H_{j}Z(H) \subseteq H$, which we again note is indistinguishable from $\bigcup_i G_i Z(G)$. Let $X_{G}$ be the non-commuting graph of $G$, and let $X_{H}$ be the non-commuting graph of $H$. Recall from \cite[Proposition~2.1]{ABDOLLAHI2006468} that as $G, H$ are non-Abelian, $X_{G}$ and $X_{H}$ are connected.

As different direct factors centralize each other, we obtain that for each non-singleton connected component $K$ of $X_{G}[\bigcup_i G_i Z(G)]$, there exists a unique indecomposable direct factor $G_{i}$ such that $K = G_{i}Z(G) \setminus Z(G)$. Thus, $G_{i}Z(G) = \langle K \rangle$. Again by \cite[Proposition~2.1]{ABDOLLAHI2006468}, all such non-Abelian direct factors appear in this way. 

We note that the claims in the preceding paragraph applies to $H$ as well. So if $(k,r + O(\log n))$-WL Version II does not distinguish $G$ and $H$, there must exist a bijection between the connected components of $X_{G}[\bigcup_i G_i Z(G)]$ and $X_{H}[\bigcup_i H_i Z(G)]$. Namely, we may assume that $G$ and $H$ admit a decomposition into $\ell = d$ directly indecomposable factors, and that these subgroups are indistinguishable by $(k-1,r)$-WL: for we have a correspondence (after an appropriate reordering of the factors) between $G_{i}Z(H)$ and $H_{i}Z(H)$, where $G_{i}Z(H)$ and $H_{i}Z(H)$ are not distinguished by $(k,r)$-WL. By \Lem{SemiAbelianDirectProduct}, we have that $(k-1,r)$-WL Version II does not distinguish $G_{i}$ from $H_{i}$. By \Thm{IdentifyAbelianDirectFactor}, $G$ and $H$ must have isomorphic maximal Abelian direct factors. So when $G_{i}, H_{i}$ are Abelian, we even have $G_{i} \cong H_{i}$.
\end{proof}

\section{Weisfeiler--Leman for semisimple groups}

In this section, we show that Weisfeiler--Leman can be fruitfully used as a tool to improve the parallel complexity of isomorphism testing of groups with no Abelian normal subgroups, also known as semisimple or Fitting-free groups. The main result of this section is:

\begin{theorem} \label{thm:ListSemisimpleIsomorphisms}
Let $G$ be a semisimple group, and let $H$ be arbitrary. We can test isomorphism between $G$ and $H$ using an $\textsf{SAC}$ circuit of depth $O(\log n)$ and size $n^{\Theta(\log \log n)}$. Furthermore, all such isomorphisms can be listed in this bound.
\end{theorem}

The previous best complexity upper bounds were $\mathsf{P}$ for testing isomorphism \cite{BCQ}, and $\mathsf{DTIME}(n^{O(\log \log n)})$ for listing isomorphisms \cite{BCGQ}.

We start with what we can observe from known results about direct products of simple groups. Brachter \& Schweitzer previously showed that $3$-WL Version II identifies direct products of finite simple groups. A closer analysis of their proofs \cite[Lemmas~5.20 \& 5.21]{BrachterSchweitzerWLLibrary} show that only $O(1)$ rounds are required. Thus, we obtain the following. 

\begin{corollary}[{cf. Brachter \& Schweitzer \cite[Lemmas~5.20 \& 5.21]{BrachterSchweitzerWLLibrary}}]
Isomorphism between a direct product of non-Abelian simple groups and an arbitrary group can be decided in $\mathsf{L}$.
\end{corollary}

Our parallel machinery also immediately lets us extend a similar result to direct products of \emph{almost} simple groups (a group $G$ is almost simple if there is a non-Abelian simple group $S$ such that $\Inn(S) \leq G \leq \Aut(S)$; equivalently, if $\Soc(G)$ is non-Abelian simple).

\begin{corollary}
Isomorphism between a direct product of almost simple groups and an arbitrary group can be decided in $\mathsf{TC}^{1}$. 
\end{corollary}

\begin{proof}
Because almost simple groups are 3-generated \cite{DVL}, they are identified by $(O(1), O(1))$-WL. By \Thm{ThmProduct}, direct products of almost simple groups are thus identified by $(O(1), O(\log n))$-WL. 
\end{proof}

\subsection{Preliminaries}
We recall some facts about semisimple groups from \cite{BCGQ}. As a semisimple group $G$ has no Abelian normal subgroups, we have that $\Soc(G)$ is the direct product of non-Abelian simple groups.  The conjugation action of $G$ on $\Soc(G)$ permutes the direct factors of $\Soc(G)$. So there exists a faithful permutation representation $\alpha : G \to G^{*} \leq \Aut(\Soc(G))$. $G$ is determined by $\Soc(G)$ and the action $\alpha$. Let $H$ be a semisimple group with the associated action $\beta : H \to \text{Aut}(\Soc(H))$. We have that $G \cong H$ precisely if $\Soc(G) \cong \Soc(H)$ via an isomorphism that makes $\alpha$ equivalent to $\beta$. 

We now introduce the notion of permutational isomorphism, which is our notion of equivalence for $\alpha$ and $\beta$. Let $A$ and $B$ be finite sets, and let $\pi : A \to B$ be a bijection. For $\sigma \in \text{Sym}(A)$, let $\sigma^{\pi} \in \text{Sym}(B)$ be defined by $\sigma^{\pi} := \pi^{-1}\sigma \pi$. For a set $\Sigma \subseteq \text{Sym}(A)$, denote $\Sigma^{\pi} := \{ \sigma^{\pi} : \sigma \in \Sigma\}$. Let $K \leq \text{Sym}(A)$ and $L \leq \text{Sym}(B)$ be permutation groups. A bijection $\pi : A \to B$ is a \textit{permutational isomorphism} $K \to L$ if $K^{\pi} = L$.

The following lemma, applied with $R = \Soc(G)$ and $S = \Soc(H)$, precisely characterizes semisimple groups \cite{BCGQ}.      
 
\begin{lemma}[{\cite[Lemma 3.1]{BCGQ}}] \label{CharacterizeSemisimple}
Let $G$ and $H$ be groups. Let $R \unlhd G$ and $S \unlhd H$, such that $R, S$ have trivial centralizers. Let $\alpha : G \to G^{*} \leq \Aut(R)$ and $\beta : H \to H^{*} \leq \Aut(S)$ be faithful permutation representations of $G$ and $H$ via the conjugation action on $R$ and $S$, respectively. Let $f : R \to S$ be an isomorphism. Then $f$ extends to an isomorphism $\hat{f} : G \to H$ if and only if $f$ is a permutational isomorphism between $G^{*}$ and $H^{*}$; and if so, $\hat{f} = \alpha f^{*} \beta^{-1}$, where $f^{*} :  G^{*} \to H^{*}$ is the isomorphism induced by $f$.
\end{lemma}

We also need the following standard group-theoretic lemmas. The first provides a key condition for identifying whether a non-Abelian simple group belongs in the socle. Namely, if $S_{1} \cong S_{2}$ are non-Abelian simple groups where $S_{1}$ is in the socle and $S_{2}$ is not in the socle, then the normal closures of $S_{1}$ and $S_{2}$ are non-isomorphic. In particular, the normal closure of $S_{1}$ is a direct product of non-Abelian simple groups, while the normal closure of $S_{2}$ is not a direct product of non-Abelian simple groups. We will apply this condition later when $S_{1}$ is a simple direct factor of $\Soc(G)$; in which case, the normal closure of $S_{1}$ is of the form $S_{1}^{k}$. We include the proofs of these two lemmas for completeness.

\begin{lemma} \label{LemmaSocle}
Let $G$ be a finite semisimple group. A subgroup $S \leq G$ is contained in $\Soc(G)$ if and only if the normal closure of $S$ is a direct product of nonabelian simple groups.
\end{lemma}

\begin{proof}
Let $N$ be the normal closure of $S$. Since the socle is normal in $G$ and $N$ is the smallest normal subgroup containing $S$, we have that $S$ is contained in $\Soc(G)$ if and only if $N$ is. 

Suppose first that $S$ is contained in the socle. Since $\Soc(G)$ is normal and contains $S$, by the definition of $N$ we have that $N \leq \Soc(G)$. As $N$ is a normal subgroup of $G$, contained in $\Soc(G)$, it is a direct product of minimal normal subgroups of $G$, each of which is a direct product of non-Abelian simple groups.

Conversely, suppose $N$ is a direct product of nonabelian simple groups. We proceed by induction on the size of $N$. If $N$ is minimal normal in $G$, then $N$ is contained in the socle by definition. If $N$ is not minimal normal, then it contains a proper subgroup $M \lneq N$ such that $M$ is normal in $G$, hence also $M \unlhd N$. However, as $N$ is a direct product of nonabelian simple groups $T_1, \dotsc, T_k$, the only subgroups of $N$ that are normal in $N$ are direct products of subsets of $\{T_1, \dotsc, T_k\}$, and all such normal subgroups have direct complements. Thus we may write $N = L \times M$ where both $L,M$ are nontrivial, hence strictly smaller than $N$, and both $L$ and $M$ are direct product of nonabelian simple groups. 

We now argue that $L$ must also be normal in $G$. Since conjugating $N$ by $g \in G$ is an automorphism of $N$, we have that $N = g L g^{-1} \times g M g^{-1}$. Since $M$ is normal in $G$, the second factor here is just $M$, so we have $N = gLg^{-1} \times M$. But since the direct complement of $M$ in $N$ is unique (since $N$ is a direct product of \emph{non-Abelian} simple groups), we must have $g L g^{-1} = L$. Thus $L$ is normal in $G$. 

By induction, both $L$ and $M$ are contained in $\Soc(G)$, and thus so is $N$. We conclude since $S \leq N$.
\end{proof}

\begin{corollary} \label{CorSocleFactor}
Let $G$ be a finite semisimple group. A nonabelian simple subgroup $S \leq G$ is a direct factor of $\Soc(G)$ if and only if its normal closure $N = ncl_G(S)$ is isomorphic to $S^k$ for some $k \geq 1$ and $S \unlhd N$.
\end{corollary}

\begin{proof}
Let $S$ be a nonabelian simple subgroup of $G$. If $S$ is a direct factor of $\Soc(G)$, then $\Soc(G) = S^k \times T$ for some $k \geq 1$ and some $T$; choose $T$ such that $k$ is maximal. Then the normal closure of $S$ is a minimal normal subgroup of $\Soc(G)$ which contains $S$ as a normal subgroup. Since the normal subgroups of a direct product of nonabelian simple groups are precisely direct products of subsets of the factors, the normal closure of $S$ is some $S^{k'}$ for $1 \leq k' \leq k$.

Conversely, suppose the normal closure $N$ of $S$ is isomorphic to $S^k$ for some $k \geq 1$ and $S \unlhd N$. By \Lem{LemmaSocle}, $S$ is in $\Soc(G)$, and thus so is $N$ (being the normal closure of a subgroup of the socle). Furthermore, as a normal subgroup of $G$ contained in $\Soc(G)$, $N$ is a direct product of minimal normal subgroups and a direct factor of $\Soc(G)$ (in fact it is minimal normal itself, but we haven't established that yet, nor will we need to). Since $S$ is a normal subgroup of $N$, and $N$ is a direct product of non-Abelian simple groups, $S$ is a direct factor of $N$. Since $N$ is a direct factor of $\Soc(G)$, and $S$ is a direct factor of $N$, $S$ is a direct factor of $\Soc(G)$. This completes the proof.
\end{proof}

\begin{lemma} \label{LemmaDirectProdSimple}
Let $S_1, \dotsc, S_k \leq G$ be nonabelian simple subgroups such that for all distinct $i,j \in [k]$ we have $[S_i, S_j] = 1$. Then $\langle S_1, \dotsc, S_k \rangle = S_1 S_2 \dotsb S_k = S_1 \times \dotsb \times S_k$.
\end{lemma}

\begin{proof}
By induction on $k$. The base case $k=1$ is vacuously true. Suppose $k \geq 2$ and that the result holds for $k-1$. Then $T := S_1 S_2 \dotsb S_{k-1} = S_1 \times \dotsb \times S_{k-1}$. Now, since $S_k$ commutes with each $S_i$, and they generate $T$, we have that $[S_k, T] = 1$. Hence $T$ is contained in the normalizer (or even the centralizer) of $S_k$, so $TS_k = S_kT = \langle T, S_k \rangle$, and $S_k$ and $T$ are normal subgroups of $TS_k$. As $TS_k = \langle T,S_k \rangle$ and $T,S_k$ are both normal subgroups of $TS_k$ with $[T, S_k] = 1$, we have that $TS_k$ is a central product of $T$ and $S_k$. As $Z(T) = Z(S_k)=1$, it is their direct product.
\end{proof}

\subsection{Groups without Abelian normal subgroups in parallel}
Here we establish \Thm{thm:ListSemisimpleIsomorphisms}. We begin with the following.

\begin{proposition} \label{IdentifySemisimple}
Let $G$ be a semisimple group of order $n$, and let $H$ be an arbitrary group of order $n$. If $H$ is not semisimple, then Spoiler can win the Version II pebble game with at most $2$ pebbles and $2$ rounds.
\end{proposition}

\begin{proof}
Recall that a group is semisimple if and only if it contains no Abelian normal subgroups. As $H$ is not semisimple, $\Soc(H) = A \times T$, where $A$ is the direct product of elementary Abelian groups and $T$ is a direct product of non-Abelian simple groups. We show that Spoiler can win using at most $2$ pebbles and $2$ rounds. Let $f : G \to H$ be the bijection that Duplicator selects. Let $a \in A$. So $\text{ncl}_{H}(a) \leq A$. Let $b := f^{-1}(a) \in G$, and let $B := \text{ncl}_{G}(b)$. As $G$ is semisimple, we have that $B$ is not Abelian. Spoiler begins by pebbling $b \mapsto a$.

So there exist $g \in G$ such that $b$ and $g b g^{-1}$ do not commute (for $B$ is generated by $\{g b g^{-1} : g \in G\}$, and if they all commuted with $b$ then $b$ would be in $Z(B)$, but $Z(B)$ is characteristic in $B$ hence normal in $G$, hence $\text{ncl}_G(B) \leq Z(B)$ and $B$ would be Abelian). Let $f' : G \to H$ be the bijection that Duplicator selects at the next round. Spoiler pebbles $g \mapsto f'(g)$. As $\ncl(a) = A$ is Abelian, $f'(g)f(b)f'(g)^{-1}$ and $f(b)$ commute. Spoiler now wins.
\end{proof}
 
We now apply Lemma \ref{LemmaSocle} to show that Duplicator must map the direct factors of $\Soc(G)$ to isomorphic direct factors of $\Soc(H)$.

\begin{lemma} \label{LemmaProdSimple}
Let $G,H$ be finite semisimple groups of order $n$. Let $\Fac(\Soc(G))$ denote the set of simple direct factors of $\Soc(G)$. Let $S \in \Fac(\Soc(G))$ be a non-Abelian simple group. Let $a \in S$, and let $f : G \to H$ be the bijection that Duplicator selects. 
\begin{enumerate}[label=(\roman*)]
\item If $f(a)$ does not belong to some element of $\Fac(\Soc(H))$, or 
\item If there exists some $T \in \Fac(\Soc(H))$ such that $f(a) \in T$, but $S \not \cong T$,
\end{enumerate}

then Spoiler wins the Version II pebble game with at most $O(1)$ pebbles and $O(1)$ rounds.
\end{lemma}

\begin{proof}
Spoiler begins by pebbling $a \mapsto f(a)$. At the next two rounds, Spoiler pebbles generators $x, y$ for $S$. Let $f' : G \to H$ be the bijection Duplicator selects at the next round. Denote $T := \langle f'(x), f'(y) \rangle$. We note that if $T \not \cong S$ or $f(a) \not \in T$, then Spoiler wins. 

So suppose that $f(a) \in T$ and $T \cong S$. We have two cases.
\begin{itemize}
\item \textbf{Case 1:} Suppose first that $T$ does not belong to $\Soc(H)$. As $S \unlhd \text{Soc}(G)$, the normal closure $\text{ncl}(S)$ is minimal normal in $G$ \cite[Exercise 2.A.7]{Isaacs2008FiniteGT}. As $T$ is not even contained in $\text{Soc}(H)$, we have by Lemma \ref{LemmaSocle} that $\text{ncl}(T)$ is not a direct product of non-Abelian simple groups, so $\text{ncl}(S) \not\cong \text{ncl}(T)$. We note that $\text{ncl}(S) = \langle \{ gSg^{-1} : g \in G \} \rangle$.

As $\text{ncl}(T)$ is not isomorphic to a direct power of $S$, there is some conjugate $g S g^{-1} \neq S$ such that $f'(g)Tf'(g)^{-1}$ does not commute with $T$, by Lemma~\ref{LemmaDirectProdSimple}. Yet since $S \unlhd \Soc(G)$, $gSg^{-1}$ and $S$ do commute. Spoiler moves the pebble pair from $a \mapsto f(a)$ and pebbles $g$ with $f'(g)$. Since Spoiler has now pebbled $x,y,g$ which generate $\langle S, gSg^{-1} \rangle = S \times gSg^{-1} \cong S \times S$ but the image is not isomorphic to $S \times S$, the map $(x,y,g) \mapsto (f'(x), f'(y), f'(g))$ does not extend to an isomorphism of $S \times g S g^{-1}$. Spoiler now wins. In total, Spoiler used $O(1)$ pebbles and $O(1)$ rounds.

\item \textbf{Case 2:} Suppose instead that $T \leq \Soc(H)$, but that $T$ is not normal in $\Soc(H)$. As $T$ is not normal in $\Soc(H)$, there exists $Q = \langle q_{1}, q_{2} \rangle \in \Fac(\Soc(H))$ such that $Q$ does not normalize $T$. At the next two rounds, Spoiler pebbles $q_{1}, q_{2}$, and their respective preimages, which we label $r_{1}, r_{2}$. When pebbling $r_{1} \mapsto q_{1}$, we may assume that Spoiler moves the pebble placed on $a \mapsto f(a)$. By Case 1, we may assume that $r_{1}, r_{2} \in \Soc(G)$, or Spoiler wins with an additional $1$ pebble and $1$ round. Now as $S \trianglelefteq \Soc(G)$, $\langle r_{1}, r_{2}\rangle$ normalizes $S$. However, $Q$ does not normalize $T$. So the pebbled map $(x, y, r_{1}, r_{2}) \mapsto (f'(x), f'(y), q_{1}, q_{2})$ does not extend to an isomorphism. Thus, Spoiler used $O(1)$ pebbles and $O(1)$ rounds.  \qedhere
\end{itemize}
\end{proof}

\begin{lemma} \label{IdentifyDirectFactors}
Let $G$ be a semisimple group. There is a logspace algorithm that decides, given $g_{1}, g_{2} \in G$, whether $\langle g_{1}, g_{2} \rangle \in \Fac(\Soc(G))$.
\end{lemma}

\begin{proof}
Using a membership test \cite{BarringtonMcKenzie, TangThesis}, we may enumerate the elements of $S := \langle g_{1}, g_{2} \rangle$ by a logspace transducer. We first check whether $S$ is simple. For each $g \in S$, we check whether $\ncl_{S}(g) = S$. This check is $\textsf{L}$-computable \cite[Thm.~7.3.3]{VijayaraghavanThesis}. 

It remains to check whether $S \in \Fac(\Soc(G))$. By \Cor{CorSocleFactor}, $S \in \Fac(\Soc(G))$ if and only if $N := \ncl_{G}(S) = S^{k}$ for some $k$ and $S \unlhd N$. As $S$ is simple, it suffices to check that each conjugate of $S$ is either (1) equal to $S$ or (2) intersects trivially with $S$ and commutes with $S$. For a given $g \in G$ and each $h \in S$, we may check whether $h \in gSg^{-1}$. If there exist non-trivial $h_{1}, h_{2} \in S$ such that $h_{1} \in gSg^{-1}$ and $h_{2} \not \in gSg^{-1}$, we return that $S \not \in \Fac(\Soc(G))$. Otherwise, we know that all conjugates of $S$ are either equal to $S$ or intersect $S$ trivially. Next we check that those conjugates that intersect $S$ trivially commute with $S$. For each $g \in G, h_1, h_2 \in S$ we check whether $g h_1 g^{-1} \in S$; if not, we check that $[gh_1 g^{-1}, h_2]=1$. If not, then we return that $S \not \in \Fac(\Soc(G))$. If all these tests pass, then $S$ is a direct factor of the socle. For both of these procedures, we only need to iterate over 3- and 4-tuples of elements of $G$ or $S$, so this entire procedure is $\textsf{L}$-computable. The result follows.
\end{proof}

\begin{lemma} \label{LemmaAC1Factors}
Let $G$ be a semisimple group. We can compute the direct factors of $\Soc(G)$ using a  logspace transducer.
\end{lemma}

\begin{proof}
Using \Lem{IdentifyDirectFactors}, we may identify in $\textsf{L}$ the ordered pairs that generate direct factors of $\Soc(G)$. Now for $x \in G$ and a pair $(g_{1}, g_{2})$ that generates a direct factor of $\Soc(G)$, define an indicator $Y(x, g_{1}, g_{2}) = 1$ if and only if $x \in \langle g_{1}, g_{2} \rangle$. We may use a membership test \cite{BarringtonMcKenzie, TangThesis} to decide in $\textsf{L}$ whether $x \in \langle g_{1}, g_{2} \rangle$. Thus, we are able to write down the direct factors of $\Soc(G)$ and their elements in $\textsf{L}$.
\end{proof}

We now prove \Thm{thm:ListSemisimpleIsomorphisms}.

\begin{proof}[Proof of \Thm{thm:ListSemisimpleIsomorphisms}]
By Prop.~\ref{IdentifySemisimple}, we may assume that both groups are semisimple. We now note that, by \Lem{LemmaProdSimple}, if $\Soc(G) \not \cong \Soc(H)$, then $(O(1), O(1))$-WL Version II will distinguish $G$ from $H$. For in this case, there is some simple normal factor $S \in \Fac(\Soc(G))$ such that there are more copies of $S$ in $\Fac(\Soc(G))$ than in $\Fac(\Soc(H))$. Thus under any bijection Duplicator selects, some element of $S$ must get mapped into a simple direct factor of $\Soc(H)$ that is not isomorphic to $S$, and thus by \Lem{LemmaProdSimple}, Spoiler can win with $O(1)$ pebbles and $O(1)$ rounds.

So suppose $\Soc(G) \cong \Soc(H)$. By \Lem{LemmaAC1Factors}, in $\textsf{L}$ we may enumerate the non-Abelian simple direct factors of $\Soc(G)$ and $\Soc(H)$. Furthermore, we may decide in $\textsf{L}$---and therefore, $\textsf{SAC}^{1}$---with a membership test \cite{BarringtonMcKenzie, TangThesis} whether two non-Abelian simple direct factors of the socle are conjugate. Thus, in $\textsf{SAC}^{1}$, we may compute a decomposition $\Soc(G)$ and $\Soc(H)$ of the form $T_{1}^{t_{1}} \times \cdots \times T_{k}^{t_{k}}$, where each $T_{i}$ is non-Abelian simple and each $T_{i}^{t_{i}}$ is minimal normal.

There are at most $|S|^{2}$ automorphisms of each simple factor $|S|$, and so there are at most
\[
n^{2} k! \prod_{i=1}^{k} t_{i}!
\]
isomorphisms between $\Soc(G)$ and $\Soc(H)$ that could extend to isomorphisms $G\cong H$. From \cite{BCGQ}, we note that this quantity is bounded by $n^{O(\log \log n)}$. (This is bound is tight, as in the case of the groups $A_{5}^{k}$.) 

We now turn to testing isomorphism of $G$ and $H$. To do so, we use the individualize and refine strategy. We individualize in $G$ arbitrary generators for each element of $\Fac(\Soc(G))$ (2 for each factor). In parallel, we try each of the $\leq k! \prod_{i=1}^k t_i! \leq n^{O(\log \log n)}$ possible bijections $\psi \colon \Fac(\Soc(G)) \to \Fac(\Soc(H))$ (this is the one place responsible for the quasi-polynomial, rather than polynomial, size of our resulting circuits).
Then for each configuration of generators for the elements of $\Fac(\Soc(H))$, we individualize those in such a way that respects $\psi$. Precisely, if $\psi(S) = T$ and $(g_{1}, g_{2})$ are individualized in $S$, then for the desired generators $(h_{1}, h_{2})$ of $T$, we individualize $h_{i}$ to receive the same color as $g_{i}$. Note that, although we are individualizing up to $2 \log |G|$ elements here, we are not choosing them from all possible $\binom{|G|}{2 \log |G|}$ choices (which would be worse than the trivial upper bound!); the algorithm only considers at most $\prod_{S \in \Fac(\Soc(H))} \binom{|S|}{2} \leq O(|G|^2)$ many choices for which tuples to individualize.

Observe that in two more rounds, no two elements of $\Soc(G)$ have the same color. Similarly, in two more rounds, no two elements of $\Soc(H)$ have the same color. However, an element of $\Soc(G)$ and an element of $\Soc(H)$ may share the same color.

Suppose now that $G \not \cong H$. We now show that constant-dimensional WL Version II coloring starting from the coloring above distinguishes $G$ from $H$, using the Spoiler--Duplicator game.
Let $f : G \to H$ be the bijection that Duplicator selects. As $G \not \cong H$, there exists $g \in G$ and $s \in \Soc(G)$ such that $f(gsg^{-1}) \neq f(g)f(s)f(g^{-1})$. Spoiler pebbles $g$. Let $f' : G \to H$ be the bijection Duplicator selects at the next round. As no two elements of $\Soc(G)$ have the same color and no two elements of $\Soc(H)$ have the same color, we have that $f'(s) = f(s)$. Spoiler pebbles $s$ and wins. So after the individualization step, $(2,4)$-WL Version II will decide whether the given map extends to an isomorphism of $G \cong H$. Now $(2,4)$-WL Version II is $\textsf{L}$-computable, and so $\textsf{SAC}^{1}$ computable. As we have to test at most $n^{O(\log \log n)}$ isomorphisms of $\Soc(G) \cong \Soc(H)$, our circuit has size $n^{O(\log \log n)}$. The result now follows.
\end{proof}

\begin{remark}
We also note that there is at most one such way of extending the given isomorphism between $\Soc(G)$ and $\Soc(H)$ to that of $G$ and $H$ \cite[Lemma~3.1]{BCGQ}. So in particular, after individualizing the generators for the non-Abelian simple direct factors of the socles, from the last paragraph in the proof we see that WL will assign a unique color to each element of the group.
\end{remark}

We also obtain the following corollary, which improves upon \cite[Corollary~4.4]{BCGQ} in the direction of parallel complexity.

\begin{corollary}
Let $G$ and $H$ be semisimple with $\Soc(G) \cong \Soc(H)$. If $\Soc(G) \cong \Soc(H)$ have $O(\log n / \log \log n)$ non-Abelian simple direct factors, then we can decide isomorphism between $G$ and $H$ in $\textsf{L}$, and list all the isomorphisms between $G$ and $H$ in $\textsf{FL}$.
\end{corollary}

\begin{proof}[Proof (Sketch).]
We proceed identically as in the proof of \Thm{thm:ListSemisimpleIsomorphisms}. As there are $O(\log n / \log \log n)$ non-Abelian simple factors of $\Soc(G) \cong \Soc(H)$, there are only $\poly(n)$ isomorphisms between $\Soc(G)$ and $\Soc(H)$ that could extend to isomorphisms between $G$ and $H$ \cite{BCGQ}. As these isomorphisms can be checked in parallel, we may thus enumerate these isomorphisms using a logspace transducer. The result now follows.
\end{proof}

\section{Count-free Weisfeiler--Leman} \label{CountFreeWL}
In this section, we examine consequences for parallel complexity of the \emph{count-free} WL algorithm. Our first main result here is to show a $\Omega(\log |G|)$ lower bound (optimal and maximal, up to the constant factor) on count-free WL-dimension for identifying Abelian groups (\Thm{CountFreeAbelian}). Despite this result showing that count-free WL on its own is not useful for testing isomorphism of Abelian groups, we nonetheless use count-free WL for Abelian groups, in combination with a few other ideas, to get improved upper bounds on the parallel complexity of testing isomorphism (\Thm{ThmAbelian}) of Abelian groups. 

We begin by defining analogous pebble games and logics for count-free WL versions I-II. Furthermore, we establish the equivalence of the three count-free WL versions up to $O(\log n)$ rounds. These results extend \cite[Section~3]{WLGroups} to the count-free setting.

\subsection{Equivalence between count-free WL, pebble games, and logics}
We define analogous pebble games for count-free WL Versions I-II. The count-free $(k+1)$-pebble game consists of two players: Spoiler and Duplicator, as well as $(k+1)$ pebble pairs $(p, p^{\prime})$. In Versions I and II, Spoiler wishes to show that the two groups $G$ and $H$ are not isomorphic. Duplicator wishes to show that the two groups are isomorphic. Each round of the game proceeds as follows.
\begin{enumerate}
\item Spoiler picks up a pebble pair $(p_{i}, p_{i}^{\prime})$.
\item The winning condition is checked. This will be formalized later.
\item Spoiler places one of the pebbles on some group element (either $p_{i}$ on some element of $G$ or $p_{i}'$ on some element of $H$). 
\item Duplicator places the other pebble on some element of the other group.
\end{enumerate}

Let $v_{1}, \ldots, v_{m}$ be the pebbled elements of $G$ at the end of step 1, and let $v_{1}^{\prime}, \ldots, v_{m}^{\prime}$ be the corresponding pebbled elements  of $H$. Spoiler wins precisely if the map $v_{\ell} \mapsto v_{\ell}^{\prime}$ does not extend to a marked equivalence in the appropriate version of WL. Duplicator wins otherwise. Spoiler wins, by definition, before the start of round $0$ if $G$ and $H$ do not have the same number of elements. We will show, for $J \in \{I, II\}$, that $G$ and $H$ are not distinguished by the first $r$ rounds of the count-free $k$-WL version $J$ if and only if Duplicator wins the first $r$ rounds of the Version $J$ $(k+1)$-pebble game. 

\begin{lemma}
Suppose $|G|=|H|$, and let $\overline{g} := (g_{1}, \ldots, g_{k}) \in G^{k}$ and $\overline{h} := (h_{1}, \ldots, h_{k}) \in H^{k}$. If the count-free $(k,r)$-WL distinguishes $\overline{g}$ and $\overline{h}$, then Spoiler can win in the count-free $(k+1)$-pebble game within $r$ moves on the initial configuration $(\overline{g}, \overline{h})$. (We use the same version of WL and the pebble game).
\end{lemma}

\begin{proof}

\noindent
\begin{itemize}
\item \textbf{Version I:} For $r = 0$, then $\overline{g}$ and $\overline{h}$ differ with respect to the Version I marked equivalence type. Fix $r > 0$. Suppose that $\chi_{r}(\overline{g}) \neq \chi_{r}(\overline{h})$. We have two cases. Suppose first that $\chi_{r-1}(\overline{g}) \neq \chi_{r-1}(\overline{h})$. Then by the inductive hypothesis, Spoiler can win in the $(k+1)$-pebble game using at most $r-1$ moves. 

Suppose instead that $\chi_{r-1}(\overline{g}) = \chi_{r-1}(\overline{h})$. So without loss of generality, there exists an $x \in G$ such that the color configuration $(\chi_{r-1}(\overline{g}(g_{1}/x)), \ldots, \chi_{r-1}(\overline{g}(g_{k}/x))$ does not appear amongst the colored $k$-tuples of $H$. That is, for all $y \in H$, there exists some $j \in [k]$ such that $\chi_{r-1}(\overline{g}(g_{j}/x)) \neq \chi_{r-1}(\overline{h}(h_{j}/y))$. Spoiler places an unused pebble on $x$. Duplicator responds by placing the corresponding pebble on some $y \in H$. Let $j$ be such that $\chi_{r-1}(\overline{g}(g_{j}/x)) \neq \chi_{r-1}(\overline{h}(h_{j}/y))$. Then at the next round, Spoiler removes the pebble on $p_j$ and reuses that. By the inductive hypothesis, Spoiler wins with $r-1$ additional moves. 


\item \textbf{Version II:} We modify the Version I argument above to use the Version II marked equivalence type. Otherwise, the argument is identical. \qedhere
\end{itemize}
\end{proof}

We now prove the converse.

\begin{lemma}
Suppose $|G|=|H|$, and let $\overline{g} := (g_{1}, \ldots, g_{k}) \in G^{k}$ and $\overline{h} := (h_{1}, \ldots, h_{k}) \in H^{k}$. Suppose that Spoiler can win in the count-free $(k+1)$-pebble game within $r$ moves on the initial configuration $(\overline{g}, \overline{h})$. Then the count-free $(k,r)$-WL distinguishes $\overline{g}$ and $\overline{h}$. (We use the same version of WL and the pebble game).
\end{lemma}

\begin{proof}
\noindent
\begin{itemize}
\item 
\textbf{Version I:} 
If $r = 0$, then the initial configuration is already a winning one for Spoiler. We have in this case, by definition, that $\overline{g}, \overline{h}$ receive different colorings at the initial round of count-free $k$-WL. Now let $r > 0$, and suppose that Spoiler wins at round $r$ of the pebble game. Suppose Spoiler begins their $r$-round winning  strategy by moving pebble $p_{j}$ from $g_{j}$ to $x$, and suppose that Duplicator responds by placing $p_{j}'$ on $y$. Thus, Spoiler has a winning strategy in the $(k+1)$-pebble, $(r-1)$-round game, starting from the configuration $( \overline{g}(g_{j}/x), \overline{h}(h_{j}/y))$. So by the inductive hypothesis, $\chi_{k,r-1}(\overline{g}(g_{j}/x)) \neq \chi_{k,r-1}(\overline{h}(h_{j}/y))$. As it was a winning strategy for Spoiler to move $p_j$ from $g_j$ to $x$ at the initial round, we thus have for any $y \in H$ that:
\[
(\chi_{k,r-1}(\overline{g}(g_{1}/x)), \ldots, \chi_{k,r-1}(\overline{g}(g_{k}/x))) \neq 
(\chi_{k,r-1}(\overline{h}(g_{1}/y)), \ldots, \chi_{k,r-1}(\overline{h}(h_{k}/y))).
\]
It follows that $\chi_{k,r}(\overline{g}) \neq \chi_{k,r}(\overline{h})$, as desired. The result now follows.

\item \textbf{Version II:} We modify the Version I argument above to use the Version II marked equivalence type. Otherwise, the argument is identical. \qedhere
\end{itemize}
\end{proof}

\subsection{Logics}

We recall the central aspects of first-order logic. We have a countable set of variables $\{x_{1}, x_{2}, \ldots, \}$. Formulas are defined inductively. As our basis, $x_{i} = x_{j}$ is a formula for all pairs of variables. Now if $\varphi$ is a formula, then so are the following: $\varphi \land \varphi, \varphi \vee \varphi, \neg{\varphi}, \exists{x_{i}} \, \varphi,$ and $\forall{x_{i}} \, \varphi$. Variables can be reused within nested quantifiers. In order to define logics on groups, it is necessary to define a relation that relates the group multiplication. We recall the two different logics introduced by Brachter \& Schweitzer \cite{WLGroups}.
\begin{itemize}
\item \textbf{Version I:} We add a ternary relation $R$ where $R(x_{i}, x_{j}, x_{\ell}) = 1$ if and only if $x_{i}x_{j} = x_{\ell}$ in the group. In keeping with the conventions of \cite{CFI}, we refer to the first-order logic with relation $R$ as $\mathcal{L}_{I}$ and its $k$-variable fragment as $\mathcal{L}_{I}^{k}$. We refer to the logic $\mathcal{C}_{I}$ as the logic obtained from $\mathcal{L}_{I}$ by adding counting quantifiers $\exists^{\geq n} x_{i} \, \varphi$ and $\exists{!n} \, \varphi$,  and its $k$-variable fragment as $\mathcal{C}_{I}^{k}$. 

\item \textbf{Version II:} We consider relations of the form $R_{w}(x_1, \ldots, x_t)$, where $w \in (\{x_{1}, \ldots, x_{t}\} \cup \{ x_{1}^{-1}, \ldots, x_{t}^{-1}\})^{*}$. We say that $R_{w}(x_1, \ldots, x_t)$ is fulfilled by a $t$-tuple $(g_1, \ldots, g_t)$ of group elements in $G$ if and only if $w(g_1, \ldots, g_t) = 1$ in $G$. Define $\mathcal{L}_{II}$ to be the extension of first-order logic obtained by adding all such relations $R_{w}$. Let $\mathcal{L}_{II}^{k}$ be the fragment of $\mathcal{L}_{II}$ that uses at most $k$ variables and relations $R_{w}$, where $w$ ranges through the $k-1$ variable words defined over these $k$ variables and their inverses. We refer to the logic $\mathcal{C}_{II}$ as the logic obtained from $\mathcal{L}_{II}$ by adding counting quantifiers $\exists^{\geq n} x_{i} \, \varphi$ and $\exists{!n} \, \varphi$,  and its $k$-variable fragment as $\mathcal{C}_{II}^{k}$. 
\end{itemize}

\begin{remark}
Brachter \& Schweitzer \cite{WLGroups} refer to $\mathcal{L}_{I}$ and $\mathcal{L}_{II}$ as the logics with counting quantifiers. We instead adhere to the conventions in \cite{CFI}.
\end{remark}

Brachter \& Schweitzer \cite[Lemma~3.6]{WLGroups} showed that for $J \in \{I, II\}$ two $k$-tuples $\overline{g}, \overline{h}$ receive a different initial color under $k$-WL Version $J$ if and only if there is a quantifier-free formula in $\mathcal{C}_{J}$ that distinguishes $\overline{g}, \overline{h}$. As such formulas do not use any quantifiers, $\overline{g}, \overline{h}$ receive a different initial color under $k$-WL Version $J$ if and only if there is a quantifier-free formula in $\mathcal{L}_{J}$ that distinguishes $\overline{g}, \overline{h}$. Now the equivalence between the $(k+1)$-pebble, $r$-round Version $J$ count-free pebble game and the $(k+1)$-variable, quantifier-depth $r$ fragment of $\mathcal{L}_{J}$ follows identically from the argument as in the case of graphs \cite{CFI}. We record this with the following theorem.

\begin{theorem}
Let $G$ and $H$ be groups of order $n$, and let $J \in \{I, II\}$. Fix $k \geq 2$, and let $\overline{g} \in G^k, \overline{h} \in H^k$. Let $\chi_{k,r}$ be the coloring computed by the count-free $(k,r)$-WL on $G$ and $H$. We have that $\chi_{k,r}(\overline{g}) \neq \chi_{k,r}(\overline{h})$ if and only if there exists a sentence formula $\varphi \in \mathcal{L}_{J}$ that uses at most $k$ free variables and quantifier depth at most $r$, such that $(G, \overline{g}) \models \varphi$ and $(H, \overline{h}) \not \models \varphi$.
\end{theorem}

\subsection{Count-free WL and Abelian groups}
We now turn to showing that the count-free WL Version II algorithm fails to yield a polynomial-time isomorphism test even for Abelian groups.

\input{abelian.tex}

\begin{remark}
\Thm{CountFreeAbelian} shows that count-free WL fails to serve as a polynomial-time (or even $|G|^{o(\log |G|)}$) isomorphism test for Abelian groups. As the $n$-pebble count-free WL game fails to distinguish $G_{n}$ and $H_{n}$, we also obtain an $\Omega(\log(|G_{n}|))$ lower bound on the quantifier rank of any $\textsf{FO}$ formula identifying $G_{n}$. In particular, this suggests that $\algprobm{GpI}$ is not in $\textsf{FO}(\poly \log \log n)$, even for Abelian groups. As $\textsf{FO}(\poly \log \log n)$ cannot compute $\algprobm{Parity}$ \cite{Smolensky87algebraicmethods}, this suggests that counting is necessary to solve \algprobm{GpI}. This is particularly interesting, as $\algprobm{Parity}$ is not $\textsf{AC}^{0}$-reducible to $\algprobm{GpI}$ \cite{ChattopadhyayToranWagner}.
\end{remark}

While count-free WL is unable to distinguish Abelian groups, the multiset of colors computed actually provides enough information to do so. That is, after count-free WL terminates, the color classes present and their multiplicities provide enough information to identify Abelian groups. The technical difficulty lies in parsing this information; we will show how to do so with a single $\textsf{Majority}$ gate and $O(\log n)$ non-deterministic bits. 

To this end, we first consider the order-finding problem. Barrington, Kadau, Lange, \& McKenzie \cite{BKLM} previously showed that order-finding is $\textsf{FOLL}$-computable. Our next result (\Prop{CountFreeOrder}) shows that the count-free Weisfeiler--Leman effectively implements this strategy. 

\begin{lemma} \label{LemmaCountFreeExponent}
Let $G,H$ be groups of order $n$. Suppose in the count-free Version I game, pebbles have already been placed on $g \mapsto h$ and $g^i \mapsto x$ with $x \neq h^i$. Then Spoiler can win with $O(1)$ additional pebbles in $O(\log \log i)$ rounds.
\end{lemma}

\begin{proof}
By induction on $i$. If $i = 0$, then $g^i = \text{id}_{G}$. In particular, $g^{i}$ is the unique element in $G$ such that for all $g' \in G$, $g^{i} \cdot g' = g'$. As $x \neq h^i$, we have that $x \neq \text{id}_{H}$. Spoiler can now pebble some $y' \in H$ such that $xy' \neq y'$. Let $y \in G$ be Duplicator's response. As $g^{i} = \text{id}_{G}$, we have that $g^{i}y = y$. So Spoiler wins. If $i = 1$, then there is one pebble pair mapping $g \mapsto h$ and another mapping $g \mapsto x \neq h$, so Spoiler wins immediately. 
So we now suppose $i > 1$, and that the result is true for all smaller exponents, say in $\leq c \log \log i'$ rounds for all $i' < i$.

The structure of the argument is as follows. If $i$ is not a power of $2$, we show how to cut the number of 1s in the binary expansion of $i$ by half using $O(1)$ rounds and only $O(1)$ pebbles that may be reused. Since the number of 1s in the binary expansion of $i$ is at most $\log_2 i$, and we cut this number in half each time, this takes only $O(\log \log i)$ rounds (and $O(1)$ pebbles) before $i$ has just one 1 in its binary expansion, that is, $i$ is a power of $2$. Once $i$ is a power of 2, we will show how to cut $\log_2 i$ in half using $O(1)$ rounds and $O(1)$ pebbles that may be reused. This takes only $O(\log \log i)$ rounds (and $O(1)$ pebbles) before getting down to the base case above. Concatenating these two strategies uses only $O(1)$ pebbles and $O(\log \log i)$ rounds. Now to the details.

If $i$ is not a power of $2$, we will show how cut the number of 1s in the binary expansion of $i$ in half. Write $i=j+k$ where $j,k$ each have at most half as many 1s in their binary expansion as $i$ does (rounded up). (Examine the binary expansion $i_\ell i_{\ell-1} \dotsb i_0$ and finding an index $z$ such that half the ones are on either side of $z$. Then let $j$ have binary expansion $i_\ell i_{\ell-1} \dotsb i_z 0 0 \dotsc 0$ and let $k$ have binary expansion $i_{z-1} i_{z-2} \dotsb i_0$.) At the next two rounds, Spoiler places pebbles on $g^j$ and $g^k$. Duplicator responds by pebbling $g^j \mapsto a$ and $g^k \mapsto b$. If $ab \neq x$, then Spoiler immediately wins. Thus we may now assume $ab=x$.

Since $x \neq h^i$, we necessarily have $\{a,b\} \neq \{h^j, h^k\}$. Without loss of generality, suppose $a \notin \{h^j, h^k\}$. Spoiler now picks up the pebble on $g^i$ and places it on $g^j$ instead (since $g^j$ is already pebbled, Spoiler could in fact reuse this pebble in the next round of the strategy, but as we don't need to be that efficient for our result, we just have Spoiler put it somewhere we know can't hurt). We are now in a situation where $g \mapsto h$ and $g^j \mapsto a \neq h^j$ are pebbled, and $j$ has at most half as many 1s in its binary expansion as $i$ did. (So there are only two pebbles that can't be re-used, which is precisely the number we started with.) The cost to get here was $O(1)$ rounds and $1$ additional pebble, which can be reused as the argument is iterated. 

After that has been iterated $\log \log i$ times, we come to the case where $i$ is a power of 2. We will show how to reduce to a case where $\log_2 i$ has been cut in half. Write $i=jk$ with $jk$ powers of $2$ such that $\log_2 j, \log_2 k \leq \lceil \frac{\log_2 i}{2}\rceil$ (if $i=2^z$, let $j=2^{\lceil z/2 \rceil}, k=2^{z-\lceil z/2 \rceil}$). Note that  we have $g^i = (g^j)^k$. Spoiler now pebbles $g^j$, and Duplicator responds by pebbling some $a$. If $a \neq h^j$, then Spoiler can re-use the pebble from $g^i$, and we now have $g \mapsto h, g^j \mapsto a \neq h^j$ pebbled with $\log_2 j \leq \lceil (1/2)\log_2 i \rceil$. This took $O(1)$ rounds and no non-reusable pebbles. On the other hand, if $a = h^j$, then we have $a^k = h^{jk} = h^i \neq x$. Spoiler may now reuse the pebble on $g \mapsto h$, and we are now in a situation where $g^j \mapsto a$ and $(g^j)^k \mapsto x \neq a^k$, just as we started, and with $\log_2 k \leq \lceil(1/2) \log_2 i \rceil$. As in the other case, this took $O(1)$ rounds and no non-reusable pebbles. This completes the proof.
\end{proof}

\begin{proposition}[Order finding in WL-I] \label{CountFreeOrder}
Let $G$ be a group. Let $g, h \in G$ such that $|g| \neq |h|$. The count-free $(O(1), O(\log \log n))$-WL Version I distinguishes $g$ and $h$.
\end{proposition}

\begin{proof}
We use the pebble game characterization, starting from the initial configuration $( (g), (h))$. We first note that if $g = 1$ and $h \neq 1$, that Spoiler wins, as $g \cdot g = g$, while $h \cdot h \neq h$. 
So now suppose that $g \neq 1$ and $h \neq 1$. 

Without loss of generality suppose $|g| < |h|$. Note that $g \mapsto h$ has already been pebbled by assumption. Spoiler now pebbles $1$. By the same argument as above, Duplicator must respond by pebbling $1$. But we have, for $i := |g|$, that $g^i = 1$ and by assumption $h^i \neq 1$. Thus, by \Lem{LemmaCountFreeExponent}, Spoiler can now win with $O(1)$ pebbles in $O(\log \log |g|) \leq O(\log \log n)$ rounds.
\end{proof}

As finite simple groups are uniquely identified amongst all groups by their order and the set of orders of their elements \cite{SimpleOrder}, we obtain the following immediate corollary.

\begin{corollary}
If $G$ is a finite simple group, then $G$ is identified by the count-free $(O(1), O(\log \log n))$-WL Version I. Consequently, isomorphism testing between a finite simple group $G$ and an arbitrary group $H$ is in $\textsf{FOLL}$.
\end{corollary}

We also obtain an improved upper bound on the parallel complexity of Abelian Group Isomorphism:

\begin{theorem} \label{ThmAbelian}
If $G$ is an Abelian group, and $H$ is an arbitrary group, then we can decide if $G \not \cong H$ in $\beta_{1}\textsf{MAC}^{0}(\textsf{FOLL})$.
\end{theorem}

Here, $\beta_1 \mathsf{MAC}^0(\mathsf{FOLL})$ denotes the class of languages decidable by a (uniform) family of circuits that have $O(\log n)$ nondeterministic input bits, are of depth $O(\log \log n)$, have gates of unbounded fan-in, and the only gate that is not an $\textsf{And}$, $\textsf{Or}$, or $\textsf{Not}$ gate is the output gate, which is a $\textsf{Majority}$ gate of unbounded fan-in. Note that, by simulating the $\poly(n)$ possibilities for the nondeterministic bits, $\beta_1 \mathsf{MAC}^0(\mathsf{FOLL})$ is contained in $\mathsf{TC}^0(\mathsf{FOLL})$, at the expense of using $\poly(n)$ \textsf{Majority} gates. Thus, our result improves on the prior upper bound of $\mathsf{TC}^0(\mathsf{FOLL})$ \cite{ChattopadhyayToranWagner}.

\Thm{ThmAbelian} is an example of the strategy of using \emph{count-free} WL, followed by a limited amount of counting afterwards. (We contrast this with the parallel implementation of the classical (counting) WL algorithm, which---for fixed $k$---uses a polynomial number of $\textsf{Majority}$ gates at each iteration \cite{GroheVerbitsky}.) After the fact, we realized this same bound could be achieved by existing techniques; we include both proofs to highlight an example of how WL was used in the discovery process.

\begin{proof}[Proof using Weisfeiler--Leman]
Let $G$ be Abelian, and let $H$ be an arbitrary group such that $G \not \cong H$. Suppose first that $H$ is not Abelian. We show that count-free $(O(1), O(1))$-WL Version I can distinguish $G$ from $H$. Spoiler uses two pebbles to identify a pair of elements $(x, y)$ in $H$ that do not commute. Duplicator responds by pebbling $(u, v) \in G$. At the next round, Spoiler pebbles $uv \in G$. Let $z \in H$ be Duplicator's response. As $x, y$ do not commute, we cannot have both $xy = z$ and $yx = z$. So Spoiler wins, having used $O(1)$ additional pebbles and $O(1)$ additional rounds. 

Suppose now that $H$ is Abelian. We run the count-free $(O(1), O(\log \log n))$-WL using the parallel WL implementation due to Grohe \& Verbitsky. As $G$ and $H$ are non-isomorphic Abelian groups, they have different order multisets. In particular, there exists an order class of greater multiplicity in $G$ than in $H$. By \Prop{CountFreeOrder}, two elements with different orders receive different colors. We use a $\beta_{1}\textsf{MAC}^{0}$ circuit to distinguish $G$ from $H$. Using $O(\log n)$ non-deterministic bits, we  guess the color class $C$ where the multiplicity differs. At each iteration, the parallel WL implementation due to Grohe \& Verbitsky records indicators as to whether two $k$-tuples receive the same color. As we have already run the count-free WL algorithm, we may in $\textsf{AC}^{0}$ decide whether two $k$-tuples have the same color. For each $k$-tuple of $G^{k}$ having color class $C$, we feed a $1$ to the $\textsf{Majority}$ gate. For each $k$-tuple of $H^k$ having color class $C$, we feed a $0$ to the $\textsf{Majority}$ gate. The $\textsf{Majority}$ gate outputs a $1$ if and only if there are strictly more $1$'s than $0$'s. The result now follows.
\end{proof}

\begin{proof}[Alternative proof using prior techniques, that we only realized after discovering the WL proof]
This proof follows the strategy of Chattopadhyay, Tor\'an, \& Wagner \cite{ChattopadhyayToranWagner}, realizing that their use of many threshold gates can be replaced by $O(\log n)$ nondeterministic bits and a single threshold gate. 

Compute the multiset of orders in $\mathsf{FOLL}$ \cite[Prop.~3.1]{BKLM}, guess the order $k$ such that $G$ has more elements of order $k$ than $H$ does. Use a single $\textsf{Majority}$ gate to compare those counts.
\end{proof}

\section{Conclusion}

We combined the parallel WL implementation of Grohe \& Verbitsky \cite{GroheVerbitsky} with the WL for groups algorithms due to Brachter \& Schweitzer \cite{WLGroups} to obtain an efficient parallel canonization procedure for several families of groups, including: (i) coprime extensions $H \ltimes N$ where $N$ is Abelian and $H$ is $O(1)$-generated, and (ii) direct products, where WL can efficiently identify the indecomposable direct factors. 

We also showed that the individualize-and-refine paradigm allows us to list all isomorphisms of semisimple groups with an $\textsf{SAC}$ circuit of depth $O(\log n)$ and size $n^{O(\log \log n)}$. Prior to our paper, no parallel bound was known. And in light of the fact that multiplying permutations is $\textsf{FL}$-complete \cite{COOK1987385}, it is not clear that the techniques of Babai, Luks, \& Seress \cite{BabaiLuksSeress} can yield circuit depth $o(\log^{2} n)$.

Finally, we showed that $\Omega(\log(n))$-dimensional count-free WL is required to identify Abelian groups. It follows that count-free WL fails to serve as a polynomial-time isomorphism test even for Abelian groups. Nonetheless, count-free WL distinguishes group elements of different orders. We leveraged this fact to obtain a new $\beta_{1}\textsf{MAC}^{0}(\textsf{FOLL})$ upper bound on isomorphism testing of Abelian groups.

Our work leaves several directions for further research that we believe are approachable and interesting.

\begin{question} \label{Question1}
Show that coprime extensions of the form $H \ltimes N$ with both $H,N$ Abelian have constant WL-dimension (the WL analogue of \cite{QST11}). More generally, a WL analogue of Babai--Qiao \cite{BQ} would be to show that when $|H|,|N|$ are coprime and $N$ is Abelian, the WL dimension of $H \ltimes N$ is no more than that of $H$ (or the maximum of that of $H$ and a constant independent of $N,H$). 
\end{question}

\begin{question}
Is the WL dimension of semisimple groups bounded? 
\end{question}

It would be of interest to address this question even in the non-permuting case when $G = \text{PKer}(G)$. Subsequent to initial versions of this paper, Brachter showed \cite{BrachterThesis} that the WL dimension of semisimple groups is $O(\log \log n)$, showing the ``WL analogue of \cite{BCGQ}'', without the need for individualization and refinement.
In a higher-arity version of WL, which is not known to admit efficient algorithms, the analogous question for semisimple groups has a positive answer \cite{GLDescriptiveComplexity}.

It is often the case that if an uncolored class of graphs is identified by WL, then so is the corresponding class of colored graphs. So if constant-dimensional WL identifies a class of graphs, it can often readily be extended to an efficient canonization procedure (cf., \cite{grohe2019canonisation}). In the case of groups, it is not clear whether WL easily identifies colored variants, which were studied in \cite{BrachterSchweitzerWLLibrary, GR16}. To this end, we ask the following. A \emph{colored group} is a group $G$ together with a coloring $\chi \colon G \to \mathcal{K}$ for some set of colors $\mathcal{K}$. Two colored groups are isomorphic if there is an isomorphism between them that preserves the colors of all elements. WL can naturally be started with such an initial coloring, which is then refined according to the initial WL coloring as defined above.

\begin{question}
Does constant-dimensional Weisfeiler--Leman identify every colored Abelian group?
\end{question}

As pointed out to us by an anonymous reviewer, a positive answer to this question would resolve, in a strong way, a question of Babai \cite[\S 13.2 of the arXiv version]{BabaiGraphIso}: namely, deciding isomorphism of colored elementary Abelian $p$-groups is the same as the \algprobm{String Isomorphism} problem on $n=p^k$ points under the natural action of the general linear group $GL(k,p)$. Babai (\emph{ibid.}) asked whether the latter problem could be solved in $p^{o(k^2)}$ time, which would follow if $o(\log n)$-dimensional WL identified such colored groups. Yet even for colored groups, we do not currently know a lower bound on the (counting) WL dimension.


For the classes of groups we have studied, when we have been able to give an $O(1)$ bound on their WL-dimension, we also get an $O(\log n)$ bound on the number of rounds needed. The dimension bound alone puts the problem into $\cc{P}$, while the bound on rounds puts it into $\cc{TC}^{1}$. A priori, these two should be distinct. For example, in the case of graphs, Kiefer \& McKay \cite{KieferMcKay} have shown that there are graphs for which color refinement takes $n-1$ rounds to stabilize. There has been also considerable work on lower bounds against the iteration number for $k$-WL for graphs, when $k > 2$ \cite{FurerRefinement, BerkholzNordstrom, GroheLichterNeuen}.

\begin{question}
Is there a family of groups $\mathcal{G}$ and a $k \geq 2$, such that each group $G \in \mathcal{G}$ is identified by $k$-WL but requiring $\omega(\log n)$ rounds? 
\end{question}

We also wish to highlight a question that essentially goes back to \cite{ChattopadhyayToranWagner}, who showed that \algprobm{GpI} cannot be hard under $\cc{AC}^0$ reductions for any class containing \algprobm{Parity}. In Theorem \ref{CountFreeAbelian}, we showed that count-free WL requires dimension $\Omega(\log n)$ (and hence, $\Theta(\log n)$) to identify even Abelian groups. This shows that this particular, natural method does not put \algprobm{GpI} into $\cc{FO}(\poly \log \log n)$, though it does not actually prove $\algprobm{GpI} \notin \cc{FO}(\poly \log \log n)$, since we cannot rule out clever bit manipulations of the Cayley (multiplication) tables. While we think the latter lower bound would  be of significant interest, we think even the following question is interesting:

\begin{question}
Show that $\algprobm{GpI}$ does not belong to (uniform) $\cc{AC}^0$.
\end{question}

\bibliographystyle{alphaurl}
\bibliography{references}

\end{document}

%% file: abelian.tex

For our lower bound we will use groups of the form $G = (\Z/2\Z)^n \times (\Z/4\Z)^m$. We begin by recalling the classification of subgroups of a finitely generated Abelian group \cite{birkhoff}, as applied to groups of this form. 

A \emph{basis} of an Abelian group $A$ is a generating set $a_1,\dotsc,a_k$ of $A$ such that $A = \prod_{i=1}^k \langle a_i \rangle$, or equivalently, $\langle a_i \rangle \cap \langle a_j : j \neq i \rangle = 1$ for all $i \in [k]$. For finite Abelian $p$-groups $A$, all bases have the same size, which is equal to the dimension of $A/A^p$ as an $\F_p$-vector space. (If $A$ is not a $p$-group this is false, e.g. $\Z/6\Z$ a basis of size 1 and a basis of size 2.) This is also the minimum number of generators of $A$, which is also denoted $d(A)$. If $A$ is an Abelian $p$-group and $x \in A$ is not the $p$-th power of another element of $A$, then there is a subgroup $A' \leq A$ such that $A = A' \times \langle x \rangle$, or equivalently there is a basis of $A$ that includes $x$ (this follows from Nakayama's Lemma, by the analogous result for bases of the $\F_p$-vector space $A / A^p$).

For Abelian groups $G = (\Z/2\Z)^n \times (\Z/4\Z)^m$, any subgroup $A \leq G$ has the following form: there is a basis $g_1,\dotsc,g_{n+m}$ of $G$ such that there exists a $k \leq n+m$, and $0 \leq a,b,c, \leq k$ such that a basis for $A$ is given by $\{g_1,\dotsc,g_a\} \cup \{g_{a+1}, \dotsc, g_{a+b}\} \cup \{g_{a+b+1}^2, \dotsc, g_{a+b+c}^2\}$, where $g_1,\dotsc,g_a$ have order 2, and $g_{a+1},\dotsc,g_{a+b},g_{a+b+1},\dotsc,g_{a+b+c}$ have order 4. (One way to see this is to write the basis of $G$ as the columns of a matrix, generators of $A$ as the rows of the matrix, and then apply Smith Normal Form. Since all the entries are taken modulo 4, the elementary divisors are all either $1 \pmod{2}$, in columns corresponding to $g_i$ of order 2, or $1$ or $2\pmod{4}$ in columns corresponding to $g_i$ of order 4.) 

We refer to $(a,b,c)$ as the \emph{type} of the subgroup $A$. An alternative characterization is: $a$ is the number of direct factors of $A$ that are $\Z/2\Z$ that do not lie in a copy of $\Z/4\Z$ in $G$; $b$ is the number of direct factors of $A$ that are $\Z/4\Z$; and $c$ is the number of direct factors of $A$ that are $\Z/2\Z$ that are subgroups of a $\Z/4\Z$ in $G$. It is clear that the action of $\Aut(G)$ preserves the type of a subgroup. The converse is also true: by the classification of automorphisms of Abelian groups (see, e.g., \cite{hillarRhea} for a nice exposition) or by using Smith Normal Form as mentioned above, it follows that for any two subgroups $A,B \leq G$ of the same type, there is an automorphism of $G$ that sends $A$ to $B$.

\begin{lemma} \label{lem:types}
Let $G = (\Z/2\Z)^n \times (\Z/4\Z)^m$, and $A \leq G$ a subgroup of type $(a,b,c)$ as defined above. Then $A$ is contained in subgroups of each of the following types:
\begin{enumerate}
\item $(a,b+1,c-1)$ if $c > 0$;
\item $(a+1,b,c)$ if $a < n$; 
\item $(a,b+1,c)$ if $b+c < m$; 
\item $(a,b,c+1)$ if $b+c < m$.
\end{enumerate}
Furthermore, if $A \leq B$ and $d(B) = d(A)+1$, then the type of $B$ is one of the latter three types.
\end{lemma}

\begin{proof}
Let $\beta = \{x_1,\dotsc,x_{a+b+c}\}$ be a basis of $A$ such that $\{x_1,\dotsc,x_a\}$ generate a subgroup of non-squares-in-$G$ that is isomorphic to $(\Z/2\Z)^a$, $\{x_{a+1},\dotsc,x_{a+b}\}$ is a basis for a subgroup isomorphic to $(\Z/4\Z)^b$, and $\{x_{a+b+1},\dotsc,x_{a+b+c}\}$ is a basis for a subgroup of square-in-$G$ that is isomorphic to $(\Z/2\Z)^c$. Let $Y$ be a copy of $(\Z/4\Z)^m$ such that $\{x_{a+1},\dotsc,x_{a+b+c}\} \subseteq Y$, and let $X$ be a copy of $(\Z/2\Z)^n \leq G$ such that $\{x_1,\dotsc,x_a\} \subseteq X$, $X \cap Y = 1$, and $G = \langle X, Y \rangle = X \times Y$.

\begin{enumerate}
\item If $c > 0$, let $x = x_{a+b+1}$; then $x^2 = 1$ and $x$ is a square in $G$, say $y^2 = x$. If we replace $x_{a+b+1}$ in $\beta$ with $y$, we get a basis for a subgroup that contains $A$ (since $y^2 = x_{a+b+1}$), and we have changed the type to $(a,b+1,c-1)$.

\item If $a < n$, then $\langle x_1,\dotsc, x_a \rangle$ is a proper subgroup of $(\Z/2\Z)^n$. Since subgroups of $(\Z/2\Z)^n$ are the same thing as subspaces of $(\Z/2\Z)^n$ as a $\Z/2\Z$-vector space, in this case the result follows from the usual result for vector spaces.

\item For $i=1,\dotsc,c$, let $y_i$ be a square root of $x_{a+b+i}$, and let 
\[
A_Y := \langle x_{a+1},\dotsc, x_{a+b}, y_1,\dotsc,y_c \rangle \leq Y.
\]
Now consider the surjective homomorphism $\pi\colon Y \to (\Z/2\Z)^m$ with kernel $Z$ (i.e., take each coordinate modulo 2). If $b + c < m$, then $\pi(A_Y)$ is not all of $(\Z/2\Z)^m$, so be the result for vector spaces, there exists $\overline{y} \in (\Z/2\Z)^m$ such that $\overline{y}$ is not in the image $\pi(A_Y)$. Let $y \in Y$ be a preimage of $\overline{y}$, i.e. so that $\pi(y) = \overline{y}$. We claim that $\langle y \rangle \cap A_Y = 1$. For we have $\langle \overline{y} \rangle \cap \pi(A_Y) = 1$, and therefore $\langle y \rangle \cap A_Y \leq \ker(\pi) = Z$. Thus at most we have $\{1,y^2\} \in \langle y \rangle \cap A_Y$. Suppose for the sake of contradiction that $y^2$ is in $A_Y$. Since $A_Y \cong (\Z/4\Z)^{b+c}$, this means there is a $z \in A_Y$ of order 4 such that $z^2 = y^2$. But then $z^{-1} y$ has order 2, so $z$ and $y$ differ by an element of $Z = \ker(\pi)$, giving $\pi(z) = \pi(y)$. But this contradicts our construction that $\pi(y) \notin \pi(A_Y)$. Thus $\langle y \rangle \cap A_Y = 1$, as claimed. Then $B = A \times \langle y \rangle$ is a subgroup containing $A$ and of type $(a,b+1,c)$.

\item As in the previous part, we can find a $y \in Y$ such that $\langle Y \rangle \cap A = 1$, but now we let $B = A \times \langle y^2 \rangle$, which has type $(a,b,c+1)$. 
\end{enumerate}

To see the furthermore, suppose that $A \leq B$ and $d(B) = d(A) + 1$. Let $y \in G$ be such that $B = A \cdot \langle y \rangle$. If $y$ is an element of order 4, then we must have $\langle y \rangle \cap A = 1$, for otherwise $y^2 \in A$, there is a basis of $A$ that includes $y^2$, and then a basis for $B$ can be gotten from that basis for $A$ by replacing $y^2$ by $y$, giving $d(B) = d(A)$, contradicting our assumption. Thus if $y$ has order 4 we are necessarily in case 3.

If $y$ has order 2, then we have $\langle y \rangle = \{1, y \}$, and therefore $\langle y \rangle \cap A = 1$ (since $d(B) > d(A)$, so $B$ strictly contains $A$), so we again have $B = A \times \langle y \rangle$. If $y$ is not a square in $G$, then we are in case 2, and if $y$ is a square in $G$ then we are in case 4. This completes the proof of the lemma.
\end{proof}

Next, note that if $A \leq G = (\Z/2\Z)^n \times (\Z/4\Z)^m$ has type $(a,b,c)$, then the set of elements in $A$ that are of order at most 2 and are also squares in $G$ is a subgroup of $A$ of type $(0,0,c)$: such elements are their own inverse, and if $x,y$ are two such elements with $a,b \in G$ such that $a^2 = x$ and $b^2= y$, then since $G$ is Abelian we have $(ab)^2 = xy$, and thus $xy$ is also of order (at most) 2 and has a square root in $G$. Let us call this the subgroup of $G$-squares in $A$, and denote it $sq_G(A)$. 

Below we will use the following observation without further mention. Suppose $G,H$ are finite groups $A \leq G$, $B \leq H$, $\varphi\colon A \to B$ is an isomorphism,  $x \in G$, $y \in H$, $\langle x \rangle \cap A = 1$, $x$ normalizes $A$ and $A$ normalizes $\langle x \rangle$, and similarly for $y$ and $B$, and $x$ and $y$ have the same order. Then there is a unique isomorphism $\hat{\varphi} \colon A \times \langle x \rangle \to B \times \langle y \rangle$ that agrees with $\varphi$ when restricted to $A$, and such that $\hat{\varphi}(x) = y$. 

\begin{theorem} \label{CountFreeAbelian}
For $n \geq 2$, let $G_n := (\Z/2\Z)^{n-2} \times (\Z/4\Z)^{n+1}$ and $H_n := (\Z/2\Z)^{n} \times (\Z/4\Z)^{n}$. The $n$-pebble count-free Version II game does not distinguish $G_n$ from $H_n$. 
\end{theorem}

We note that for these pairs of groups, this result is tight, as with $(n+1)$ pebbles, Spoiler can pebble generators of $(\Z/4\Z)^{n+1}$ in $G_n$, and there is no subgroup of $H_n$ isomorphic to $(\Z/4\Z)^{n+1}$.

\begin{proof}
We may assume that the game starts with all $n$ pebbles on the identity of both groups, so that for simplicity we can assume that all $n$ pebbles are already placed. For the purposes of this proof, we define a \emph{good pebbling} as follows. If $A$ is the subgroup generated by the pebbled elements in $G_n$ and $B$ is the subgroup generated by the pebbled elements in $H_n$, an isomorphism $\varphi \colon A \to B$ is \emph{good} if $\varphi(sq_G(A)) = sq_H(B)$. We then call the pebbling \emph{good} if the pebbled map extends to a marked isomorphism that is also a good isomorphism.

Next, we claim that if there is a good isomorphism from $A \leq G$ to $B \leq H$, then $A$ and $B$ have the same type. Let $(a,b,c)$ be the type of $A$ and $(a',b',c')$ be the type of $B$. Since $A \cong (\Z/2\Z)^{a+c} \times (\Z/4\Z)^b$ and $A \cong B$, we have $a' + c' = a + c$ and $b' = b$. Directly from the definition of good, it follows that $c = c'$, and thus by the preceding equality, we also get $a = a'$. Thus the existence of a good isomorphism $A \to B$ implies $A$ and $B$ have the same type.

We will show by induction on the rounds of the game that Duplicator can always ensure the pebbled map is good, and therefore never lose.

The initial pebbled map (all pebbles on the identity in both groups) is vacuously good.

Suppose inductively that the pebbled map $g_i \mapsto h_i$ ($i=1,\dotsc,n$) is good. Spoiler picks up a pebble---for simplicity of notation, say the first pebble---and places it on the element $g'_1$. If $g'_1 = g_1$, then Duplicator may respond by placing the other pebble on $h_1$. If $g'_1 = 1$ then Duplicator responds by placing the pebble on $1$ as well. Otherwise, we break into three cases based on $\langle g'_1 \rangle \cap \langle g_2, \dotsc, g_{n} \rangle$. Let $A_1 = \langle g_2, \dotsc, g_{n} \rangle$ and $B_1 = \langle h_2, \dotsc, h_{n} \rangle$. 

\textbf{Case 1: $\langle g'_1 \rangle \cap A_1$ is all of $\langle g'_1 \rangle$.} Equivalently, $g'_1$ is in $A_1$. By assumption, there exists a good isomorphism $\varphi \colon A_1 \to B_1$ that extends the pebbled map, that is, such that $\varphi(g_i) = h_i$ for all $i=1,\dotsc,n$. Thus $\varphi$ restricted to $A_1$ is also a good isomorphism that sends $g_i \mapsto h_i$ for all $i=2,3,\dotsc,n$. Since $g'_1$ is in $A_1$, Duplicator responds with $\varphi(g'_1)$. The newly pebbled map is still good, for $\varphi$ is a good isomorphism that extends the pebbled map.

\textbf{Case 2: $\langle g'_1 \rangle \cap A_1$ is neither trivial nor all of $\langle g'_1 \rangle$.}  Then $g'_1$ must have order 4, and the intersection must be $\{1, (g'_1)^2 \} = \langle (g'_1)^2 \rangle$. Let $\varphi \colon A_1 \to B_1$ be a good isomorphism that extends the pebbled map $g_i \mapsto h_i$ for $i=2,3,\dotsc,n$, and let $(a,b,c)$ be the type of $A_1$ (hence also the type of $B_1$). Let $y = \varphi((g'_1)^2)$. Then $(g'_1)^2$ is a square in $G$; it may either be a square of an element of $A_1$ or not, so we split further into those two cases:
\begin{itemize}
\item[2a] There exists $z \in A_1$ such that $z^2 = (g'_1)^2$. Then $z$ necessarily has order $4$, $\langle A_1, g'_1 \rangle = \langle A_1, z^{-1} g'_1 \rangle$. We have that $z^{-1} g'_1$ has order 2 since $z^2 = (g'_1)^2$, and $z^{-1} g'_1$ cannot be in $A_1$ since $z$ is in $A_1$ while $g'_1$ is not. Thus $z^{-1} g'_1$ is part of a basis for $\langle A_1, g'_1 \rangle$.

Now, If $z^{-1} g'_1$ is a square in $G$, then the type of $\langle A_1, g'_1 \rangle$ is $(a,b,c+1)$. Since $\langle A_1, g'_1 \rangle$ is generated by at most the $n$ pebbled elements, we have $b+c+1 \leq n$, and therefore $b+c < n$. Thus by Lemma~\ref{lem:types} there exists $w \in H$ of order 4 such that $\langle w \rangle \cap B_1 = 1$. We will select Duplicator's response so that $\varphi$ extends to a good isomorphism that maps $z^{-1} g'_1$ to $w^2$. Since $z$ is in $A_1$, we have $\varphi(z) \in B_1$, and then Duplicator we will respond by placing the pebble on $h'_1 := \varphi(z) w^2$. Thus $\varphi(z)^{-1} h'_1 = w^2$ is part of a basis of $\langle B_1, h'_1 \rangle$, and thus $\langle B_1, h'_1 \rangle$ has type $(a,b,c+1)$, and there is a unique extension of $\varphi$ to $\langle A_1, g'_1 \rangle$ defined by mapping $g'_1 \mapsto h'_1$, and this extension is a good isomorphism, so the new pebbling is good.

On the other hand, if $z^{-1} g'_1$ is not a square in $G$, then the type of $\langle A_1, g'_1 \rangle$ is $(a+1,b,c)$. Since this subgroup is generated by at most $n-1$ elements, we have $a+1 + b + c \leq n$, and in particular $a+1 \leq n$. Since $B_1$ has type $(a,b,c)$ and $a < n$, by Lemma~\ref{lem:types} there is an element of $w$ of order 2 in $H$ that is not a square in $H$ such that $w \notin B_1$. We will select Duplicator's response so that $\varphi$ extends to a good isomorphism that maps $z^{-1} g'_1$ to $w$. Namely, Duplicator responds with $h'_1 := \varphi(z) w$. By a similar argument to the preceding paragraph, $\varphi$ extends uniquely to a good isomorphism that maps $g'_1 \mapsto h'_1$, so the new pebbling is good.

\item[2b] $(g'_1)^2$ is not the square of any element of $A_1$. Since $\varphi$ is an isomorphism from $A_1$ to $B_1$, $\varphi((g'_1)^2)$ is not the square of any element of $B_1$. And since $\varphi$ is good and $(g'_1)^2$ is in $sq_G(A)$, $\varphi((g'_1)^2)$ must be in $sq_H(B)$. Thus there is an element $h'_1 \in H$ of order 4 such that $(h'_1)^2 = \varphi((g'_1)^2)$, and since no such element is in $B_1$, we also have $h'_1 \notin B_1$. Duplicator responds with $h'_1$. By construction, $\varphi$ extends uniquely to a good isomorphism that maps $g'_1 \mapsto h'_1$, and therefore the new pebbling is good.
\end{itemize}

\textbf{Case 3: $\langle g'_1 \rangle \cap A_1= 1$.} In this case, $g'_1$ is part of a basis for $\langle A_1, g'_1 \rangle$, thus the type of the latter is either $(a+1,b,c)$ if the order of $g'_1$ is 2, or $(a,b+1,c)$ if the order of $g'_1$ is 4. In the former case, we have $a+1 \leq n-2$ since $G$ only contains $(\Z/2\Z)^{n-2}$ as a direct factor, but not any larger power of $(\Z/2\Z)$, that is, $G_n$ does not contain subgroups of type $(n-1,*,*)$. (And in any case we must have $a + 1 < n$ since there are only $n$ pebbles.) Thus $a < n$, and by Lemma~\ref{lem:types} there is an element $h'_1 \in H$ of order 2 that is not a square in $H$ and is not in $B_1$, and Duplicator may respond with $h'_1$. As before, since $\langle h'_1 \rangle \cap B_1 = 1$, there is a unique extension of $\varphi$ to $\langle A_1, g'_1 \rangle$ that sends $g'_1 \mapsto h'_1$ which is also good, and thus the new pebbling is good.

On the other hand, if the order of $g'_1$ is 4 and the type of $\langle A_1, g'_1 \rangle$ is $(a,b+1,c)$, then we have $b+1+c \leq n$ because there are only $n$ pebbles, and thus $b+c < n$, so by Lemma~\ref{lem:types} there is an element $h'_1 \in H$ of order 4 with $\langle h'_1 \rangle \cap B_1 = 1$, and Duplicator responds with $h'_1$. As in the other cases, there is a unique good extension of $\varphi$ that sends $g'_1 \mapsto h'_1$, and thus the new pebbling is good.

In all cases, Duplicator can respond in a way that results in a good pebbling. Thus, by induction, Duplicator can play forever and Spoiler cannot win.
\end{proof}

\begin{remark}
A crucial fact that is being used ``behind the scenes'' in the above proof is that the automorphism group of $G$ acts transitively on subgroups of a given type, and therefore the type of a subgroup uniquely determines the existence of elements outside of the subgroup with various relations to the subgroup (encapsulated in Lemma~\ref{lem:types}). Without this, a much more complicated strategy would be needed. We can also see where this proof breaks down for counting WL, namely: while subgroups of any needed type \emph{exist} in the argument, the number of such subgroups will be different in $G$ and $H$.
\end{remark}

